\documentclass{article}

    \PassOptionsToPackage{numbers, sort&compress}{natbib}
\usepackage[preprint]{neurips_2025}

\usepackage[utf8]{inputenc} % allow utf-8 input
\usepackage[T1]{fontenc}    % use 8-bit T1 fonts
\usepackage{hyperref}       % hyperlinks
\usepackage{url}            % simple URL typesetting
\usepackage{booktabs}       % professional-quality tables
\usepackage{amsfonts}       % blackboard math symbols
\usepackage{nicefrac}       % compact symbols for 1/2, etc.
\usepackage{microtype}      % microtypography
\usepackage{xcolor}         % colors

\newcommand{\C}{\mathbbm{C}}

\newcommand{\N}{\mathbbm{N}}

\newcommand{\cT}{\mathcal{T}}

\newcommand{\supp}{\operatorname{supp}}

\usepackage{amsmath}
\usepackage{amssymb}
\usepackage{amsthm}
\usepackage{bbm}
\usepackage{graphicx}
\usepackage{algorithm2e}
\usepackage{mathrsfs}
\usepackage{braket}
\usepackage{doi}
\usepackage{tikz}
\usetikzlibrary{quantikz2}
\usepackage{tcolorbox}
\xdefinecolor{bluebox}{RGB}{0,101,189}
\newtcolorbox{mybox}{colback=bluebox!5!white,colframe=bluebox!75!black}
\usepackage{subcaption}

\usepackage{theoremref}

\newtheorem{thm}{\protect\theoremname}[section]
\theoremstyle{plain}
\newtheorem{lem}[thm]{\protect\lemmaname}
\theoremstyle{plain}

\theoremstyle{plain}

\theoremstyle{plain}
\newtheorem*{lem*}{\protect\lemmaname}
\theoremstyle{plain}
\newtheorem*{rem*}{\protect\remarkname}
\theoremstyle{plain}
\newtheorem*{thm*}{\protect\theoremname}
\theoremstyle{plain}
\newtheorem{prop}[thm]{\protect\propositionname}
\theoremstyle{plain}
\newtheorem*{prop*}{\protect\propositionname}
\theoremstyle{plain}
\newtheorem*{alg*}{Algorithm}
\theoremstyle{plain}
\newtheorem{cor}[thm]{\protect\corollaryname}

\newtheorem{defn}[thm]{Definition}

\usepackage[USenglish]{babel}
  \providecommand{\corollaryname}{Corollary}
  \providecommand{\lemmaname}{Lemma}
  \providecommand{\propositionname}{Proposition}
  \providecommand{\remarkname}{Remark}
\providecommand{\theoremname}{Theorem}

\usepackage[normalem]{ulem}

\newcommand{\cH}{\mathcal{H}}
\newcommand{\cK}{\mathcal{K}}

\newcommand{\cL}{\mathcal{L}}
\newcommand{\cHproj}{\mathcal{H}_{\mathrm{proj}}}

\newcommand{\cKproj}{\mathcal{K}_{\mathrm{proj}}}

\newcommand{\Or}{\mathcal{O}}

\newcommand{\tr}{\mathrm{tr}}
\newcommand{\dd}{\mathrm{d}}

\DeclarePairedDelimiter{\norm}{\Vert}{\Vert}

\renewcommand{\Re}{\operatorname{Re}}
\renewcommand{\Im}{\operatorname{Im}}

\title{Heisenberg-limited Hamiltonian learning continuous variable systems via engineered dissipation}

\author{%
  Tim Möbus \\
  Department of Mathematics \\
  University of T\"ubingen\\
  T\"ubingen, 72074, Germany \\
  \texttt{tim.moebus@uni-tuebingen.de} \\
   \And
  Andreas Bluhm \\
  Univ.\ Grenoble Alpes, CNRS, Grenoble INP, LIG\\
  38401 Saint Martin d'Hères, France \\
  \texttt{andreas.bluhm@univ-grenoble-alpes.fr} \\
   \AND
  Tuvia Gefen \\
  Racah Institute of Physics\\
  The Hebrew University of Jerusalem \\
  Jerusalem 91904, Givat Ram, Israel\\
  \texttt{tuvia.gefen@mail.huji.ac.il} \\
  \And
  Yu Tong \\
  Department of Mathematics,\\
  Department of Electrical and Computer Engineering,\\
  and Duke Quantum Center\\
  Duke University\\
  Durham, NC 27708, USA\\
  \texttt{yu.tong@duke.edu} \\
  \And
  Albert H. Werner \\
  QMATH, Department of Mathematical Sciences\\
  University of Copenhagen\\
  2100 Copenhagen, Denmark\\
  \texttt{Werner@math.ku.dk} \\
  \And
  Cambyse Rouz{\'e}\\
  Inria Saclay, Télécom Paris - LTCI, \\
  Institut Polytechnique de Paris\\
  91120 Palaiseau, France\\
  \texttt{cambyse.rouze@inria.fr}\\
}

\mathtoolsset{showonlyrefs}

\begin{document}

\maketitle

\begin{abstract}
    \vspace*{-1ex}Discrete and continuous variables oftentimes require different treatments in many learning tasks. Identifying the Hamiltonian governing the evolution of a quantum system is a fundamental task in quantum learning theory. While previous works mostly focused on quantum spin systems, where quantum states can be seen as superpositions of discrete bit-strings, relatively little is known about Hamiltonian learning for continuous-variable quantum systems.
    In this work we focus on learning the Hamiltonian of a bosonic quantum system, a common type of continuous-variable quantum system. This learning task involves an infinite-dimensional Hilbert space and unbounded operators, making mathematically rigorous treatments challenging. We introduce an analytic framework to study the effects of strong dissipation in such systems, enabling a rigorous analysis of cat qubit stabilization via engineered dissipation. This framework also supports the development of Heisenberg-limited algorithms for learning general bosonic Hamiltonians with higher-order terms of the creation and annihilation operators. Notably, our scheme requires a total Hamiltonian evolution time that scales only logarithmically with the number of modes and inversely with the precision of the reconstructed coefficients. On a theoretical level, we derive a new quantitative adiabatic approximation estimate for general Lindbladian evolutions with unbounded generators. Finally, we discuss possible experimental implementations.
\end{abstract}

\newpage

\section{Introduction}

Hamiltonian learning is the task of identifying the unknown Hamiltonian governing the evolution of a quantum system. Besides being a natural problem fundamental to our understanding of quantum systems, it is useful for tasks ranging from quantum simulation to quantum error correction, where precise knowledge of the Hamiltonian allows us to better control the quantum systems and mitigate errors. 
From a learning theory perspective, Hamiltonian learning is deeply connected to many classical tasks such as learning Markov random fields \cite{GaitondeMoitraMossel2024bypassing}.

While Hamiltonian learning has been studied in a large number of works \cite{Seif2021compressed,evans2019scalablebayesianhamiltonianlearning,li2020hamiltonian,che2021learning,HaahKothariTang2022optimal,yu2023robust,hangleiter2024robustlylearninghamiltoniandynamics,StilckFrança2024,ZubidaYitzhakiEtAl2021optimal,BaireyAradEtAl2019learning, bairey2020learning,GranadeFerrieWiebeCory2012robust,gu2022practical,wilde2022learnH,KrastanovZhouEtAl2019stochastic,Caro_2024,mobus2023dissipation,HolzapfelEtAl2015scalable, HuangTongFangSu2023learning,dutkiewicz2023advantage,MiraniHayden2024learning,NiLiYing2024quantum,LiTongNiGefenYing2023heisenberg,BoixoSomma2008parameter,bakshi2024structure,WangLi2024simulation,odake2023universal,ZhangLinNarangLuo2025hamiltonian,MaFlammiaPreskillTong2024learning,HuMaGongEtAl2025ansatz}, most of the works focus on quantum spin systems where the local Hilbert space dimension is finite and the Hamiltonian is a bounded operator. These assumptions are often physically acceptable and make the analysis much easier. However, many quantum systems involve infinite-dimensional Hilbert spaces and unbounded operators in the Hamiltonians. Examples include superconducting circuits \cite{krantz2019quantum, clerk2020hybrid, blais2021circuit}, integrated photonic circuits \cite{wang2020integrated} and optomechanical platforms \cite{aspelmeyer2014cavity, metcalfe2014applications}. Previous works that study Hamiltonian learning for bosonic systems \cite{hangleiter2024robustlylearninghamiltoniandynamics,LiTongNiGefenYing2023heisenberg,mobus2023dissipation,ZhangLinNarangLuo2025hamiltonian} typically restrict to special models, e.g., generalized Bose-Hubbard models, or do not achieve the Heisenberg-limited scaling, which is the fundamental bound of quantum metrology that can be saturated in the spin Hamiltonian learning setting.

Studying a more general form of bosonic Hamiltonians presents many challenges. Notably, the squeezing terms, which are quadratic in the bosonic creation and annihilation operators and are commonly seen in the quantum metrology setting, can cause the particle number to grow exponentially in time, leading to the breakdown of the Lieb-Robinson bound that governs the information propagation speed in the quantum system \cite{eisert2009supersonic}. Terms that are higher-order polynomials of the creation and annihilation operators present even greater challenges.

In this work we use engineered dissipation to regularize the dynamics to partially avoid the problems caused by general bosonic Hamiltonians. Engineered dissipation also allows us to reshape the Hamiltonian following \cite{HuangTongFangSu2023learning} into a form that makes it easy for us to extract coefficients from it. This forms the basis for our learning protocol, which allows us to learn an $m$-mode low-intersection Hamiltonian (each Hamiltonian terms involves at most a constant number of bosonic modes and each mode is involved in at most a constant number of terms) with $\Or(\epsilon^{-1}\log(m/\delta))$ total evolution time, where $\epsilon$ is the maximum error on each Hamiltonian coefficients, and $\delta$ is the allowed failure probability. Here, by the total evolution time we mean the total time we need to let the system evolve under the unknown Hamiltonian.

Our result relies on a framework to analyze the effect of engineered dissipation. Specifically, we show strong dissipation projects the dynamics to the kernel of the jump operators. This effect is used in various ways (see Propositions~\ref{prop-main:single-modes-adiabatic-limit}, \ref{prop:few-modes-adiabatic-limit}, \ref{prop:decoupling-adiabatic-limit}) to decouple the system into clusters that do not interact with each other, and to make the Hamiltonian diagonal in a given basis within each cluster. 

The same effect also underlies the stability of the cat code, where a bosonic mode is stabilized within the subspace spanned by $\ket{\pm \alpha}$ through engineered dissipation, significantly reducing the phase-flip (and sometimes also the bit-flip) error rate \cite{Mirrahimi.2014,Azouit.2016,Guillaud.2019,Guillaud.2023}. While previous works have analyzed convergence toward the limiting dynamics in related settings, these results have so far focused on the optimal convergence rate and extensions in finite dimensions \cite{Zanardi.2014, Zanardi.2015, burgarth2020quantum} or on bounded Lindbladians \cite{Glueck.2016, Barankai.2018}. More recently, a convergence result for an unbounded driving term with a bounded supplementary Lindbladian and a slightly suboptimal convergence rate was proven in \cite{salzmann2024quantitative}. Here, the convergence of the dissipative driving term to a projection for large times is assumed.

Our work is the first to directly address the cat qubit setting. To do so, we first prove that the dissipation converges exponentially fast to the cat code space for large times, thereby extending the result in \cite{Azouit.2016}. Next, we provide an adiabatic limit showing that Hamiltonians are projected to the cat code space. For the driving Hamiltonian defined by the position operator, the effective dynamics reduces to the rotation around the $x$-axis \cite{Mirrahimi.2014}, experimentally realized in \cite{Touzard.2018}, and we provide a rigorous proof with optimal and explicit constants to verify the construction of this in the cat code. The techniques we develop in this work open the way to analyzing the effect of strong dissipation in quantum control and learning settings in a mathematically rigorous manner. %

Hamiltonian learning is essentially the task of identifying the generator of Markovian dynamics, and can be seen as a quantum analog of learning the generator of a Markov chain \cite{Hao2018learning,Han2021optimal,Wolfer2019minimax,Wolfer2021discrete}. We believe that some of the ideas presented in this work, such as the divide-and-conquer approach used in multi-mode learning, and the multi-step decision-making strategy to enhance the robustness of frequency estimation, can also be valuable for classical learning tasks. Additionally, our work motivates further exploration of learning tasks for continuous-variable quantum systems. While these systems present certain challenges, they also offer unique features that can enhance learning tasks. As an example, diffusion models \cite{sohl2015deep,ho2020denoising,song2019generative,song2020score} are most naturally formulated for distributions over continuous variables, and are only recently generalized to discrete variables with non-trivial modifications \cite{lou2023discrete,ren2024discrete,ren2025fast}. Our work provides important ideas and techniques for exploring continuous-variable learning in the quantum realm.

\section{Main results}
\label{sec:main_results}

The main contribution of this article is to show that there exist Heisenberg-limited algorithms to learn bosonic Hamiltonians of bounded degree, with the help of engineered dissipation. 

\paragraph{Setting} Bosonic systems with $\mathfrak{k}$ bosonic modes are described in terms of creation and annihilation operators $b^\dagger_i$ and $b_i$, one for each mode. 
A bosonic mode $i$ is in a state $\ket{k}_i$, indexed by a non-negative integer $k$, or their superposition. The annihilation operator $b_i$ annihilates an excitation on the $i$th mode by $b_i\ket{k}_i=\sqrt{k}\ket{k-1}_i$, while the creation operator $b_i^\dag$ has the opposite effect, with $b_i^\dag\ket{k}_i=\sqrt{k+1}\ket{k+1}_i$.
We start by defining the considered Hamiltonian as follows:
        \begin{defn}[Low-intersection bosonic Hamiltonian] \label{def-main:bosonic-H}
            A low-intersection bosonic Hamiltonian acting on $m$ modes is a Hamiltonian that takes the following form:
            \begin{align}\label{eq:main-bosonic-H}
                H = \sum_{a=1}^M E_a\,,
            \end{align}
            where, labeling the sites in $\operatorname{supp}(E_a)$ as $1, \ldots, \mathfrak{k}$, each $E_a$ is an $\mathfrak{k}$-mode interaction of the form
            \begin{align*} 
                E_a = \sum_{\mathbf{j},\mathbf{j}' \in \mathbb{N}^{ \mathfrak{k}}\,:\, \|\mathbf{j}+\mathbf{j}'\|_1\le d} h^{(a)}_{\mathbf{j},\mathbf{j}'} (\mathbf{b}^\dagger)^{\mathbf{j}}\,\mathbf{b}^{\mathbf{j}'}\,.
            \end{align*}
            We assume that at least one of $\mathbf{j} \neq 0$ or $\mathbf{j}' \neq 0$ holds. More generally, given a set $C\subseteq [m]$, we denote by $H_C$ the sum of all interactions supported in $C$. Here, we use the multi-index notation $\mathbf{b}^{\mathbf{j}}=\prod_{\ell\in \operatorname{supp}(E_a)}b_\ell^{j_\ell}$. Additionally, we assume that $|h^{(a)}_{\mathbf{j},\mathbf{j}'}| \le 1$. Since $E_a$ is self-adjoint, we have $\overline{h^{(a)}_{\mathbf{j},\mathbf{j}'}} = h^{(a)}_{\mathbf{j}',\mathbf{j}}$. We assume that each $E_a$ acts on at most $ \mathfrak{k}=\mathcal{O}(1)$ modes and that each $E_a$ overlaps with at most $ \mathfrak{d}=\mathcal{O}(1)$ other interactions $E_b$.
        \end{defn}
Our learning protocol generates estimates $\hat{h}^{(a)}_{\mathbf{j},\mathbf{j}'}$ such that
\begin{equation}
    \max_{\mathbf{j}, \mathbf{j}'}|\hat{h}^{(a)}_{\mathbf{j},\mathbf{j}'}-{h}^{(a)}_{\mathbf{j},\mathbf{j}'}|\leq \epsilon\text{ with probability at least }1-\delta.
\end{equation} 

    \paragraph{Heisenberg-limited learning of low-intersection bosonic Hamiltonians} Our main result is that a Hamiltonian as in Definition \ref{def-main:bosonic-H} can be learned by a Heisenberg-limited algorithm. Our core insight is that by adding sufficiently strong engineered dissipation, we can restrict the time evolution to a subspace of our choosing, which allows us to extract the coefficients of $H$ from the time evolution under an effective Hamiltonian. The following is an informal version of Theorem \ref{thm:estimating_multi_mode}: 
\begin{thm} \label{thm-main:multi-mode-learning}
    There exists an algorithm which makes use of dissipation with strength $\gamma=\Or(m^2\epsilon^{-1}\log^{2d+1/2}(1/\epsilon))$ as introduced below, that can estimate all coefficients of $H$ to precision $\epsilon$ with probability at least $1-\delta$. It requires
    \begin{align*}
        \Or((1/\epsilon)\log(m/\delta)) & \quad \text{total evolution time, and} \\
        \Or(\log^2(\log(1/\epsilon)/\epsilon)\log(m/\delta)) & \quad \text{experiments.}
    \end{align*}
\end{thm} 
Thus, the learning results in this article can be seen as combining the best aspects of previous works like \cite{mobus2023dissipation} and \cite{LiTongNiGefenYing2023heisenberg}: While \cite{mobus2023dissipation} gives learning algorithms for any bosonic Hamiltonian of bounded degree using engineered dissipation, it does not attain Heisenberg scaling, whereas \cite{LiTongNiGefenYing2023heisenberg} gives learning algorithms with Heisenberg scaling, but is limited to Hamiltonians of Bose-Hubbard type.

The generator $\mathcal L$ of the dissipation consists of single-mode terms $\mathcal L_i$.  On each mode $i$, for $\alpha_i \in \mathbb C$ we consider a dissipation in the Lindbladian form $\mathcal L_i := \mathcal L[L_{1,\alpha_i}]+\mathcal L[L_{ r_i,\alpha_i}]$, where $\mathcal{L}[L]=L\cdot L^\dagger-\frac{1}{2}\{L^\dagger L,\cdot\}$ and
\begin{equation*}
    L_{r, \alpha} = b^{r} (b-\alpha)
\end{equation*}
are the jump operators of the Lindblad operators $\mathcal L[L_{r,\alpha}]$. Here, $r_i$ has to be chosen sufficiently large as a function of $d$. We then run a time evolution defined by
\begin{equation}\label{eq:main-master-equation}
    \frac{d}{dt}\rho(t) = -i[H, \rho(t)] + \gamma \mathcal{L}(\rho(t)), \qquad \rho(0) = \rho_0 \,.
\end{equation}
Above, we used the standard notations for the (formal) commutator $[A,B]=AB-BA$ and anticommutator $\{A,B\}=AB+BA$. Motivated by the bosonic cat codes \cite{lescanne2020exponential, chamberland2022building, Guillaud.2023}, which are based on the jump operators $ L = b^r - \alpha^r $ defining higher-order photon losses, we use a modified photon dissipation with the jump operators $ L_{r, \alpha} $. In the case of $ r = 2 $, the photon dissipation of the bosonic cat code was experimentally realized, for example, in \cite{Touzard.2018}. Following the construction of the photon dissipation, we discuss a possible implementation of the modified photon dissipation in Section \ref{sec:experimental_realization}. As in \cite{mobus2023dissipation}, one can use Trotterization to alternate between Hamiltonian evolution and dissipation for short times.

\paragraph{Adiabatic approximation for general Lindbladian evolutions}
To prove soundness of our learning scheme, our main technical contribution is to provide precise and explicit bounds on the convergence of the engineered dissipation to the time evolution generated by the effective Hamiltonian. For that, we derive a new quantitative adiabatic approximation estimate for general Lindbladian evolutions with unbounded generators. As we will show at the end of this section, it has applications beyond the setting of Hamiltonian learning and is therefore of independent interest. The following theorem is a simplified version of Theorem \ref{thm:general-adiabatic-limit}. 
\begin{thm}[Informal]\label{thm-main:adiabatic-approximation}
    Let $H$ be a self-adjoint operator defined on a Hilbert space $\mathscr{H}$, and let $\mathcal{L}$ be a Lindbladian over the state space on $\mathscr{H}$. Let $P$ be the orthogonal projection onto the intersection of the finite-dimensional kernels of the jump operators defining $\mathcal{L}$. Then, under some reasonable technical conditions on $H$ and $\mathcal L$, for any state $\rho=P\rho P$ and $\gamma>0$,
    \begin{equation*}
        \left\|e^{t(\gamma \mathcal{L}+\mathcal{H})}(\rho)-e^{t\mathcal{H}_{\operatorname{proj}}}(\rho)\right\|_1\le \frac{tC}{\gamma}+\frac{C'}{\gamma}\,,    
    \end{equation*}
    where $C$ and $C^\prime$ are constants. Here, $\mathcal H=-i[H, \cdot]$ and $\mathcal{H}_{\operatorname{proj}}$ is an effective version of $\mathcal H$ restricted to the image of $P$.
\end{thm}
The theorem states that if the dissipation strength $\gamma$ is large enough, the time evolution generated by the Hamiltonian $H$ is effectively restricted to the invariant subspace of the dissipative evolution generated by $\mathcal L$. 

\paragraph{Application to cat codes} Another interesting and valuable application of our adiabatic Theorem \ref{thm-main:adiabatic-approximation} is in the context of bosonic cat codes.
This class of continuous-time error correction codes relies on the use of $r$-photon driven dissipation, which takes the form 
\begin{equation*}
    \frac{d}{dt} \rho(t) = \mathcal{L}[L_r](\rho(t)) \qquad \text{with} \qquad L_r \coloneqq b^r - \alpha^r\,.
\end{equation*}
One fundamental property of the above dissipation is that for large times, the evolution drives any state exponentially fast to the finite-dimensional code space $$ \mathcal{C}_r(\alpha) \coloneqq \mathrm{span}\left\{\ket{\alpha_1}\bra{\alpha_2} : \alpha_1, \alpha_2 \in \left\{\alpha e^{\frac{i2\pi j}{r}} : j \in \{0, \ldots, r-1\}\right\}\right\}$$ in the following sense \cite{Azouit.2016}:
\begin{equation*}
    \tr[L_r e^{t\mathcal{L}[L_r]}(\rho) L_r^\dagger] \leq e^{-{tr}!} \tr[L_r \rho L_r^\dagger]\,.
\end{equation*}
We do not only extend the above weighted convergence rate to the convergence of the quantum dissipative evolution in the natural trace norm (see Proposition \ref{prop:cat-convergence}), but we also prove an adiabatic limit for bounded-degree Hamiltonians, which is stated informally below (see Proposition \ref{prop:adiabatic-cat}):
\begin{prop}\label{prop-main:adiabatic-cat} 
    Let $H$ be as in \eqref{eq:main-bosonic-H} for $m=1$, $d/2\leq r$, and let $P$ be the orthogonal projection onto $\mathcal{C}_r(\alpha)$. Then, for all $t \geq 0$ and $\rho \in \mathcal{C}_r(\alpha)$,
    \begin{equation*}
        \|e^{-it[H, \cdot] + t\gamma\mathcal{L}[b^r - \alpha^r]}(\rho) - e^{-it[PHP, P \cdot P]}(\rho)\|_1 \leq \frac{tC}{\gamma}+\frac{C'}{\gamma}\,
        \end{equation*}
    for constants $C, C'\geq0$.
\end{prop}
Full details of our results on cat codes can be found in Appendix \ref{subsec:bosonic_cat}. Thus, we have demonstrated that the $r$-photon driven dissipation of sufficient strength leads to an effective time evolution on the code space for any bounded-degree bosonic Hamiltonian. For $r = 2$, this specifically provides explicit convergence bounds for implementing rotations around the $x$-axis on the code space, which has been experimentally realized in \cite{Touzard.2018}.

\section{Adiabatic evolutions}
    Quantum dynamical time-evolutions are usually characterized by the initial value problem
    \begin{equation}\label{eq:inital-value-pb}
        \frac{d}{dt}\rho(t)=\mathcal{L}(\rho(t))\,,\qquad \rho(0)=\rho_0
    \end{equation}
    given by a possibly unbounded operator $\mathcal{L}$ defined on a subspace (its domain $\mathcal{D}(\mathcal{L})$) of trace-class operators $T_1(\mathscr{H})$ on a separable Hilbert space $\mathscr{H}$ (see \cite[Sec.~3.6]{Simon.2015-4} for more details). Note that the trace-class operators are equipped with the trace norm $\|\cdot\|_1$, the Hilbert space with the induced norm $\|\cdot\|$, and the bounded operators $\mathcal{B}(\mathscr{H})$ with the operator norm $\|\cdot\|_{\infty}$. Due to the quantum mechanical nature of the generated evolution, it is assumed that the initial value is a quantum state, i.e., a positive semidefinite operator $\rho_0\geq0$ with unit trace $\tr[\rho_0]=1$. 

    Assume that the initial value problem \eqref{eq:inital-value-pb} admits a unique, trace-preserving, and completely positive solution, called a quantum dynamical semigroup, for which we write $\rho(t)=e^{t\mathcal{L}}$ (see \cite[Sec.~3.3]{Holevo.2001} for more details). Motivated by the GKSL-representation theorem \cite{Gorini.1976,Lindblad.1976}, we consider unbounded generators --- Lindbladians --- of the form
    \begin{equation}\label{eq:linbladian}
        \mathcal{L}(X):=\sum_{i\in\mathcal{I}} L_i X L_i^\dagger - \frac{1}{2} \big\{L_i^\dagger L_i, X\big\}
    \end{equation}
    defined by a finite set of jump operators $\{L_i\}_{i\in\mathcal{I}}$, i.e., for some closed, densely defined linear operators $(L_i,\operatorname{dom}(L_i))$ on $\mathscr{H}$. Note that $L_i^\dagger$ denotes the adjoint operator of $L_i$ (see \cite[Ch.~7.1]{Simon.2015-4}).

The two main problems we need to overcome to prove Theorem \ref{thm-main:adiabatic-approximation} are, first, that we are dealing with unbounded generators, meaning that in every step where we use tools like Duhamel's formula, operator norm bounds are cannot be used—for example, to approximate the difference to the effective time evolution. Second, we are dealing with Lindbladians, which describe non-reversible quantum evolutions. In addition to having no quantum mechanical inverse, they do not have a mathematical inverse time representation either due to the unboundedness of the generator. The reasonable technical conditions mentioned in Theorem \ref{thm-main:adiabatic-approximation} guarantee that we can still control the generators even though they are unbounded. Concretely, we have to impose that 
\begin{enumerate}
    \item the jump operators $L_i$ defining $\mathcal L$ have an energy gap,
    \item $H$ is relatively bounded with respect to the jumps $L_i$,
    \item different products of jump operators and the Hamiltonian are bounded on the image of $P$.
\end{enumerate}
To use the adiabatic approximation estimate in Theorem \ref{thm-main:adiabatic-approximation}, we thus need to verify that these technical conditions are met. This way, we obtain an explicit dependence of $C$ and $C^\prime$ on the degree of the Hamiltonian $d$ and the parameter $\alpha$ in the dissipation. For the proof of Theorem \ref{thm-main:adiabatic-approximation} and more applications, we refer the reader to Appendix \ref{sec:adiabatic_evolutions}.

We will illustrate this for the single-mode setting, i.e., where the Hamiltonian is as in \eqref{eq:main-bosonic-H} with $m=1$. The dissipation is defined by $\mathcal{L}=\mathcal{L}[L_{1,\alpha}]+\mathcal{L}[L_{r,\alpha}]$. It is easy to see in this case that $P$ is the projection onto the kernel of $L_{1, \alpha}$, since the latter is included into that of $L_{r, \alpha}$. Then, we prove the following adiabatic theorem:
        \begin{prop}\label{prop-main:single-modes-adiabatic-limit}
            Let $d \in \mathbb{N}$, $H$ be as defined as in \eqref{eq:main-bosonic-H} with $m=1$,  $r=\lceil\frac{d}{2}\rceil-1$, $\alpha\in\C$, and let $P$ be the orthogonal projection onto $\mathrm{span}(\ket{0},\ket{\alpha})$. Then, for all $t \geq 0$ and $\rho=P\rho P$, 
            \begin{equation*}
                \|e^{t\cH + t\gamma\mathcal{L}[b^r(b-\alpha)]+t\gamma\mathcal{L}[b(b-\alpha)]}(\rho) - e^{-it[PHP, P \cdot P]}(\rho)\|_1 \le \frac{tC}{\gamma}+\frac{C'}{\gamma}\,
            \end{equation*}
            for constants $C=C_d(|\alpha|^{4d+2}+1)$ and $C'=C_d'(|\alpha|^d+1)$ for $C_d$ and $C_d'$ depending on $d$. 
        \end{prop}

        The proof can be found in Appendix \ref{subsec:single-mode-adiabatic-thm}. It consists of verifying the assumptions needed to apply Theorem \ref{thm-main:adiabatic-approximation}.

\section{Learning a single-mode Hamiltonian} \label{sec:main-learning}

In this section we will describe our Hamiltonian learning protocol in the single-mode case, i.e., $m=1$. At the end of this section, we will sketch how to obtain Theorem \ref{thm-main:multi-mode-learning} from this.

We consider learning a single mode Hamiltonian, which is the Hamiltonian \eqref{eq:main-bosonic-H} for $m=1$:
\begin{equation}
    H = \sum_{j=0}^d\sum_{j'=0}^{d-j} h_{j,j'}(b^{\dag})^j b^{j'}.
\end{equation}
To remove the global phase we let $h_{0,0}=0$. We also assume that $|h_{j,j'}|\leq 1$. Because $H$ is self-adjoint, we have $\overline h_{j,j'} = h_{j',j}$. Our learning protocol will generate estimates $\hat{h}_{j,j'}$ such that
\begin{equation}
    \label{eq:single_mode_err_metric}
    \max_{j,j'}|\hat{h}_{j,j'}-h_{j,j'}|\leq \epsilon\text{ with probability at least }1-\delta.
\end{equation}
A pseudocode for the protocol is provided in Algorithm~\ref{alg:single_mode_learning} in the appendix and an illustration of it is given in Fig. \ref{fig:algo}.

\subsection{The effective Hamiltonian}

We apply strong dissipation with jump operators $L_{1, \alpha}=b(b-\alpha)$ and $L_{r, \alpha}=b^{r}(b-\alpha)$ where $r=\lceil\frac{d}{2}\rceil-1$ to constrain the dynamics to the subspace spanned by $\ket{0}$ and $\ket{\alpha}$ as well as to regularize it. Here, $\ket{\alpha}$ is the coherent state defined by $\alpha \in \mathbb C$ as $e^{-|\alpha|^2/2}\sum_{n \geq 0} \alpha^n/\sqrt{n!} \ket{n}$. In particular, it holds that $b\ket{\alpha} = \alpha \ket{\alpha}$. The effective Hamiltonian is then
\begin{equation}
    H^{\mathrm{proj}} = \Pi H \Pi,
\end{equation}
where $\Pi = \ket{\Psi_0}\bra{\Psi_0} + \ket{\Psi_1}\bra{\Psi_1},$
with
\begin{equation}
    \ket{\Psi_0} = \ket{0},\quad \ket{\Psi_1} = \frac{\ket{\alpha}-\ket{0}\braket{0|\alpha}}{\|\ket{\alpha}-\ket{0}\braket{0|\alpha}\|}.
\end{equation}
In the above $\ket{\Psi_0}$ and $\ket{\Psi_1}$ are obtained by the Gram-Schmidt process, and they form an orthonormal basis for the subspace. One can compute that $\braket{0|\alpha}=e^{-|\alpha|^2/2}$, and  $\braket{\alpha|\Psi_1} = \sqrt{1-e^{-|\alpha|^2}}$.
From this we can see that $\ket{\Psi_1}$ is in fact very close to $\ket{\alpha}$ if $|\alpha|$ is large:
\begin{equation}
    \label{eq:distance_bw_Psi1_and_alpha}
    \|\ket{\Psi_1}-\ket{\alpha}\|=\Or(e^{-|\alpha|^2/2}).
\end{equation}
Relative to the basis $\ket{\Psi_0}$ and $\ket{\Psi_1}$, the effective Hamiltonian is represented by
\[
\begin{pmatrix}
    \braket{\Psi_0|H|\Psi_0} & \braket{\Psi_0|H|\Psi_1} \\
    \braket{\Psi_1|H|\Psi_0} & \braket{\Psi_1|H|\Psi_1}
    \end{pmatrix}
    \approx
    \begin{pmatrix}
    0 & 0 \\
    0 & \braket{\alpha|H|\alpha}
\end{pmatrix}.
\]
Direct calculation shows that the entry-wise error in the above approximation is at most $\Or(|\alpha|^d d^2 e^{-|\alpha|^2/2})$. 
Moreover, we have 
$
\|\ket{\Psi_1}\bra{\Psi_1}-\ket{\alpha}\bra{\alpha}\|\leq 2\|\ket{\alpha}-\ket{\Psi_1}\|.
$
Therefore we define
\begin{equation}
\label{eq:approx_effective_ham}
    \widetilde{H}^{\mathrm{proj}} = \ket{\alpha}\braket{\alpha|H|\alpha}\bra{\alpha},
\end{equation}
which is a good approximation of the effective Hamiltonian but is of a simpler form:
\begin{equation}
\label{eq:effective_ham_approx_err}
    \|\widetilde{H}^{\mathrm{proj}}-H^{\mathrm{proj}}\|_\infty = \Or(|\alpha|^d d^2 e^{-|\alpha|^2/2}).
\end{equation}

\subsection{Experimental setup}
\label{sec:experimental_setup}

We will run phase estimation experiments that are to be described in this section, to extract the eigenvalues of $H^{\mathrm{proj}}$, from which we will extract the coefficients $h_{j,j'}$. For this, we need to introduce an ancilla qubit. In each experiment, we initialize the system in the state $\ket{+}\ket{0}$, apply controlled displacement $D(\alpha)$ to obtain the state $\frac{1}{\sqrt{2}}(\ket{0}\ket{0}+\ket{1}\ket{\alpha})$, and then let the system evolve under the effective Hamiltonian. We then apply the controlled displacement $D(-\alpha)$, and measure the ancilla qubit in either the $X$ or $Y$ basis.
This protocol is illustrated with a quantum circuit in Fig. \ref{fig:circuit1}. 
Mathematically, the procedure can be described as follows:
\begin{equation}
    \begin{aligned}
        \ket{+}\ket{0}&\xmapsto{\operatorname{controlled}-D(\alpha)}\frac{1}{\sqrt{2}}(\ket{0}\ket{0}+\ket{1}\ket{\alpha}) \\
        &\xmapsto{e^{-iH^{\mathrm{proj}}t}} \frac{1}{\sqrt{2}}(\ket{0}e^{-iH^{\mathrm{proj}}t}\ket{0}+\ket{1}e^{-iH^{\mathrm{proj}}t}\ket{\alpha}) \\
        &\xmapsto{\operatorname{controlled}-D(-\alpha)} \frac{1}{\sqrt{2}}(\ket{0}e^{-iH^{\mathrm{proj}}t}\ket{0}+\ket{1}D(-\alpha)e^{-iH^{\mathrm{proj}}t}\ket{\alpha}).
    \end{aligned}
\end{equation}
Measuring the ancilla qubit in the $X$ and $Y$ bases, we have
\begin{equation}
    \braket{X\otimes I} = \Re \braket{0|e^{iH^{\mathrm{proj}}t}D(-\alpha)e^{-iH^{\mathrm{proj}}t}|\alpha},\quad \braket{Y\otimes I} = \Im \braket{0|e^{iH^{\mathrm{proj}}t}D(-\alpha)e^{-iH^{\mathrm{proj}}t}|\alpha}.
\end{equation}
From \eqref{eq:distance_bw_Psi1_and_alpha} and \eqref{eq:effective_ham_approx_err}, we have
\begin{equation}
\begin{aligned}
    \braket{0|e^{iH^{\mathrm{proj}}t}D(-\alpha)e^{-iH^{\mathrm{proj}}t}|\alpha} &= \braket{0|e^{i\widetilde{H}^{\mathrm{proj}}t}D(-\alpha)e^{-i\widetilde{H}^{\mathrm{proj}}t}|\alpha} + \Or(|\alpha|^d d^2 e^{-|\alpha|^2/2}t). \\
\end{aligned}
\end{equation}
Moreover
\begin{equation}
\begin{aligned}
     \braket{0|e^{i\widetilde{H}^{\mathrm{proj}}t}D(-\alpha)e^{-i\widetilde{H}^{\mathrm{proj}}t}|\alpha} 
    &= e^{-i\braket{\alpha|H|\alpha}t}\braket{0|e^{i\widetilde{H}^{\mathrm{proj}}t}D(-\alpha)|\alpha}  \\
    &= e^{-i\braket{\alpha|H|\alpha}t}\braket{0|e^{i\widetilde{H}^{\mathrm{proj}}t}|0}  \\
    &= e^{-i\braket{\alpha|H|\alpha}t} + \Or(|\alpha|^d d^2 e^{-|\alpha|^2/2}t), \\
\end{aligned}
\end{equation}
where the last line uses $\|\widetilde{H}^{\mathrm{proj}}\ket{0}\|\leq \Or(|\alpha|^d d^2 e^{-|\alpha|^2/2})$. Hence we have
\begin{equation}
     \braket{0|e^{iH^{\mathrm{proj}}t}D(-\alpha)e^{-iH^{\mathrm{proj}}t}|\alpha} = e^{-i\braket{\alpha|H|\alpha}t}+\Or(|\alpha|^d d^2 e^{-|\alpha|^2/2}t). 
\end{equation}
Therefore from the experiment we can approximately obtain the value of $e^{-i\braket{\alpha|H|\alpha}t}$. We can also write out the real and imaginary parts separately:
\begin{equation}
\label{eq:cos_sin_from_experiment}
    \begin{aligned}
        \braket{X\otimes I} &= \cos(\braket{\alpha|H|\alpha}t)+\Or(|\alpha|^d d^2 e^{-|\alpha|^2/2}t),\,\\\braket{Y\otimes I} &= \sin(\braket{\alpha|H|\alpha}t)+\Or(|\alpha|^d d^2 e^{-|\alpha|^2/2}t)\,.
    \end{aligned}
\end{equation}

\begin{figure}[t]
   \centering
\begin{subfigure}{0.3\textwidth}
        \centering
        \includegraphics[width=\textwidth]{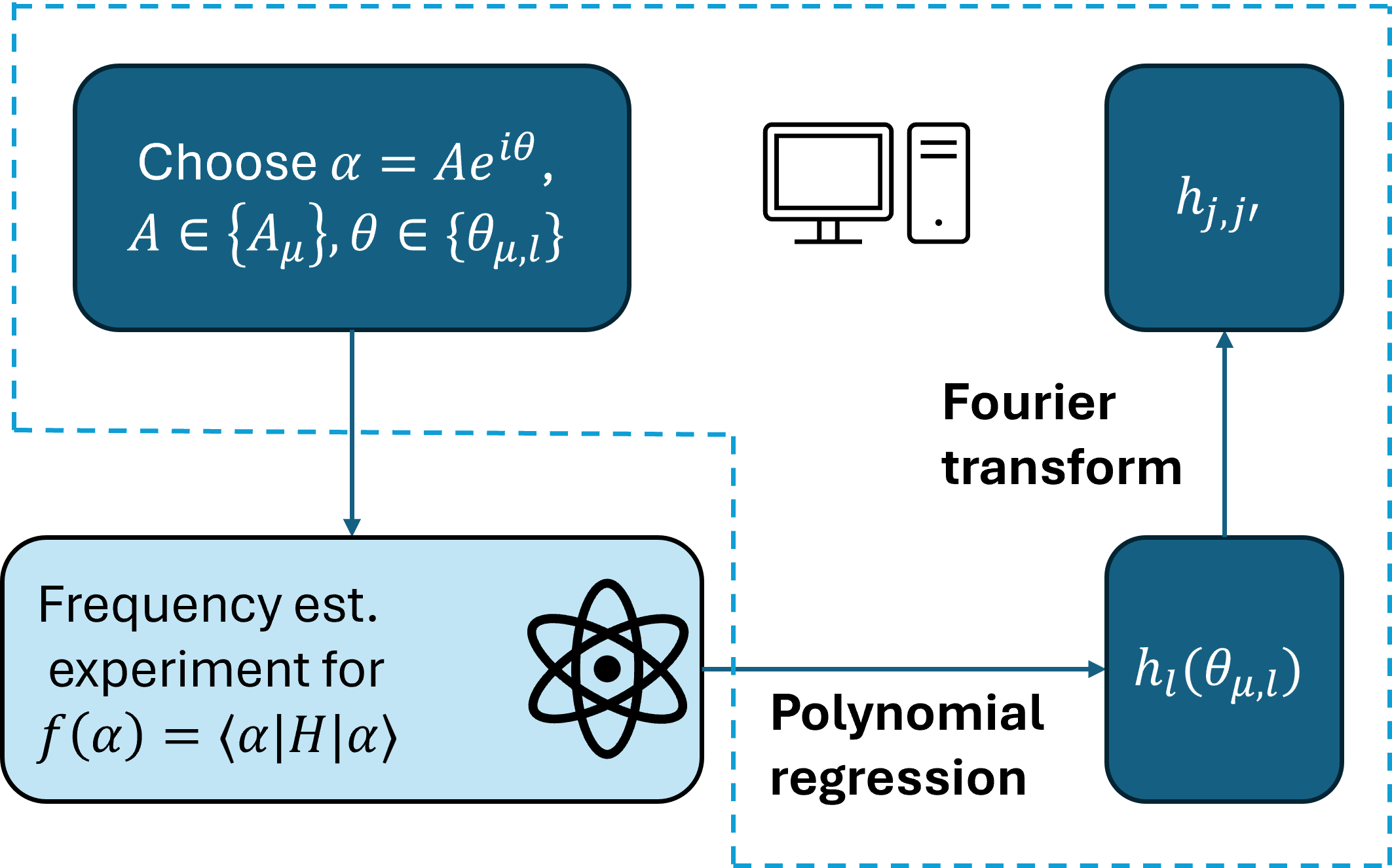}
       \caption{}
       \label{fig:algo}
\end{subfigure}        
\quad
\begin{subfigure}{0.5\textwidth}
\centering
\resizebox{\textwidth}{0.09 \textwidth}    
{\begin{quantikz}
\lstick{qubit $|0\rangle$} & \gate{H} & \ctrl{1} & \qw
& \ctrl{1} &\gate{R_{y/x}(\pi/2)} & \meter{}  \\
\lstick{b. mode $|\text{0}\rangle$} & \qw  & \gate{D(\alpha)} &  \gate{e^{-iH^{\text{proj}}t} } & \gate{D(-\alpha)} &\qw  &
\end{quantikz}
}
\caption{ }
 \label{fig:circuit1}
\end{subfigure}
\caption{(a) The algorithm for estimating coefficients $\{h_{j,j'}\}$ for a single-mode bosonic Hamiltonian. All steps except for the frequency estimation experiment are run on a classical computer. (b) A quantum circuit describing the experimental measurement protocol for a single bosonic mode. Note that the time evolution by $H_{\text{proj}}$ is implemented via the original time evolution and a sufficiently strong dissipation. }
\end{figure}

\subsection{Robust frequency estimation}
\label{sec:robust_frequency_estimation}

In this section we introduce the robust frequency estimation protocol used in \cite{LiTongNiGefenYing2023heisenberg}. 

\begin{thm}
    \label{thm:robust_frequency_estimation}
    Let $\theta\in[-\Phi,\Phi]$.
    Let $X(t)$ and $Y(t)$ be independent random variables satisfying
    \begin{equation*}
        \begin{aligned}
            &|X(t)-\cos(\theta t)|< 1/\sqrt{8}, \text{ with probability at least }2/3, \\
            &|Y(t)-\sin(\theta t)|< 1/\sqrt{8}, \text{ with probability at least }2/3.
        \end{aligned}
    \end{equation*}
    Then with independent samples $X(t_1),X(t_2),\cdots,X(t_{\Gamma})$ and $Y(t_1),Y(t_2),\cdots,Y(t_{\Gamma})$, with
    \begin{equation}
        {\Gamma}=\Or(\log^2(\Phi/\epsilon)), \quad T=t_1+t_2+\cdots+t_{\Gamma}=\Or(1/\epsilon), \quad \max_j t_j=\Or(1/\epsilon),
    \end{equation}
    and $t_j\geq 0$, we can construct  an estimator $\hat{\theta}$ such that
    \begin{equation}
        \sqrt{\mathbb{E}[|\hat{\theta}-\theta|^2]}\leq \epsilon.
    \end{equation}
\end{thm}
The proof can be found in Appendix \ref{subsec:single_mode_learning}. The proof is given by describing the algorithm to achieve the above scalings. The algorithm works by turning the estimation problem into a sequence of decision problems. Each decision problem refines our knowledge of the unknown parameter $\theta$, and can be correctly solved with high probability using a few rounds of majority voting. This results in an algorithm that is robust against noise and achieves the $1/\epsilon$ total evolution time (i.e., $T$) scaling.

\subsection{Recovering the coefficients}
\label{sec:recovering_the_coefficients}

In Section~\ref{sec:experimental_setup} we introduced an experiment that, for any $t>0$, allows us to obtain in expectation $\cos(\braket{\alpha|H|\alpha}t)$ and $\sin(\braket{\alpha|H|\alpha}t)$ by measuring in the $X$ and $Y$ bases respectively on an ancilla qubit (see \eqref{eq:cos_sin_from_experiment}). We note that this is exactly what is needed in the robust frequency estimation procedure in Theorem~\ref{thm:robust_frequency_estimation}. Using the fact that
$
|\braket{\alpha|H|\alpha}|\leq \sum_{j=0}^d\sum_{j'=0}^{d-{j}}h_{j,j'}|(\alpha^*)^j \alpha^{j'}|=\Or(d^2|\alpha|^{d}),
$ 
we know from Theorem~\ref{thm:robust_frequency_estimation} that $\braket{\alpha|H|\alpha}$ can be estimated to precision $\epsilon_1$ with total evolution time $\Or(1/\epsilon_1)$ and $\Or(\log^{{2}}(d^2|\alpha|^{{d}}/\epsilon_1))$ samples.

With $\braket{\alpha|H|\alpha}$ we will be able to recover the coefficients $h_{j,j'}$ by carefully choosing $\alpha$. Letting $\alpha=A e^{i\theta}$, we have 
\begin{equation}
    \braket{\alpha|H|\alpha} = \sum_{j=0}^d\sum_{j'=0}^{d-{j}}h_{j,j'}A^{j+j'}e^{i(j'-j)\theta} = \sum_{l=1}^d A^l e^{i\theta l} \sum_{j=0}^{l} h_{j,l-j} e^{-2ij\theta}.
\end{equation}
We note that this is a polynomial of $A$ and proceed to estimate the coefficients
\begin{equation}
    \label{eq:poly_of_r_coefs}
    h_l(\theta) = e^{i\theta l} \sum_{j=0}^{l} h_{j,l-j} e^{-2ij\theta}.
\end{equation}
Fixing $\theta$, we first estimate $\braket{\alpha|H|\alpha}$ for $\alpha=A e^{i\theta}$ and $A=A_\mu$, $\mu=1,2,\cdots,d+1$,  
where
\begin{equation}
\label{eq:choices_of_ri_ab}
    A_\mu = \frac{A_-+A_+}{2} + \frac{A_+-A_-}{2}\cos\left(\frac{(2\mu-1)\pi}{2d+2}\right).
\end{equation}
We choose $0<A_-<A_+$ such that $A_+-A_-\geq 4A_-$.
For each $A_\mu$ we estimate $\braket{\alpha|H|\alpha}$ to precision $\epsilon_1$.
We then perform Chebyshev interpolation on the interval $[A_-,A_+]$ using these data points to get a set of coefficients $\hat{h}_l$ for $l=0,1,\cdots,d$.
From Corollary~\ref{cor:coef_est_error} in Appendix \ref{sec:polyregression}, we know that 
\begin{equation}
    \label{eq:coef_est_error_r}
    |h_l(\theta)-\hat{h}_l|=\Or\left(\frac{2^l \epsilon_1}{(A_+-A_-)^l}\right).
\end{equation}
Further requiring $A_+-A_-=\Omega(1)$ we have $|h_l(\theta)-\hat{h}_l|=\Or(\epsilon_1)$.

Having estimated all polynomial coefficients $h_l(\theta)$, we then choose a set of $\theta$ to recover the coefficients $h_{j,j'}$. Note that $h_{j,l-j}$, $j=0,1,\cdots,l$ are in fact the Fourier coefficients of $h_l(\theta)$. Therefore we can recover $h_{j,l-j}$ via a Fourier transform.

Choosing $\theta_{u,l} = \pi u/(l+1)$, $u=0,1,\cdots,l$, we have
\begin{equation}
    e^{-i\theta_{u,l}l}h_l(\theta_{u,l}) = \sum_{j=0}^l h_{j,l-j} e^{-i\frac{2\pi ju}{l+1}},
\end{equation}
and therefore
\begin{equation}
\label{eq:fourier_transform_for_coefs}
    h_{j,l-j} = \frac{1}{l+1}\sum_{u=0}^l e^{-i\theta_{u,l}l}h_l(\theta_{u,l})e^{i\frac{2\pi ju}{l+1}}.
\end{equation}
From this we can see that after we evaluate $h_l(\theta_{u,l})$ to precision $\Or(\epsilon_1)$, \eqref{eq:fourier_transform_for_coefs} helps us compute $h_{j,l-j}$ to precision $\Or(\epsilon_1)$ as well.

In the above we need to estimate $\braket{\alpha|H|\alpha}$ for ${\alpha}=A_\mu e^{i\theta_{u,l}}$, $\mu=1,2,\cdots,d+1$, and $u=0,1,\cdots,l$ for each $l=1,2,\cdots,d$. In total $\Or(d^3)=\Or(1)$ different values of $\braket{\alpha|H|\alpha}$ are needed, and ensuring all of them to be $\epsilon_1$-accurate with probability at least $1-\delta$ requires an overhead of $\Or(\log(1/\delta))$ using the Chernoff bound.

From Proposition~\ref{prop-main:single-modes-adiabatic-limit}, we know that to keep the deviation from the effective dynamics below $\Or(1)$, which is needed for the robust frequency estimation procedure, we need the dissipation strength to be 
\begin{equation}
\label{eq:choice_gamma1_gammad}
    \gamma = \Or\left(t |\alpha|^{4d+2}\right).
\end{equation}
From \eqref{eq:cos_sin_from_experiment} we can see that in order to keep the error from having a finite $\alpha$ to be a small constant for $t=\Or(1/\epsilon)$, we only need to choose $A_-$ (the lower bound of $|\alpha|$) to be
\begin{equation}
\label{eq:choice_of_a}
    A_- = \Or(\sqrt{\log(1/\epsilon)}).
\end{equation}
Therefore the total deviation from the effective dynamics is bounded by a small constant as is needed for Theorem~\ref{thm:robust_frequency_estimation}. We then arrive at the following theorem: 

\begin{thm}
    \label{thm:estimating_single_mode}
    Running the experiments in Section~\ref{sec:experimental_setup} with dissipation strength of order  $\gamma=\Or(\epsilon^{-1}\log^{2d+1/2}(1/\epsilon))$ and measuring the observables in \eqref{eq:cos_sin_from_experiment}, with $\alpha=A e^{i\theta}$, $A$ taking values $A_\mu$ defined in \eqref{eq:choices_of_ri_ab} with $A_-,A_+$ satisfying $A_+-A_-\geq 4A_- >0$ and $A_-$ chosen according to \eqref{eq:choice_of_a}, we can estimate all coefficients $h_{j,j'}$ to precision $\epsilon$ with probability at least $1-\delta$ with 
    \begin{align*}
    \Or((1/\epsilon)\log(1/\delta)) ~~\text{total evolution time and} ~~~~\Or(\log^2(\log(1/\epsilon)/\epsilon)\log(1/\delta)) ~~\text{experiments.}
    \end{align*}
\end{thm}

\subsection{Extensions to few-mode and multi-mode Heisenberg-limited learning}

The great generality of Theorem \ref{thm-main:adiabatic-approximation} allows us to prove that our single mode algorithm can be adapted to the few-mode and multi-mode setting. In both settings, we consider a low-intersection bosonic Hamiltonian as in Definition \ref{def-main:bosonic-H}. The difference between the few-mode and the multi-mode setting is that while the few-mode setting directly generalizes the single-mode setting, the multi-mode setting uses decoupling to reduce the global Hamiltonian to a sum of local few-mode Hamiltonians: First, we introduce a coloring that defines the choice of single-mode dissipation. To decouple a color from its neighboring color, we choose the dissipation $\mathcal{L}^c_{\mathrm{dec}}$ on the interacting/ boundary modes such that it projects to the vacuum state. This directly erases all interactions in the adiabatic limit. For terms within a color, we act in exactly the same way as in the few-mode setting and use the same learning scheme. This guarantees that the total evolution time and the number of experiments scale only logarithmically with the system size, i.e., the number of modes $m$, and that dissipation strength is only required to scale quadratically in $m$. In the few-mode case, we consider $m = \mathcal{O}(1)$, for which Theorem \ref{thm:estimating_few_mode} in Appendix \ref{sec:learning_few_mode_ham} shows that the few-mode setting has essentially the same requirements as the single-mode setting. The few-mode setting can then be extended via decoupling to the multi-mode setting, which gives Theorem \ref{thm-main:multi-mode-learning}. For the details, we refer to Appendix \ref{sec:learn_multimode}.

\section{Possible experimental realization schemes}
\label{sec:experimental_realization}

The modified photon-dissipation master equation, i.e. $L_{r,\alpha}=b^{r}\left(b-\alpha\right)$, can be obtained using a scheme similar to the one used for engineering the cat code dissipation \cite{lescanne2020exponential, chamberland2022building, Guillaud.2023}.
Cat code dissipation, $L=b^{2}-\alpha^{2},$ is obtained by introducing an ancillary mode, denoted as $a$, and engineering a dynamics of:
\begin{align*}
\frac{d}{dt}\rho\left(t\right)=-ig \left[a^{\dagger}\left(b^{2}-\alpha^{2}\right)+h.c.,\rho\right]+\kappa_{a}\mathcal{D}\left[a\right]\rho.    
\end{align*}
The probe mode is thus coupled to a strongly decaying ancilla through $H=g a^{\dagger}\left(b^{2}-\alpha^{2}\right)+h.c.$
In the limit of $g\ll\kappa_{a}$ we can adiabatically eliminate the ancilla, leading to an effective dynamics of $\frac{d}{dt}\rho_{s}=4\frac{g^{2}}{\kappa_{a}}\mathcal{D}\left[b^{2}-\alpha^{2}\right]\left(\rho_{s}\right)$ of the probe's density matrix $\rho_s$.
In our case, we similarly need to engineer a coupling of $H=ga^{\dagger}b^{r}\left(b-\alpha\right)+h.c.$ between the probe and the ancilla. Using the same derivation as in \cite{Guillaud.2023}, we can adiabatically eliminate the ancilla in the limit of $g\ll\kappa_{a}$ leading to $\mathcal{L}\left[b^{r}\left(b-\alpha\right)\right]$.
A coupling of the form of $H=ga^{\dagger}b^{r}\left(b-\alpha\right)+h.c.$ can be engineered using transmons or asymmetrically threaded SQUID (ATS) \cite{lescanne2020exponential, chamberland2022building}. Relevant Hamiltonian engineering schemes are discussed in Appendix \ref{app:experimental_realizations}.

\section{Outlook} \label{sec:outlook}

In this work, we present the first efficient Hamiltonian learning scheme for a broad class of local Hamiltonians governing 
$m$-mode continuous variable systems, expressed as polynomials in the creation and annihilation operators. Our approach achieves a Heisenberg-limited scaling of
$\mathcal{O}(\epsilon^{-1}\log(m))$  in the total Hamiltonian evolution time required to estimate each coefficient to precision 
$\epsilon$. A key technical contribution enabling this optimal scaling is a new adiabatic theorem for unbounded generators of quantum dynamical semigroups, which we anticipate will have broader applications. As a demonstration, we rigorously establish the emergence of effective finite-dimensional unitary dynamics in multi-photon driven dissipation processes, relevant to quantum computing with cat codes.

An important open question is the optimal strength $ \gamma $ required to confine the effective evolutions within the finite-dimensional subspace where the learning protocol is conducted. While our analysis establishes convergence from a quadratic dependence on the number of modes $ m $, we anticipate that a more refined approach could yield a scaling independent of $ m $. A natural first step in this direction would be to incorporate the physically relevant assumption that the unitary dynamics generated by $ H $ does not lead to an unbounded increase in the system’s energy, as measured by the total photon number.

Finally,
It would be interesting to further explore feasible experimental realizations, in particular finding experimental schemes for the required photon dissipation, $L_{r,\alpha},$ for r>1. It would be also interesting to explore the robustness of this protocol to realistic noises.    

\section*{Acknowledgments}

We would like to thank Paul Gondolf, Robert Salzmann and Yotam Vaknin for their valuable discussions on the topic. A.B. was supported by the ANR project PraQPV, grant number ANR-24-CE47-3023. T.G. acknowledges financial support from the quantum science and technology early-career grant of the Israeli council for higher education. T.M.~acknowledges the support QuantERA II Programme that has received funding from the EU’s H2020 research and innovation programme under the GA No 101017733 and of the Deutsche Forschungsgemeinschaft (DFG, German Research Foundation) - Project-ID 470903074 - TRR 352 and funding by the Federal Ministry of Education and Research (BMBF) and the Baden-Württemberg Ministry of Science as part of the Excellence Strategy of the German Federal and State Governments. C.R. acknowledges financial support from France 2030 under the French National Research Agency award number “ANR-22-PNCQ-0002”. A.H.W. thanks the VILLUM FONDEN for its support with a Villum Young Investigator  (Plus) Grant (Grant No. 25452 and Grant No. 60842) as well  as via the QMATH Centre of Excellence (Grant No. 10059).

\newpage 

\bibliographystyle{unsrt}
\bibliography{ref}

\newpage
\appendix

\section{Adiabatic evolutions}
\label{sec:adiabatic_evolutions}

\subsection{Adiabatic approximation for general Lindbladian evolutions with unbounded generators}
    In this appendix, we will prove our main technical tool, the adiabatic approximation for general Lindbladian evolutions, of which Theorem \ref{thm-main:adiabatic-approximation} is a simplified version. Here, we are interested in the so-called adiabatic limit, which describe the convergence of a time-evolution with a dominant term to an effective dynamic, i.e., time-evolutions with generators of the form $\mathcal{H}+\mathcal{K}+\gamma\mathcal{L}$ for large $\gamma>0$ (see Section \ref{sec:adiabatic_evolutions} for an introduction in quantum dynamical time-evolutions). Here, $\mathcal{L}$ and $\mathcal{K}$ are generators in Lindblad form (see \eqref{eq:linbladian}) defined by the sets of unbounded operators $\{L_i\}_{i\in\mathcal{I}}$ and $\{K_j\}_{j\in\mathcal{J}}$. To prove convergence in $\gamma$ and to obtain an explicit representation of the effective dynamics, we further assume that $P$ is the projection onto the intersection of the (finite-dimensional) kernels of each $L_i$. Given a Hamiltonian $H$, i.e., an (unbounded) essentially self-adjoint operator (see \cite[Ch.~7.1]{Simon.2015-4}), we denote by $\mathcal{H}:=-i[H,\cdot]$ the superoperator corresponding to the commutation with $H$, and write $\mathcal{H}_{\operatorname{proj}}:=-i[H_{\operatorname{proj}},\cdot]$ for the generator corresponding to the Hamiltonian $H_{\operatorname{proj}}:=PHP$. Similarly, we denote by $\mathcal{K}_{\operatorname{proj}}$ the generator corresponding to the jumps $\{PK_jP\}_{j\in\mathcal{J}}$: Finally, we make the following additional assumptions, which are trivial in finite dimensions but important in the case of unbounded operators. We denote $L':=\sum_i (I-P)L_i^\dagger L_i (I-P)$ and $G:=-\frac{1}{2}\sum_{j\in\mathcal{J}} K_j^\dagger K_j.$
    \begin{itemize}
        \item[(i)] $L'\ge \eta \,(I-P)$ for some constant $\eta >0$;
        \item[(ii)] There exist constants $\iota_1,\iota_2\geq 0$ such that, for any $|\psi\rangle\in \operatorname{dom}(L')\subset\operatorname{dom}(H)$,
        \begin{equation*}
            \| H|\psi\rangle\|\le \iota_1 \|L'|\psi\rangle\|+\iota_2\||\psi\rangle\|\,;
        \end{equation*}
        \item[(iii)] There exist constants $\iota'_1,\iota'_2\geq 0$ such that, for any $|\psi\rangle\in\operatorname{dom}(L')\subset\operatorname{dom}(K_j)\cap\operatorname{dom}(K_j^\dagger K_j)$ for all $j$, 
        \begin{equation*}
            \sum_j\|K_j|\psi\rangle\|,\,\sum_j\|K_j^\dagger K_j|\psi\rangle\|\le \iota_1'\|L'|\psi\rangle\|+\iota_2'\||\psi\rangle\|\,;
        \end{equation*}
        \item[(iv)] $\|K_jP\|_\infty$, $\|K_j^\dagger K_jP\|_\infty \le C_K$, $\|HP\|_\infty\le C_H$, and $\|L'HP\|_\infty+\|L'GP\|_\infty\le C_L$ for some constants $C_H,C_K,C_L\geq 0$;
        \item[(v)] For all $j\in\mathcal{J}$, $\mathcal{K}_{j,\operatorname{proj}}:=PK_jP=K_jP$.
    \end{itemize}
    Note that Assumptions (iv) and (v) include the domain assumptions $\operatorname{range}(P)\subset \operatorname{dom}({K_j})$, $\operatorname{dom}(K_j^\dagger K_j)$, $\operatorname{dom}({H})$, $\operatorname{dom}({L'H})$ and $\operatorname{dom}({L'G})$ for all $j\in\mathcal{J}$. By the uniform boundedness theorem, the operators 
    \begin{equation*}
        K_jP\,,\quad K_j^\dagger K_jP\,,\quad HP\,,\quad L'HP\,,\quad\text{and}\quad L'GP
    \end{equation*}
    are uniquely extended to an bounded operator defined on $\mathscr{H}$. Under these assumptions, we are interested in the behavior of the semigroup generated by $\mathcal{K}+\mathcal{H}+\gamma \mathcal{L}$ for some strength $\gamma>0$.
    
    \begin{thm}\label{thm:general-adiabatic-limit}
        Assume that $H$ is a Hamiltonian defined on a Hilbert space $\mathscr{H}$, and that $\mathcal{L}$, $\mathcal{K}$ are Lindbladians defined in \eqref{eq:linbladian}, satisfying assumptions $(i)-(v)$. Then, for any state $\rho=P\rho P$ and $\gamma>0$,
        \begin{equation*}
            \left\|e^{t(\gamma \mathcal{L}+\mathcal{H}+\mathcal{K})}(\rho)-e^{t(\mathcal{K}_{\operatorname{proj}}+\mathcal{H}_{\operatorname{proj}})}(\rho)\right\|_1\le \frac{tC+C'}{\gamma}\,,    
        \end{equation*}
        where $ C:=\frac{4}{\eta}\Big(\big((C_K+1) \iota_1'+\iota_1\big)C_L+\big(\iota_2+\iota_2'+C_H+(|\mathcal{J}|+\iota_2')C_K\big) (C_H+|\mathcal{J}|C_K)\Big)$ and $C':=\frac{4}{\eta}(C_H+|\mathcal{J}|C_K)$. If $\mathcal{K}=0$, the constants simplify to $C=\frac{4}{\eta}\Big(\iota_1C_L+(\iota_2+C_H) C_H\Big)$ and $C'=\frac{4C_H}{\eta}$.
    \end{thm}

    \begin{proof}
        For the following proof, we denote $\rho_s^{\mathrm{proj}}:=e^{s(\cKproj+\cHproj)}(\rho)$. Note that, $\cKproj$ and $\cHproj$ are bounded operators, so that both define a quantum Markov semigroup \cite{Lindblad.1976}. Then,
        \begin{equation}\label{eq:general-adiabatic-thm-step1}
            \begin{aligned}
                e^{t(\gamma \mathcal{L}+\mathcal{H}+\cK)}(\rho)-\rho_t^{\mathrm{proj}}&=-\int_{0}^t \frac{\partial}{\partial s}\,e^{(t-s)(\gamma \mathcal{L}+\mathcal{H}+\cK)}(\rho_s^{\mathrm{proj}})\,ds\\
                &=\int_{0}^t \,e^{(t-s)(\gamma \mathcal{L}+\mathcal{H}+\cK)}(\gamma \mathcal{L}+\mathcal{H}-\mathcal{H}_{\mathrm{proj}}+\cK-\cKproj)(\rho_s^{\mathrm{proj}})\,ds\\
                &=\int_{0}^t \,e^{(t-s)(\gamma \mathcal{L}+\mathcal{H}+\cK)}(I-\mathcal{P})(\mathcal{H}+\mathcal{K})\mathcal{P}(\rho_s^{\mathrm{proj}})\,ds\,,
            \end{aligned}
        \end{equation}
        where we used that $\rho_s^{\mathrm{proj}}=\mathcal{P}(\rho_s^{\mathrm{proj}})$ with $\mathcal{P}=P\cdot P$, $\cL\mathcal{P}=0$ and $(\mathcal{H}_{\mathrm{proj}}+\cKproj)\mathcal{P}=\mathcal{P}(\mathcal{H}+\cK)\mathcal{P}$. Note that by Assumption $(iv)$ and the uniform boundedness theorem $(\mathcal{H}+\mathcal{K})\mathcal{P}$ is well-defined for any input. In particular, for any $\rho\in\cT_f$
        \begin{equation*}
            \begin{aligned}
                (I-\mathcal{P})\cH\mathcal{P}(\rho)&=-i\bigl((I-P)HP\rho P-P\rho PH(I-P)\bigr)\\
                (I-\mathcal{P})\cK\mathcal{P}(\rho)&=(I-P)GP \rho P+P\rho PG (I-P)\,,
            \end{aligned}
        \end{equation*}
        with $G=-\frac{1}{2}\sum_j K_j^\dagger K_j$ because of 
        \begin{align*}
            (I-\mathcal{P})\cK\mathcal{P}(\rho)&=\sum_{j\in\mathcal{J}}\, (I-\mathcal{P})\Bigl(K_j\mathcal{P}(\rho)K_j^\dagger-\frac{1}{2}\big\{K_j^\dagger K_j,\mathcal{P}(\rho)\big\}\Bigr)\\
            &\overset{(1)}{=}-\frac{1}{2}\sum_{j\in\mathcal{J}}\, (I-\mathcal{P})\Bigl(\big\{K_j^\dagger K_j,\mathcal{P}(\rho)\big\}\Bigr)\\
            &= (I-P)G P\rho P+P\rho PG (I-P)\,,
        \end{align*}
        where we used that $(I-P)K_jP=0$ in $(1)$. To characterize the above found structure, we define the following projection
        \begin{equation}\label{eq:def-sym-projection}
            \widetilde{\mathcal{P}}=(I-P)\cdot P+P\cdot (I-P)
        \end{equation}
        and $X_s=(I-P){(-iH+}G {)}P\rho_s^{\operatorname{proj}}P$  so that \eqref{eq:general-adiabatic-thm-step1} can be rewritten as
        \begin{equation}\label{eq:general-adiabatic-thm-step2}
            \begin{aligned}
                e^{t(\gamma \mathcal{L}+\mathcal{H}+\cK)}(\rho)-\rho_t^{\mathrm{proj}}&=\int_{0}^t \,e^{(t-s)(\gamma \mathcal{L}+\mathcal{H}+\cK)}\widetilde{\mathcal{P}}(X_s+X^\dagger_s)\,ds\,.
            \end{aligned}
        \end{equation}
        Note that $GP$ as well as $HP$ is a bounded operator so that $X_s$ is a bounded operator and the adjoint well-defined. Next, we observe that
        \begin{equation*}
            \mathcal{L}\widetilde{\mathcal{P}}(X)=-\frac{1}{2}\sum_i\{(I-P)L_i^\dagger L_i(I-P),\widetilde{\mathcal{P}}(X)\}
        \end{equation*}
        and $\mathcal{L}\widetilde{\mathcal{P}}(X)=\widetilde{\mathcal{P}}\mathcal{L}\widetilde{\mathcal{P}}(X)$ for all $X\in\cT_f$. Due to these two properties, we show that the above identity extends to the semigroups. First, we remark that  $L'=\sum_i(I-P)L_i^\dagger L_i(I-P)$ is a positive self-adjoint operator by \cite[Proposition 3.18]{schmudgen2012unbounded}. Hence, by \cite[Proposition 6.14]{schmudgen2012unbounded}, $-L'$ generates a contraction semigroup. Next, we use Duhamel's formula again to show
        \begin{equation*}
            \begin{aligned}
                e^{t\cL}\widetilde{\mathcal{P}}&(X)-e^{-\frac{t}{2}\{L',\cdot\}}\widetilde{\mathcal{P}}(X)=\int_0^te^{s\cL}(\mathcal{L}+\frac{1}{2}\{L',\cdot\})\widetilde{\mathcal{P}}e^{-\frac{t-s}{2}\{L',\cdot\}}\widetilde{\mathcal{P}}(X)\,ds=0\,.
            \end{aligned}
        \end{equation*}
        By adding and subtracting this equation from \eqref{eq:general-adiabatic-thm-step2}, we get
        \begin{equation}\label{eq:general-adiabatic-thm-step3}
            \begin{aligned}
                e^{t(\gamma \mathcal{L}+\mathcal{H}+\cK)}(\rho)-\rho_t^{\mathrm{proj}}&=\int_{0}^t \,\Bigl(e^{(t-s)(\gamma \mathcal{L}+\mathcal{H}+\cK)}-e^{(t-s)\gamma \cL}+e^{(t-s)\gamma \cL}\Bigr)\widetilde{\mathcal{P}}(X_s+X^\dagger_s)\,ds\\
                &=\int_{0}^t\int_0^s \,e^{(s-u)(\gamma \mathcal{L}+\mathcal{H}+\cK)}(\mathcal{H}+\mathcal{K})e^{-\frac{u\gamma}{2}\{L',\cdot\}}\widetilde{\mathcal{P}}(X_{t-s}+X^\dagger_{t-s})\,duds\\
                &\qquad+\int_{0}^t \,e^{-\frac{\gamma(t-s)}{2}\{L',\cdot\}}\widetilde{\mathcal{P}}(X_s+X^\dagger_s)\,ds\,.
            \end{aligned}
        \end{equation}
        Next, we take care of each of these integrals separately. First of all, 
        \begin{equation*}
            e^{-\frac{\gamma t}{2}\{L',\cdot\}}\widetilde{\mathcal{P}}(X_s+X^\dagger_s)=e^{-\frac{\gamma t}{2}L'}X_s+X_s^\dagger e^{-\frac{\gamma t}{2}L'}\,.
        \end{equation*}
        Then, by assumption (i), we have 
        \begin{equation}\label{eqint1}
            \begin{aligned}
                \left\|\int_0^t\,e^{-\frac{1}{2}(t-s)\gamma L'}X_s\,ds\right\|_1&\le \int_0^t e^{-\frac{\gamma(t-s)\eta }{2}}\|X_s\|_1 ds \\
                &\le \frac{2\|(I-P)(iH-G) P\|_\infty }{\eta \gamma}\le \frac{2(C_H+|\mathcal{J}|C_K)}{\eta\gamma}
            \end{aligned}
        \end{equation}
        and similarly
        \begin{align}
            \left\|\int_0^tX_s^\dagger e^{-\frac{1}{2}(t-s)\gamma L'}\,ds\right\|_1\le \frac{2(C_H+|\mathcal{J}|C_K)}{\eta\gamma}\,.\label{eqint2}
        \end{align}
        It remains to control the trace norm of the first integral in \eqref{eq:general-adiabatic-thm-step3}: after denoting $B:=(I-P)(iH-G)P$, and since $\|X_s\|_1\le C_H+|\mathcal{J}|C_K$, we get
        \begin{align}
            \|(\mathcal{H}&+\mathcal{K})e^{-\frac{u\gamma}{2}\{L',\cdot\}}\widetilde{\mathcal{P}}(X_{t-s}+X^\dagger_{t-s})\|_1\nonumber\\
            &\le \|(\cH+\cK) (e^{-\frac{\gamma u}{2}L'} X_{t-s})\|_1+\|(\cH+\cK) (X^\dagger_{t-s}e^{-\frac{\gamma u}{2}L'})\|_1\nonumber \\
            &\le 2 \|He^{-\frac{\gamma u}{2}L'}B\|_\infty  +2(C_H+|\mathcal{J}|C_K)e^{-\frac{\gamma \eta u}{2}}\,\|PH\|_\infty\nonumber\\
            &\qquad+\sum_{j\in\mathcal{J}}\Big\|K_j \big(e^{-\frac{\gamma u}{2}L'}X_{t-s}\big)K_j^\dagger-\frac{1}{2}\big\{K_j^\dagger K_j, e^{-\frac{\gamma u}{2}L'}X_{t-s}\big\}\Big\|_1\nonumber\\
            &\qquad+\sum_{j\in\mathcal{J}}\Big\|K_j\big(X_{t-s}^\dagger e^{-\frac{\gamma u}{2}L'}\big)K_j^\dagger-\frac{1}{2}\big\{K_j^\dagger K_j,X_{t-s}^\dagger e^{-\frac{\gamma u}{2}L'}\big\}\Big\|_1\nonumber\\
            &\le 2 \Big(\iota_1\|L'e^{-\frac{\gamma u}{2}L'}B\|_\infty+\iota_2 (C_H+|\mathcal{J}|C_K)e^{-\frac{\gamma u\eta}{2}}\Big)  +2(C_H+|\mathcal{J}|C_K)e^{-\frac{\gamma \eta u}{2}}\,C_H\nonumber\\
            &\qquad+\sum_{j\in\mathcal{J}}\Big\|K_j \big(e^{-\frac{\gamma u}{2}L'}X_{t-s}\big)K_j^\dagger-\frac{1}{2}\big\{K_j^\dagger K_j, e^{-\frac{\gamma u}{2}L'}X_{t-s}\big\}\Big\|_1\nonumber\\
            &\qquad+\sum_{j\in\mathcal{J}}\Big\|K_j\big(X_{t-s}^\dagger e^{-\frac{\gamma u}{2}L'}\big)K_j^\dagger-\frac{1}{2}\big\{K_j^\dagger K_j,X_{t-s}^\dagger e^{-\frac{\gamma u}{2}L'}\big\}\Big\|_1\nonumber
        \end{align}
        where we also used Assumptions (ii) and (iv) above. Similarly, each of the summands in the above two sums can be controlled as follows:
        \begin{align*}
            &\sum_j\Big\|K_j \big(e^{-\frac{\gamma u}{2}L'}X_{t-s}\big)K_j^\dagger-\frac{1}{2}\big\{K_j^\dagger K_j, e^{-\frac{\gamma u}{2}L'}X_{t-s}\big\}\Big\|_1\\
            &\, \le \sum_j\left(\|K_jP\|_\infty\Big\|K_je^{-\frac{\gamma u}{2}L'} B\Big\|_\infty+\Big\|K_j^\dagger K_j e^{-\frac{\gamma u}{2}L'}B\Big\|_\infty\!\!\!\!+ (C_H+|\mathcal{J}|C_K)e^{-\frac{\gamma \eta u}{2}}\|P K_j^\dagger K_j\|_\infty\right)\\
            &\,\le \big(C_K+1\big) \Big(\iota_1'\Big\|L'e^{-\frac{\gamma u}{2}L'}B\Big\|_\infty\!\!\!\!\!\!+(C_H+|\mathcal{J}|C_K)\iota_2'e^{-\frac{\gamma u \eta}{2}}\Big)\!+\!(C_H+|\mathcal{J}|C_K)|\mathcal{J}| e^{-\frac{\gamma \eta u}{2}} C_K
        \end{align*}
        where we used Assumptions (iii) and (iv) again. Combining the above bounds and using that $\|L'B\|_\infty\le C_L$, we get 
        \begin{align*}
            \|&(\mathcal{H}+\mathcal{K})e^{-\frac{u\gamma}{2}\{L',\cdot\}}\widetilde{\mathcal{P}}(X_{t-s}+X^\dagger_{t-s})\|_1\\
            &\le 2e^{-\frac{\gamma\eta u}{2}}\Bigl(\iota_1C_L+(\iota_2+C_H+|\mathcal{J}|C_K) (C_H+|\mathcal{J}|C_K) \\
            &\qquad\qquad\qquad\qquad+\big(C_K+1\big) \Big(\iota_1' C_L\!\!+(C_H+|\mathcal{J}|C_K)\iota_2'\Big)\Bigr)\\
            &\eqcolon\frac{\eta C}{2}e^{-\frac{\gamma\eta u}{2}}\,.
        \end{align*}
        We can now integrate over $u$ and $s$ in \eqref{eq:general-adiabatic-thm-step3} to get 
        \begin{align*}
            &\left\|\int_{0}^t\int_0^s \,e^{(s-u)(\gamma \mathcal{L}+\mathcal{H}+\cK)}(\mathcal{H}+\mathcal{K})e^{-\frac{u\gamma}{2}\{L',\cdot\}}\widetilde{\mathcal{P}}(X_{t-s}+X^\dagger_{t-s})\,duds\right\|_1\\
            &\qquad\le \int_{0}^t\int_0^s \,\|(\mathcal{H}+\mathcal{K})e^{-\frac{u\gamma}{2}\{L',\cdot\}}\widetilde{\mathcal{P}}(X_{t-s}+X^\dagger_{t-s})\|_1\,duds\\
            &\qquad \le  \frac{\eta C}{2} \int_{0}^t\int_0^se^{-\frac{\gamma \eta u}{2}}duds\\
            &\qquad \le \frac{tC}{\gamma}\,.
        \end{align*}
        The result follows after combining this bound with \eqref{eqint1} and \eqref{eqint2}.
    \end{proof}
    
\subsection{Bosonic cat code}
\label{subsec:bosonic_cat}
    Next, we briefly discuss the bosonic cat code as introduced in \cite{Mirrahimi.2014,Azouit.2016,Guillaud.2019,Guillaud.2023} and others. For that, we define the annihilation operator $b$ and the creation operator $b^\dagger$ on the separable Hilbert space $\mathscr{H} \coloneqq L^2(\mathbb{R})$ of square-integrable functions with orthonormal basis denoted by $\{\ket{n}\}_{n \in \mathbb{N}}$ for $\mathbb{N} = \{0,1,\dots\}$, known as the Fock basis. The annihilation and creation operators are characterized by
    \begin{equation*}
        b\ket{n} = \sqrt{n} \ket{n-1}, \qquad b^\dagger\ket{n} = \sqrt{n+1} \ket{n+1}\,.
    \end{equation*}
    In particular, they obey the canonical commutation relation:
    \begin{equation}\label{eq:ccr-onemode}
        [b, b^\dagger] = I.
    \end{equation}
    Furthermore, we define the number operator as $N \coloneqq b^\dagger b$ and refer for more details to \cite[Sec.~12]{Holevo.2012}. The bosonic cat code is constructed on these bosonic systems and is introduced in the following discussion. The cat code relies on $r$-photon driven dissipation, defined by the master equation:
    \begin{equation}\label{eq:def-photon-dissipation}
        \frac{d}{dt} \rho(t) = \mathcal{L}[L_r](\rho(t)) \qquad \text{with} \qquad L_r \coloneqq b^r - \alpha^r,
    \end{equation}
    with $\alpha \in \mathbb{C}$ and $\mathcal{L}[L](\rho) = L \rho L^\dagger - \frac{1}{2} \{L^\dagger L, \rho\}$. Motivated by a driven damped harmonic oscillator, the $r$-photon driven dissipation is theoretically constructed using a harmonic oscillator with multi-photon drive and dissipation. A more detailed description can be found in Section \ref{sec:experimental_realization} and Appendix  \ref{app:experimental_realizations}. 
    For $r = 2$, a detailed construction of the above dynamics can be found in \cite[Sec.~3]{Guillaud.2023}. These master equations admit unique solution operators, as established in \cite{Azouit.2016, Gondolfetal2023}. One of the major properties of the $r$-photon driven dissipation is that it admits the following invariant subspace --- the code space:
    \begin{equation}\label{eq:codespace-r-photon-diss}
        \mathcal{C}_r(\alpha) \coloneqq \mathrm{span}\left\{\ket{\alpha_1}\bra{\alpha_2} : \alpha_1, \alpha_2 \in \left\{\alpha e^{\frac{i2\pi j}{r}} : j \in \{0, \ldots, r-1\}\right\}\right\}.
    \end{equation}
    The property that the $r$-photon dissipation rapidly converges to the code space is fundamental for continuous-time error correction codes. However, this was proven only in the following sense \cite{Azouit.2016}:
    \begin{equation*}
        \mathrm{tr}\big[L_r e^{t\mathcal{L}[L_r]}(\rho)L_r^\dagger\big] \leq e^{-tr!} \mathrm{tr}[L_r\rho L_r^\dagger]\,.
    \end{equation*}
    Using a similar approach to Theorem \ref{thm:general-adiabatic-limit}, we prove the following convergence in trace norm:

    \begin{prop}\label{prop:cat-convergence}
        Let $r\geq0$, $L_r$ be the jump operator defining the $r$-photon dissipation $\mathcal{L}[L_r]$ \eqref{eq:def-photon-dissipation} and let $P$ be the orthogonal projection onto $\mathcal{C}_r(\alpha)$ \eqref{eq:codespace-r-photon-diss}. Then, for all $t\geq0$
        \begin{equation*}
            \begin{aligned}
                \|(e^{t\mathcal{L}[L_r]} - \mathcal{P})(\rho)\|_1 &\leq 3e^{- \frac{t}{2} r!}\Bigl( \frac{2}{r!}\|L_r \rho L_r^\dagger\|_1 + \|\rho\|_1 \Bigr), &  & \rho \in \mathrm{dom}(L\cdot L^\dagger), \\
                &\leq 6e^{- \frac{t}{2} r!}\Bigl( (1 + \frac{|\alpha|^r}{\sqrt{r!}})^2\|\widetilde{N}^{r/2} \rho \widetilde{N}^{r/2}\|_1 + \|\rho\|_1 \Bigr), & & \rho \in \mathrm{dom}(\widetilde{N}^{r/2}\cdot \widetilde{N}^{r/2}),
            \end{aligned}
        \end{equation*}
        where $\widetilde{N} = N + I$ and $\mathcal{P} = P \cdot P$.
    \end{prop}
    \begin{proof}
        As stated in the result, we assume that $\rho \in \mathrm{dom}(L\cdot L^\dagger)$ or $\rho \in \mathrm{dom}(\widetilde{N}^{r/2} \cdot \widetilde{N}^{r/2})$. By \cite[Prop.~3.18 and 6.14]{schmudgen2012unbounded}, the operator $-L' \coloneqq (I - P)L_r^\dagger L_r (I - P)$ generates a contractive semigroup on $\mathscr{H}$, so that
        \begin{align*}
            \|(e^{t\mathcal{L}} - \mathcal{P})(\rho)\|_1 
            &= \|(I - \mathcal{P})\circ e^{t\mathcal{L}} (\rho)\|_1 \\
            &\leq \|(I - \mathcal{P})\circ(e^{t\mathcal{L}} - e^{-\frac{t}{2}\{L',\cdot\}})(\rho)\|_1 
            + \|(I - \mathcal{P})\circ e^{-\frac{t}{2}\{L',\cdot\}}(\rho)\|_1 \\
            &\leq t\bigl\|\int_0^1 (I - \mathcal{P})\circ e^{-\tau \frac{t}{2} \{L',\cdot\}} 
            \Bigl(L_r (e^{(1-\tau)t\mathcal{L}}(\rho))L_r^\dagger\Bigr) d\tau\bigr\|_1 \\
            &\qquad\quad+ \|(I - \mathcal{P})\circ e^{-\frac{t}{2} \{L',\cdot\}}(\rho)\|_1 \\
            &\leq t\int_0^1 \|(I - \mathcal{P})\circ e^{-\tau \frac{t}{2} \{L',\cdot\}}\|_{1\rightarrow 1} \|L_r (e^{(1-\tau)t\mathcal{L}} (\rho)) L_r^\dagger\|_1 d\tau \\
            &\quad\qquad + \|(I - \mathcal{P})\circ e^{-\frac{t}{2} \{L',\cdot\}}(\rho)\|_1\,,
        \end{align*}
        where we used that $\cL-\frac{1}{2}\{L',\cdot\}=L_r\cdot L^\dagger_r$. In the next step, we use the following two inequalities: First, for every $X\in\mathcal{T}_1$, we have
        \begin{equation*}
            \begin{aligned}
                \|&(I - \mathcal{P})\circ e^{-\frac{s}{2} \{L',\cdot\}}(X)\|_{1\rightarrow 1}\\
                &=\| Pe^{-\frac{s}{2} L'} X e^{-\frac{s}{2} L'}(I-P)+(I-P)e^{-\frac{s}{2} L'} X e^{-\frac{s}{2} L'}P+(I-P)e^{-\frac{s}{2} L'} X e^{-\frac{s}{2} L'}(I-P)\|_{1\rightarrow 1}\\
                &\leq \Bigl(\|e^{-\frac{s}{2} L'}(I - P)\|_\infty + \|(I - P)e^{-\frac{s}{2} L'}\|_\infty + \|e^{-\frac{s}{2} L'}(I - P)\|_\infty   \|(I - P)e^{-\frac{s}{2} L'}\|_\infty \Bigr) \|X\|_1\\
                &\leq 3e^{-\frac{s}{2}\eta}\,,
            \end{aligned}
        \end{equation*}
        where we used assumption $(i)$, which is satisfied due to Lemma \ref{lemmaetaAB} combined with Lemma \ref{lem:pseudoinverse-bound-photon-dissipation} with $L'\geq \eta (I-P)$ and $\eta\geq r!$. Second, we use the bound given in \cite{Azouit.2016}, which is
        \begin{equation*}
            \mathrm{tr}\big[L_r e^{t\mathcal{L}[L_r]}(\rho)L_r^\dagger\big] \leq e^{-tr!} \mathrm{tr}[L_r\rho L_r^\dagger]\,.
        \end{equation*}
        Together these two bounds show
        \begin{align*}
            \|(e^{t\mathcal{L}} - \mathcal{P})(\rho)\|_1 
            &\leq 3t\int_0^1 e^{-\frac{\tau t}{2} r!} e^{-(1 - \tau) t r!}d\tau \, \|L_r \rho L_r^\dagger\|_1 d\tau + 3 e^{-\frac{t}{2} r!} \|\rho\|_1 \\
            &\leq 3e^{- \frac{t}{2} r!}\Bigl(t\int_0^1 e^{-(1-\tau) t\frac{r!}{2}}d\tau\|L_r \rho L_r^\dagger\|_1 + \|\rho\|_1 \Bigr) \\
            &\leq 3e^{- \frac{t}{2} r!} \Bigl(\frac{2}{r!} (1-e^{- \frac{t}{2} r!})\|L_r \rho L_r^\dagger\|_1 + \|\rho\|_1 \Bigr) \\
            &\leq 3e^{- \frac{t}{2} r!}\Bigl( \frac{2}{r!}(\sqrt{r!} + |\alpha|^r)^2\|\widetilde{N}^{r/2} \rho \widetilde{N}^{r/2}\|_1 + \|\rho\|_1 \Bigr).
        \end{align*}
        In the last line, we used Hölder's inequality combined with the relative boundedness
        \begin{equation*}
            \|\widetilde{N}^{-r/2}L_r\|_{\infty}=\|L_r^\dagger\widetilde{N}^{-r/2}\|_{\infty}\leq(1 + |\alpha|^r)
        \end{equation*}
        which is proven by following the strategy in \cite[Lemma B.1]{Gondolfetal2023}, i.e. 
        \begin{equation*}
			\begin{aligned}
				\|b^{r\dagger}\widetilde{N}^{-r/2}\ket{\varphi}\|^2&=\bra{\varphi}(N+I)^{-r/2} b^rb^{r\dagger}(N+I)^{-r/2}\ket{\varphi}\\
                &=\bra{\varphi}(N+I)^{-r}(N+I)\cdots(N+r)\ket{\varphi}\\
				&\leq r!\|\varphi\|^2
			\end{aligned}
		\end{equation*}      
        so that $\|L_r^\dagger\widetilde{N}^{-r/2}\|_{\infty}\leq \|b^{r\dagger}\widetilde{N}^{-r/2}\|_{\infty}+|\alpha|^r\leq \sqrt{r!}+|\alpha|^r$, which completes the proof.
    \end{proof}

    Next, add a Hamiltonian to the $r$-photon dissipation and use the adiabatic Theorem \ref{thm:general-adiabatic-limit} to prove that for any Hamiltonian \begin{equation}\label{eq:single-mode-Hamiltonian}
        H = \sum_{j+j'\leq d} h_{j,j'}(b^\dagger)^j b^{j'},
    \end{equation}
    where $\overline{h}_{j,j'} = h_{j',j}$ and $|h_{j,j'}| \leq 1$, the time evolution it generates converges to its effective dynamics if it is dominated by a $r$-photon driven dissipation with $ d/2 \leq r$:
    \begin{prop} \label{prop:adiabatic-cat}
        Let $d \in \mathbb{N}$, $H$ be as defined in \eqref{eq:single-mode-Hamiltonian}, $d/2\leq r\in\N$, $\alpha\in\C$, $L_r$ be the jump operator defining the $r$-photon dissipation $\mathcal{L}[L_r]$ \eqref{eq:def-photon-dissipation}, and let $P$ be the orthogonal projection onto $\mathcal{C}_r(\alpha)$ \eqref{eq:codespace-r-photon-diss}. Then, for all $t \geq 0$ and $\rho \in \mathcal{C}_r(\alpha)$
        \begin{equation*}
            \|e^{-it[H, \cdot] + t\gamma\mathcal{L}[b^r - \alpha^r]}(\rho) - e^{-it[PHP, P \cdot P]}(\rho)\|_1 \leq\frac{tC+C'}{\gamma}\,
        \end{equation*}
        for constants $C, C'\geq0$ (see \eqref{eq:constant-cat-perturbation}).
    \end{prop}
    \begin{proof}        
        The statement is a direct consequence of Theorem \ref{thm:general-adiabatic-limit}, which is based on the assumptions $(i)-(v)$. For that reason, we specify the constants $\eta, \iota_1, \iota_2, \iota_1', \iota_2', C_H, C_K, C_L, |\mathcal{J}|$ in the following. Due to the structure of the generator, it holds that $C_K=\iota_1'=\iota_2'=0$ and $|\mathcal{J}|=1$. Moreover,
        \begin{itemize}
            \item[(i)] in Lemma \ref{lemmaetaAB} with Lemma \ref{lem:pseudoinverse-bound-photon-dissipation} it is proven that $\eta\geq r!$;
            \item[(ii)]  Corollary \ref{cor:rel-bounded-gate} proves that
                \begin{equation*}
                    \|H\ket{\psi}\|\leq d^{d}\begin{cases}
                    \frac{1}{\sqrt{1-\delta}}\|L^\dagger L\ket{\psi}\|+\sqrt{\frac{c_r}{2r(1-\delta)}\Bigl(\frac{2r-1}{2r}\frac{c_r}{\delta}\Bigr)^{2r-1}}\|\ket{\psi}\| \quad& \text{for}\qquad r\geq2\\
                    \frac{1}{\sqrt{1-\delta}}\|L^\dagger L\ket{\psi}\|+\sqrt{\frac{c_1}{4(1-\delta)}\left(\frac{3}{4}\frac{c_1}{\delta}\right)^3}\|\ket{\psi}\|\quad& \text{for}\qquad r=1
                \end{cases}
                \end{equation*}
                for any $\delta\in(0,1)$, state $|\psi\rangle\in \operatorname{dom}(L^\dagger L)$ and $c_r=(r+1)r+4|\alpha|^rr!^{3/2}+2\sqrt{(2r)!}|\alpha|^{2r}$ so that for $\delta=1/2$
                \begin{equation*}
                    \begin{aligned}
                        \iota_1=\sqrt{2}d^d\qquad\text{and}\qquad\iota_2= d^d\sqrt{\frac{c_r}{r}(2 c_r)^{2r-1+2\cdot1_{r=1}}}\,;
                    \end{aligned}
                \end{equation*}
            \item[(iv)] $\mathcal{C}_r(\alpha) \subset \mathrm{dom}(H), \mathrm{dom}(L'H)$, and $H P$, $L'H P$ are finite-dimensional operators when restricted to $\mathcal{C}_r(\alpha)$. Since all finite-dimensional operators are bounded, this implies the existence of the constants $C_H$ and $C_L$ such that $\|H P\|_\infty \leq C_H$ and $\|L'H P\|_\infty \leq C_L$. An explicit construction of this result can be found by following the approach presented in Appendix \ref{subsec:single-mode-adiabatic-thm}.
        \end{itemize}
        These define the constants 
        \begin{equation}\label{eq:constant-cat-perturbation}
            C=\frac{4}{\eta}\Big(\iota_1C_L+(\iota_2+C_H) C_H\Big)\quad \text{and} \quad C'=\frac{4C_H}{\eta}    
        \end{equation}
        and finish the proof.
    \end{proof}
    Next, we consider the cases $r =2$ which is of particular interest (see \cite{Guillaud.2019,Guillaud.2023}). Here, the projection onto the code space can be defined using the orthonormal Schrödinger cat states, which also define logical qubits:
    \begin{equation*}
        \ket{CAT_{\alpha}^+} \coloneqq \frac{1}{\sqrt{2(1+e^{-2|\alpha|^2})}}\bigl(\ket{\alpha} + \ket{-\alpha}\bigr), \quad 
        \ket{CAT_{\alpha}^-} \coloneqq \frac{1}{\sqrt{2(1-e^{-2|\alpha|^2})}}\bigl(\ket{\alpha} - \ket{-\alpha}\bigr)\,.
    \end{equation*}
    Then, the rotation around the $x$-axis, described in \cite{Mirrahimi.2014} and experimentally realized in \cite{Touzard.2018}, is 
    \begin{equation}\label{eq:rotation}
        X(\theta) = \cos\Bigl(\frac{\theta}{2}\Bigr)(P^+_{\alpha} + P^-_{\alpha}) + i\sin\Bigl(\frac{\theta}{2}\Bigr)X_{\alpha},
    \end{equation}
    where $P^+_\alpha = \ket{CAT^+_{\alpha}}\bra{CAT^+_{\alpha}}$, $P^-_\alpha = \ket{CAT^-_{\alpha}}\bra{CAT^-_{\alpha}}$, $P_\alpha = P^+_\alpha + P^-_\alpha$, and $X_\alpha = \ket{CAT^+_{\alpha}}\bra{CAT^-_{\alpha}} + \ket{CAT^-_{\alpha}}\bra{CAT^+_{\alpha}}$. One way to construct this gate is via the driving Hamiltonian $H = b + b^\dagger$, combined with the adiabatic convergence of the dynamics. For $\rho \in \mathcal{C}_2(\alpha)$, we have
    \begin{equation*}
        \|e^{-it[H,\cdot] + \gamma t \mathcal{L}[b^2 - \alpha^2]}(\rho) - e^{-it[P_\alpha H P_{\alpha}, P_{\alpha} \cdot P_{\alpha}]} (\rho)\|_{1} = \mathcal{O}\Bigl(\frac{t+1}{\gamma}\Bigr),
    \end{equation*}
    where
    \begin{equation*}
        P_\alpha H P_\alpha = (\alpha + \overline\alpha) X_{\alpha}.
    \end{equation*}
    Thus, the effective dynamics generate the gate:
    \begin{equation*}
        e^{it P_\alpha H P_\alpha} P_\alpha = \cos(t(\alpha + \overline\alpha)) P_\alpha + i \sin(t(\alpha + \overline\alpha)) X_{\alpha}.
    \end{equation*}

\subsection{Modified photon-dissipation and decoupling}\label{subsec:mod-photon-dissipation}
    In the following appendix, we consider a modified family of photon dissipation processes. Here, the invariant subspace is adjusted in such a way that the effective dynamics of a Hamiltonian still retain all the information needed to characterize the original Hamiltonian. The method for extracting this information is presented in Section \ref{sec:main-learning} and Appendix~\ref{sec:learning}, achieving the Heisenberg limit:
    \begin{equation}\label{eq:def-modified-dissipation}
        \mathcal{L}[L_{r,\alpha}](\cdot) = L_{r,\alpha} \cdot L_{r,\alpha}^\dagger - \frac{1}{2} \{L_{r,\alpha}^\dagger L_{r,\alpha}, \cdot\} \qquad \text{with} \qquad L_{r,\alpha} = b^r (b - \alpha),
    \end{equation}
    for $\alpha \in \mathbb{C}$ and $r \in \mathbb{N}$.  The above single-mode case naturally extends to the multi-mode case by summing over all modes. However, we first consider the single-mode case.

    \subsubsection{Single-mode dynamics}\label{subsec:single-mode-adiabatic-thm}
        In this appendix, we prove Proposition \ref{prop-main:single-modes-adiabatic-limit}, restated as Proposition \ref{prop:single-modes-adiabatic-limit} below for convenience. In what follows, we make the notation less cumbersome by writing $L_1=L_{1,\alpha}$ and $L_r=L_{r,\alpha}$. 

        \begin{prop}\label{prop:single-modes-adiabatic-limit}
            Let $d \in \mathbb{N}$, $H$ be as defined in \eqref{eq:single-mode-Hamiltonian}, $r=\lceil\frac{d}{2}\rceil-1$, $\alpha\in\C$, and let $P$ be the orthogonal projection onto $\mathrm{span}(\ket{0},\ket{\alpha})$. Then, for all $t \geq 0$ and $\rho=P\rho P$ 
            \begin{equation*}
                \|e^{t\cH + t\gamma\mathcal{L}[b^r(b-\alpha)]+t\gamma\mathcal{L}[b(b-\alpha)]}(\rho) - e^{-it[PHP, P \cdot P]}(\rho)\|_1 \le \frac{tC+C'}{\gamma}
            \end{equation*}
            for constants $C=C_d(|\alpha|^{4d+2}+1)$ and $C'=C_d'(|\alpha|^d+1)$ for $C_d$ and $C_d'$ depending on $d$. 
        \end{prop}
        \begin{proof}
            To prove the above statement, we have to verify the assumptions of Theorem \ref{thm:general-adiabatic-limit}. For $\eta$, we establish the following bound for $L'=L_1^\dagger L_1 + L_r^\dagger L_r$ and $L_r=b^r(b-\alpha)$
            \begin{align*}
                (I - P)L'(I - P)\geq (I - P) b^\dagger (b^\dagger - \overline{\alpha}) b (b - \alpha) (I - P)\geq (I - P),
            \end{align*}
            where the last bound was derived in Lemmas \ref{lemmaetaAB} and \ref{lempseudoinv}. For $C_H$, we need to bound $\|H P\|_\infty$. Since $P$ is the projection operator onto $\ker(b(b - \alpha)) = \mathrm{span}\{\ket{0}, \ket{\alpha}\}$, an orthonormal basis can be constructed as $\{\ket{0}, \ket{\perp}\}$, where
            \begin{equation*}
                \ket{\perp} = \frac{\sum_{k=1}^\infty \frac{\alpha^k}{\sqrt{k!}} \ket{k}}{\sqrt{\sum_{k=1}^\infty \frac{|\alpha|^{2k}}{k!}}} = \frac{1}{\sqrt{\sum_{k=1}^\infty \frac{|\alpha|^{2k}}{k!}}} \big(e^{|\alpha|^2/2} \ket{\alpha} - \ket{0} \big)\,.    
            \end{equation*}
            Then, 
            \begin{equation*}
                b \ket{\perp} = \frac{\alpha}{\sqrt{\sum_{k=1}^\infty \frac{|\alpha|^{2k}}{k!}}} \ket{0} + \alpha \ket{\perp}\,.
            \end{equation*}
            shows $\|b \ket{\perp} \| \leq \sqrt{1 + |\alpha|^2}$. Consequently, we obtain: $\|b P\|_\infty \leq \sqrt{1 + |\alpha|^2}$ and similarly $\|P b^\dagger\|_\infty \leq \sqrt{1 + |\alpha|^2}$. To upper bound $\|HP\|_\infty$, first note that 
            \begin{equation}\label{eq:norm-bound-H2}
                \|H P\|_\infty^2 = \|P H^2 P\|_\infty\,.
            \end{equation}
            Then, we rearrange the terms in $H^2$ in normal ordering (placing creation operators to the left and annihilation operators to the right) and place a $P$ in between the annihilation and creation operators using the property $P b P = bP$ and $P b^\dagger P = P b^\dagger$. This allows us to use Hölder's inequality iteratively and we obtain $\|P H^2 P\|_\infty \leq (\widetilde{C}_d)^2 (|\alpha|^{2d}+1)$, which implies $\|H P\|_\infty \leq \widetilde{C}_d (|\alpha|^d+1)$. Here, $\widetilde{C}_d$ is a non-negative constant depending only on $d$. 
            Next, we establish a bound on $\|L' H P\|_\infty$. We do this by splitting the norm as follows, using $Pb^\dagger (b^\dagger - \overline{\alpha})=0$:
            \begin{equation}\label{eq:bound-L'HP}
                \begin{aligned}
                    \|L' H P\|_\infty &= \Big\|\big[(I - P) (b^\dagger (b^\dagger - \overline{\alpha}) b (b - \alpha) (I - P) \\
                    &\qquad + (I - P) b^{\dagger r} (b^\dagger - \overline{\alpha}) b^r (b - \alpha) (I - P) \big] H P \Big\|_\infty \\
                    &\leq \|b^\dagger (b^\dagger - \overline{\alpha}) b (b - \alpha) H P\|_\infty + \|b^{\dagger r} (b^\dagger - \overline{\alpha}) b^r (b - \alpha) H P\|_\infty \\
                    &\leq C'_{r,d} (|\alpha|^{2(r+1)+d}+1),
                \end{aligned}
            \end{equation}
            for some constant $C'_{r,d}$ that depends only on $r$ and $d$. Note that for the above bound, we used again the trick in \eqref{eq:norm-bound-H2}, but for $L'HP$, which shows $\|L'HP\|_\infty=\sqrt{\|PHL'L'HP\|_\infty}$. Then, reordering the annihilation operators and creation operators via the canonical commutation relation and bounding the terms as above shows the above upper bound. For the choice $r=\lceil\frac{d}{2}\rceil-1$ we defined $\widetilde{C}_d'=C_{r,d}'$ For $\iota_1$ and $\iota_2$, we use Corollary \ref{cor:rel-bounded-learning}, which states that for $\delta=1/2$ and any state $\ket{\psi} \in \operatorname{dom}(L')$, we have
            \begin{equation*}
                \begin{aligned}
                    \|H\ket{\psi}\|&\leq \sqrt{2}d^{d}\biggl(\|L'\ket{\psi}\|+\sqrt{\frac{c}{(4r+4)}\left(\frac{4r+3}{4r+4}2c\right)^{4r+3}}\|\ket{\psi}\|\biggr)\\
                    &\leq \sqrt{2}d^{d}\biggl(\|L'\ket{\psi}\|+\frac{\left(2c\right)^{2(r+1)}}{\sqrt{4r+4}}\|\ket{\psi}\|\biggr)\,,
                \end{aligned}
            \end{equation*}
            with $c=(r+2)(r+1)+16|\alpha|(1+|\alpha|^2)+8\sqrt{2}|\alpha|^2$. Choosing $r=\lceil\frac{d}{2}\rceil-1 $, we obtain $\iota_1 =\sqrt{2}d^{d}$ and $\iota_2 = \widetilde{C}_d'' (|\alpha|^{3d+2}+1)$. Overall, we found the constants summarized in \eqref{eq:single-mode-constants} needed in order to apply Theorem \ref{thm:general-adiabatic-limit}.
            \begin{equation}\label{eq:single-mode-constants}
                \eta = 1, \quad C_H \leq \widetilde{C}_d(|\alpha|^d+1), \quad C_L\leq\widetilde{C}_d'(|\alpha|^{2d+1}+1),\quad\iota_1\leq\sqrt{2}d^{d},\quad\iota_2\leq\widetilde{C}_d'' (|\alpha|^{3d+2}+1)
            \end{equation}
            These determine the final constants
            \begin{equation}\label{eq:single-mode-constants-final}
                C=\frac{4}{\eta}\Big(\iota_1C_L+(\iota_2+C_H) C_H\Big)\leq C_d(|\alpha|^{4d+2}+1)\quad\text{and}\quad C'=\frac{4C_H}{\eta}\leq C_d'(|\alpha|^d+1)
            \end{equation}
            and finish the proof.
        \end{proof}

    \subsubsection{Few-modes dynamics}\label{subsec:few-mode-adiabatic-thm}
        Next, we extend the above single-mode system to an $m$-mode system with notation adopted from \cite{HuangTongFangSu2023learning}. We are in the setting of Definition \ref{def-main:bosonic-H}, thus Hamiltonians of the form
        \begin{align*}
            H = \sum_{a=1}^M E_a\qquad\text{with}\qquad E_a = \sum_{\substack{\mathbf{j},\mathbf{j}' \in \mathbb{N}^{ \mathfrak{k}}\\ \|\mathbf{j}+\mathbf{j}'\|_1\le d}} h^{(a)}_{\mathbf{j},\mathbf{j}'} (\mathbf{b}^\dagger)^{\mathbf{j}}\,\mathbf{b}^{\mathbf{j}'}\,,
        \end{align*}
        
     We define the modified photon dissipation as
        \begin{equation}\label{eq:dissipation-multi-mode}
            \mathcal{L} := \sum_{\ell=1}^m \mathcal{L}[b_\ell (b_\ell - \alpha_\ell)] + \mathcal{L}[b_\ell^r (b_\ell - \alpha_\ell)]   
        \end{equation}
        for some coefficients $\alpha \in \mathbb{C}^m$. As before, we then consider generators of the form
        \begin{align*}
            \mathcal{L} = \mathcal{H} + \gamma \mathcal{L}\,,
        \end{align*}
        in particular, we are interested in the adiabatic limit for large $\gamma$.
        \begin{prop}\label{prop:few-modes-adiabatic-limit}
            Let $d \in \mathbb{N}$, $H$ as defined in Definition \ref{def-main:bosonic-H}, $\mathcal{L}$ the sum over single-mode modified $r$-photon dissipation \eqref{eq:dissipation-multi-mode} with $r=\lceil\frac{d}{2}\rceil-1$, $\alpha\in\C^m$ and let $P=\bigotimes_{\ell=1}^m P_\ell$, where $P_\ell$ is the orthogonal projection onto $\operatorname{span}\{|0\rangle,|\alpha_\ell\rangle\}$. Then, for all $t \geq 0$ and $\rho=P\rho P$
            \begin{equation*}
                \|e^{-it\cH + \gamma t\mathcal{L}}(\rho) - e^{-it[PHP, P \cdot P]}(\rho)\|_1 \leq \frac{tC}{\gamma}+\frac{C'}{\gamma}\,
            \end{equation*}
            for constants $C=(m+M)M C_{d, \mathfrak{k}}(\|\alpha\|_\infty^{4d+2}+1)$ and $C'=M\,C_{d, \mathfrak{k}}'(\|\alpha\|_\infty^d+1)$ for $C_{d, \mathfrak{k}}$ and $C_{d, \mathfrak{k}}'$ depending on $d$ and $ \mathfrak{k}$.
        \end{prop}
        \begin{proof}
            To apply Theorem \ref{thm:general-adiabatic-limit}, we first note that the gap $\eta$ for the global system can be controlled by the ones of its single-mode constituents, which implies $\eta\geq 1$ by Lemma \ref{lemmaetaAB} and \ref{lempseudoinv}. 

            Next, we consider the upper bound on $\|HP\|_\infty$ and $\|L'HP\|_\infty$. By the product structure, these bounds mostly reduce to the single-mode case (see Proposition \ref{prop:single-modes-adiabatic-limit}) because the orthogonal projection onto the kernel of $\cL$ is given by $P=\bigotimes_{\ell=1}^m P_\ell$, where $P_\ell$ projects onto $\operatorname{span}\{|0\rangle,|\alpha_\ell\rangle\}$ and moreover
            \begin{equation}\label{eq:few-mode-H}
                \begin{aligned}
                    \|HP\|_\infty&\leq M \binom{ \mathfrak{k}+d}{d}^2\max_{\|\mathbf{j}+\mathbf{j}'\|_1\le d}\|(\mathbf{b}^\dagger)^{\mathbf{j}}\,\mathbf{b}^{\mathbf{j}'}P\|_\infty\\
                    &\leq M \binom{ \mathfrak{k}+d}{d}^2\max_{\|\mathbf{j}+\mathbf{j}'\|_1\le d}\prod_{\ell=1}^{ \mathfrak{k}}\|b_\ell^{\dagger {j}_\ell}b_\ell^{j_\ell'}P_\ell\|_\infty\\
                    &\leq M \binom{ \mathfrak{k}+d}{d}^2\widehat{C}_{d,\mathfrak{k}} (\|\alpha\|_\infty^d+1)\,,
                \end{aligned}
            \end{equation}
            where the upper bound on the sum 
            \begin{equation*}
                \sum_{\substack{\mathbf{j},\mathbf{j}' \in \mathbb{N}^{ \mathfrak{k}}\\ \|\mathbf{j}+\mathbf{j}'\|_1\le d}}1\leq\Bigl(\sum_{\substack{\mathbf{j} \in \mathbb{N}^{ \mathfrak{k}}\\ \|\mathbf{j}\|_1\le d}}1\Bigr)^2
            \end{equation*}
            can be interpreted as distributing $d$ plain balls to $ \mathfrak{k}+1$ labeled boxes giving the above binomial coefficient. This shows that there is a constant $\widetilde{C}_{d, \mathfrak{k}}\geq0$ depending on $d$ and $ \mathfrak{k}$ so that $\|HP\|_\infty\leq M \widetilde{C}_{d, \mathfrak{k}}(\|\alpha\|_\infty^d+1)$. Similarly, there is a constant $\widetilde{C}_{d, \mathfrak{k}}'\geq0$ depending on $d$ and $ \mathfrak{k}$ so that $\|L'HP\|_\infty\leq m M \widetilde{C}_{d, \mathfrak{k}}'(\|\alpha\|_\infty^{2d+1}+1)$ assuming the choice $r=\lceil\frac{d}{2}\rceil-1$. Here, $L'=\sum_{j=1}^m L_{r,j}^\dagger L_{r,j}+L_{1,j}^\dagger L_{1,j}$ for $L_{j,r}=b_j^r(b_j-\alpha_j)$. First, we apply the triangle inequality so that
            \begin{equation*}
                \begin{aligned}
                    \|L'HP\|_\infty\leq m M \max_{j\in[m], a\in[M]}\|(I-P)&L_{r,j}^\dagger L_{r,j} (I-P)E_a P\|_\infty\\
                    &+\|(I-P)L_{1,j}^\dagger L_{1,j} (I-P)E_a P\|_\infty\,.
                \end{aligned}
            \end{equation*}
            Then, we follow the reasoning of \eqref{eq:few-mode-H} to reduce the problem to a product of single-mode inequalities, for which we then apply  \eqref{eq:norm-bound-H2} and \eqref{eq:bound-L'HP}. Finally, we use Corollary \ref{cor:rel-bounded-learning} for bounding $\iota_1$ and $\iota_2$. The result directly states that for $\delta=1/2$ and any state $\ket{\psi} \in \operatorname{dom}(L')$, we have
            \begin{equation*}
                \|H\ket{\psi}\|\leq M( \mathfrak{k}+d-1)^{d}d!^{\mathfrak{k}/2-1}\biggl(\sqrt{2}\|L'\ket{\psi}\|+\sqrt{\frac{c}{2(r+1)}\left(2c\right)^{4r+3}}\|\ket{\psi}\|\biggr)
            \end{equation*}
            with $c=(r+2)(r+1)+16\|\alpha\|_\infty(1+\|\alpha\|_\infty^2)+8\sqrt{2}\|\alpha\|_\infty^2$ and note that $L_{r,j}P=0$ implies $(I-P)L_{r,j}^\dagger L_{r,j}(I-P)=L_{r,j}^\dagger L_{r,j}$. With $r=\lceil\frac{d}{2}\rceil-1 $, we obtain $\iota_1 =\sqrt{2}M( \mathfrak{k}+d-1)^{d}d!^{\mathfrak{k}/2-1}$ and $\iota_2 = \widetilde{C}_d''' (\|\alpha\|_\infty^{3d+2}+1)$. Overall, we found the constants summarized in \eqref{eq:few-mode-constants} needed to apply Theorem \ref{thm:general-adiabatic-limit}.
            \begin{equation}\label{eq:few-mode-constants}
                \begin{aligned}
                    \eta &\geq 1, &\quad C_H &\leq M \widetilde{C}_{d, \mathfrak{k}}(\|\alpha\|_\infty^d+1), &\quad C_L&\leq m M \widetilde{C}_{d, \mathfrak{k}}'(\|\alpha\|_\infty^{2d+1}+1),\\
                    &&\iota_1&\leq \sqrt{2}M( \mathfrak{k}+d-1)^{d}d!^{\mathfrak{k}/2-1},&\quad\iota_2&\leq\widetilde{C}_d''' (\|\alpha\|_\infty^{3d+2}+1)
                \end{aligned}
            \end{equation}
            These determine the final constants $C_{d, \mathfrak{k}}$ and $C_{d, \mathfrak{k}}'$
            \begin{equation}\label{eq:few-mode-constants-final}
                \begin{aligned}
                    C&=\frac{4}{\eta}\Big(\iota_1C_L+(\iota_2+C_H) C_H\Big)\leq (m+M)M C_{d, \mathfrak{k}}(\|\alpha\|_\infty^{4d+2}+1),\\
                    C'&=\frac{4C_H}{\eta}\leq M\,C_{d, \mathfrak{k}}'(\|\alpha\|_\infty^d+1)\,,
                \end{aligned}
            \end{equation}
            which finish the proof.
        \end{proof}

    \subsubsection{Multi-mode decoupling}\label{subsec:decoupling-mode-adiabatic-thm}
        In this appendix, we adapt the result in Proposition \ref{prop:few-modes-adiabatic-limit} such that the limit is a product of decoupled, local, Hamiltonian evolutions. Again, we consider Hamiltonian of the form introduced in Definition \ref{def-main:bosonic-H} given by
        \begin{align*}
            H = \sum_{a=1}^M E_a\qquad\text{with}\qquad E_a = \sum_{\substack{\mathbf{j},\mathbf{j}' \in \mathbb{N}^{ \mathfrak{k}}\\ \|\mathbf{j}+\mathbf{j}'\|_1\le d}} h^{(a)}_{\mathbf{j},\mathbf{j}'} (\mathbf{b}^\dagger)^{\mathbf{j}}\,\mathbf{b}^{\mathbf{j}'}\,,
        \end{align*}
        To describe the interaction pattern of the above Hamiltonian, we define the interacting cluster in the following:
        \begin{defn}[Interacting cluster]
        \label{defn:interacting_cluster}
            For each $a$, let $\operatorname{supp}(E_a)$ be the support of $E_a$, that is, the collection of modes in which $E_a$ acts non-trivially. From the set $\{\operatorname{supp}(E_a)\}$, we remove all $\operatorname{supp}(E_a)$ such that $\operatorname{supp}(E_a)\subset \operatorname{supp}(E_b)$ for some $b\in[M]$. The remaining $\operatorname{supp}(E_a)$'s form the set $\mathcal{V}$. Each element of $\mathcal{V}$ we call an interacting cluster. Clearly $|\mathcal{V}|\le M$.
        \end{defn}
        \begin{defn}[Cluster interaction graph]\label{def:interaction-graph}
            The cluster interaction graph $\mathcal{G}=(\mathcal{V},\mathcal{E})$ has interacting clusters as its vertices. The set of edges $\mathcal{E}$ is defined as follows: for each pair of interacting clusters $C$ and $C'$ ($C\ne C'$) in $\mathcal{V}$, $(C,C')\in\mathcal{E}$ if $C\cap C'\ne \emptyset$ or if there exists $C''\in\mathcal{V}$ such that $C\cap C''\ne \emptyset$ and $C'\cap C''\ne\emptyset$. 
        \end{defn}
        For a low-intersecting Hamiltonian, the degree of $\mathcal{G}$, $\operatorname{deg}(\mathcal{G})$, is upper bounded by a constant that is independent of the system size $m$. More precisely, $\operatorname{deg}(\mathcal{G})\le \mathfrak{d}^2$. In order to optimize the sample complexity, we aim at estimating different interacting clusters in parallel. For this, we need to color the graph $\mathcal{G}$ so that adjacent vertices have different colors. The number of colors $\chi$, known as chromatic number of the graph, satisfies $\chi\le \operatorname{deg}(\mathcal{G})+1=\mathcal{O}(1)$. Hence, $\mathcal{V}$ can be divided into a disjoint union 
        \begin{align}
            \mathcal{V}=\bigsqcup_{c=1}^\chi\mathcal{V}_c\,,
        \end{align}
        where no two adjacent vertices are in the same $\mathcal{V}_c$. In other words, for any $C$ and $C'$ in $\mathcal{V}_c$, $C\cap C'=\emptyset $, and for any $C''\in \mathcal{V}$, either $C\cap C''=\emptyset$ or $C'\cap C''=\emptyset$. Next, given a color $c\in[\chi]$, we denote
        \begin{align*}
            \mathcal{A}_c=\bigsqcup_{C\in\mathcal{V}_c}C\,.
        \end{align*}
        Next, we fix a color $c$ and we consider two dissipative generators, one supported on $[m]\backslash \mathcal{A}_c$ for decoupling interaction clusters of color $c$ from other modes, $\mathcal{L}^{c}_{\operatorname{dec}}$, and the other supported on $\mathcal{A}_c$ which serves as a means to locally project to an effective finite dimensional Hamiltonian on each cluster, $\mathcal{L}^c_{\operatorname{proj}}$. More precisely, 
        \begin{equation}\label{eq:coloring}
            \begin{aligned}
                \mathcal{L}^c_{\operatorname{dec}}&\coloneqq\sum_{\ell\in [m]\backslash \mathcal{A}_c}\mathcal{L}_\ell\\
                \mathcal{L}^c_{\operatorname{proj}}&\coloneqq\sum_{C\in \mathcal{V}_c}\cL^c_C
            \end{aligned}
        \end{equation}
        where each evolution $e^{t\cL_\ell}$ acts on mode $\ell$, with Lindblad operators $L_{\ell,1}=b_\ell$ and $L_{\ell,2}=b_\ell^{r+1}$, i.e. 
        \begin{equation}\label{eq:decoupling-lindbladian}
            \cL_\ell=\cL[b_\ell^{r+1}]+\cL[b_\ell]
        \end{equation}
        with $r=\lceil\frac{d}{2}\rceil-1$, such that the projection onto the common kernel is given by the vacuum state $P_\ell=\ket{0}\bra{0}$. Therefore, $P_\ell E_a P_\ell=0$ for all $a$ so that $\ell \in \operatorname{supp}(E_a)$. On the other site $e^{t\cL_C^c}$ acts on modes in cluster $C\in\mathcal{V}_c$ and has Lindblad operators $L_{C}$ described in Appendix \ref{subsec:few-mode-adiabatic-thm}, in particular in \eqref{eq:dissipation-multi-mode}
        \begin{equation}\label{eq:projection-lindbladian}
            \mathcal{L}_C^c \coloneqq \sum_{\ell\in C} \mathcal{L}[b_\ell (b_\ell - \alpha_\ell)] + \mathcal{L}[b_\ell^r (b_\ell - \alpha_\ell)]\,.   
        \end{equation}
        With this choice of generators, we can prove the following decoupling result:
        \begin{prop}\label{prop:decoupling-adiabatic-limit}
            Let $d \in \mathbb{N}$, $H$ as defined in Definition \ref{def-main:bosonic-H}, $\mathcal{L}$ the sum over single-mode modified $r$ photon dissipation with respect to the coloring \eqref{eq:coloring} defined in \eqref{eq:decoupling-lindbladian} and \eqref{eq:projection-lindbladian} with $r=\lceil\frac{d}{2}\rceil-1$, $\alpha\in\C^m$ and let $P=\bigotimes_{\ell=1}^m P_\ell$, where $P_\ell$ is the orthogonal projection onto $\operatorname{span}\{|0\rangle,|\alpha_\ell\rangle\}$. Here, $\alpha_\ell=0$ if $\ell\in[m]\backslash \mathcal{A}_c$. Then, for all $t \geq 0$ and $\rho=P\rho P$
            \begin{equation*}
                \|e^{-it\cH + \gamma t\mathcal{L}}(\rho) - e^{-it[PHP, P \cdot P]}(\rho)\|_1 \leq \frac{tC+C'}{\gamma}\,
            \end{equation*}
            for constants $C=(m+M)M C_{d, \mathfrak{k}}(\|\alpha\|_\infty^{4d+2}+1)$ and $C'=M\,C_{d, \mathfrak{k}}'(\|\alpha\|_\infty^d+1)$ for $C_{d, \mathfrak{k}}$ and $C_{d, \mathfrak{k}}'$ depending on $d$ and $ \mathfrak{k}$.
        \end{prop}
        \begin{proof}
            The statement follows from a similar proof to that of Proposition \ref{prop:few-modes-adiabatic-limit}. By definition of the decoupling and projective dissipation, the only difference lies in the entries of $\alpha$ that are assumed to be $\alpha_\ell=0$ if $\ell \in [m] \backslash \mathcal{A}_c$ and the second term of the dissipation defined by $\mathcal{L}[b_\ell^{r}(b_\ell-0)]+\mathcal{L}[b_\ell]$. However, the latter term does not contribute to the bounds of $C_L$, $C_H$, $\iota_1$, and $\iota_2$ because it is not leading order, so that 
            \begin{equation*}
                \begin{aligned}
                    \eta &\geq 1, &\quad C_H &\leq M \widetilde{C}_{d, \mathfrak{k}}(\|\alpha\|_\infty^d+1), &\quad C_L&\leq m M \widetilde{C}_{d, \mathfrak{k}}'(\|\alpha\|_\infty^{2d+1}+1),\\
                    &&\iota_1&\leq \sqrt{2}M( \mathfrak{k}+d-1)^{d}d!^{\mathfrak{k}/2-1},&\quad\iota_2&\leq\widetilde{C}_d''' (\|\alpha\|_\infty^{3d+2}+1)
                \end{aligned}
            \end{equation*}
            with the only difference to \eqref{eq:few-mode-constants} being that all $\alpha_\ell=0$ for $\ell\in[m]\backslash \mathcal{A}_c$. Then, 
            \begin{equation*}
                \begin{aligned}
                    C&=\frac{4}{\eta}\Big(\iota_1C_L+(\iota_2+C_H) C_H\Big)\leq (m+M)M C_{d, \mathfrak{k}}(\|\alpha\|_\infty^{4d+2}+1),\\
                    C'&=\frac{4C_H}{\eta}\leq M\,C_{d, \mathfrak{k}}'(\|\alpha\|_\infty^d+1)\,.
                \end{aligned}
            \end{equation*}
            Moreover, the projection is defined by $P=\bigotimes_{\ell=1}^m P_\ell$, where $P_\ell$ projects onto $\operatorname{span}\{|0\rangle,|\alpha_\ell\rangle\}$ if $\ell\in\mathcal{A}_c$ and $\operatorname{span}\{|0\rangle\}$ otherwise. This then directly deletes all parts of the Hamiltonian interacting with $[m] \backslash \mathcal{A}_c$ since every Hamiltonian term consists of at least one annihilation or creation operator.
        \end{proof}

\section{The learning protocol}\label{sec:learning}

In this appendix we will describe our Hamiltonian learning protocol. We will start with the remaining proofs for the single-mode case that we obmitted from Section \ref{sec:main-learning}, generalize it to the few-mode setting where we assume the number of modes $\mathfrak{k}=\Or(1)$, and then use it to solve the multi-mode setting where we apply strong engineered dissipation to decouple the system into constant-sized clusters as discussed in Appendix~\ref{subsec:decoupling-mode-adiabatic-thm}. Our goal is always to learn all Hamiltonian coefficients to precision $\epsilon$ with probability at least $1-\delta$.

\subsection{Learning a single-mode Hamiltonian}
\label{subsec:single_mode_learning}
In this appendix, we provide proofs for statements that had to be omitted from Section \ref{sec:main-learning} due to space constraints and we provide pseudocode in Algorithm \ref{alg:single_mode_learning} for our learning algorithm in the single mode case.
\begin{algorithm}
%\label{alg:single_mode_learning}
  \DontPrintSemicolon
  \SetKwFunction{FHamLearn}{HamLearn}
  \SetKwFunction{FExpectVal}{ExpectVal}
  \SetKwComment{Comment}{/* }{ */}
  \SetKwProg{Fn}{Function}{:}{}
  \Fn{\FHamLearn{$\epsilon$, $\delta$, $d$}}{
        Let $A_\mu$, $\mu=1,2,\cdots,d+1$, be as defined in \eqref{eq:choices_of_ri_ab}\;
        Let $\theta_{u,l} = \pi u/(l+1)$, $l=0,1,\cdots,d$, $u=0,1,\cdots,l$\;
        Let $\epsilon_1=\Theta(\epsilon),\delta_1=\Theta(\delta)$ be sufficiently small (but within a constant factor of $\epsilon$ and $\delta$ respectively) \;
        \For{$l=0,1,\cdots,d$}{
            \For{$u=0,1,\cdots,l$}{
                \For{$\mu=1,2,\cdots,d+1$}
                {
                    \For{$s=1,2,\cdots,S=\Or(\log(1/\delta_1))$}{
                        Let $H_{\mu,u,l,s} = \textsc{ExpectVal}(A_\mu e^{i\theta_{u,l}},\epsilon_1,d)$ \;
                    }
                    $H_{\mu,u,l}=\mathrm{median}\{H_{\mu,u,l,s}\}$ \;
                }
                Compute estimates for $h_l(\theta_{u,l})$ (defined in \eqref{eq:poly_of_r_coefs}) using polynomial regression with data $\{H_{\mu,u,l}\}_{\mu=1}^d$ as described in Section~\ref{sec:recovering_the_coefficients} \; 
            }
            Apply Fourier transform on the estimates for $\{h_l(\theta_{u,l})\}_{u=0}^l$ to obtain coefficient estimates $\{\hat{h}_{j,l-j}\}_{j=0}^l$ as described in Section~\ref{sec:recovering_the_coefficients}\;
        }
        \KwRet $\{\hat{h}_{j,j'}\}$ such that $|\hat{h}_{j,j'}-h_{j,j'}|\leq\epsilon$ for all $j,j'$ with probability at least $1-\delta$\;
  }
  \;
  \SetKwProg{Fn}{Function}{:}{\KwRet}
  \Fn{\FExpectVal{$\alpha$, $\epsilon$, $d$}}{
        Let $\Phi=\Or(d^2|\alpha|^d)$ be an upper bound of $|\braket{\alpha|H|\alpha}|$ \;
        Let $L=\lceil\log_{3/2}(3\Phi/\epsilon)\rceil$ \;
        Let $t_l,m_l$, $l=1,2,\cdots,L$, be as given in the proof of Theorem~\ref{thm:robust_frequency_estimation}\;
        \For{$l=1,2,\cdots,L$}{
            \For{$n=1,2,\cdots, m_l$}{
                \For{$\mathfrak{m}=1,2,\cdots,100$}{
                    Prepare state $\ket{+}\ket{0}$ \Comment*[r]{The first register is an ancilla qubit}
                    Apply displacement $D(\alpha)$ controlled on the ancilla qubit \;
                    Evolve under $-i[H,\cdot]+\gamma\mathcal{L}$ \Comment*[r]{$\mathcal{L}$ is the engineered dissipation}
                    Apply displacement $D(-\alpha)$ controlled on the ancilla qubit \;
                    Measure the ancilla qubit in $X$ basis to obtain $X_{l,n,\mathfrak{m}}\in\{\pm 1\}$ \;
                    Repeat the same experiment but measure in $Y$ basis to obtain $Y_{l,n,\mathfrak{m}}\in\{\pm 1\}$ \;
                }
                Compute $X_{l,n}=(1/100)\sum_{\mathfrak{m}=1}^{100}X_{l,n,\mathfrak{m}}$ \;
                Compute $Y_{l,n}=(1/100)\sum_{\mathfrak{m}=1}^{100}Y_{l,n,\mathfrak{m}}$ \;
            }
        }
        Apply robust frequency estimation with data $\{X_{l,n},Y_{l,n}\}$ as described in Theorem~\ref{thm:robust_frequency_estimation} to generate $\hat{\theta}$ \;
        \KwRet  $\hat{\theta}$ such that $\mathbb{E}[|\hat{\theta}-\braket{\alpha|H|\alpha}|^2]^{1/2}\leq \epsilon + \Or(d^2 |\alpha|^d e^{-|\alpha|^2/2})$\;
  }
  \caption{The pseudocode for the single-mode learning protocol. }\label{alg:single_mode_learning}
\end{algorithm}

The key tool for the proof of Theorem \ref{thm:robust_frequency_estimation} in the main text is the following lemma that allows us to incrementally refine the frequency estimate: 

\begin{lem}
    \label{lem:frequency_est_refine}
    Let $\theta\in[a,b]$.
    Let $Z(t)$ be a random variable such that, with probability $1$,
    \begin{equation}
        \label{eq:rv_correctness_condition}
        |Z(t)-e^{i\theta t}|< 1/2.
    \end{equation}
    Then we can correctly distinguish between two overlapping cases $\theta\in[a,(a+2b)/3]$ and $\theta\in[(2a+b)/3,b]$ with one sample of $Z(\pi/(b-a))$. 
\end{lem}

\begin{proof}
    These two situations described in the statement can be distinguished by looking at the value of $f(\theta)=\sin\left(\frac{\pi}{b-a}\left(\theta-\frac{a+b}{2}\right)\right)$. Eq.\eqref{eq:rv_correctness_condition} shows $\left|\Im\left(e^{-i\frac{(a+b)\pi}{2(b-a)}}Z(\pi/(b-a))\right)-f(\theta)\right|<1/2$. If $$\Im\left(e^{-i\frac{(a+b)\pi}{2(b-a)}}Z(\pi/(b-a))\right)\leq 0,$$ then $f(\theta)<1/2$, which implies $\theta\in[a,(a+2b)/3]$. If $$\Im\left(e^{-i\frac{(a+b)\pi}{2(b-a)}}Z(\pi/(b-a))\right)> 0,$$ then $f(\theta)>-1/2$, which implies $\theta\in[(2a+b)/3,b]$. 
\end{proof}

Now we are ready to prove Theorem \ref{thm:robust_frequency_estimation}, which we restate here for convenience:
\begin{thm}
    Let $\theta\in[-\Phi,\Phi]$.
    Let $X(t)$ and $Y(t)$ be independent random variables satisfying
    \begin{equation} \label{eq:rv_correctness_condition_prob}
        \begin{aligned}
            &|X(t)-\cos(\theta t)|< 1/\sqrt{8}, \text{ with probability at least }2/3, \\
            &|Y(t)-\sin(\theta t)|< 1/\sqrt{8}, \text{ with probability at least }2/3.
        \end{aligned}
    \end{equation}
    Then with independent samples $X(t_1),X(t_2),\cdots,X(t_{\Gamma})$ and $Y(t_1),Y(t_2),\cdots,Y(t_{\Gamma})$, with
    \begin{equation}
        {\Gamma}=\Or(\log^2(\Phi/\epsilon)), \quad T=t_1+t_2+\cdots+t_{\Gamma}=\Or(1/\epsilon), \quad \max_j t_j=\Or(1/\epsilon),
    \end{equation}
    and $t_j\geq 0$, we can construct  an estimator $\hat{\theta}$ such that
    \begin{equation}
        \sqrt{\mathbb{E}[|\hat{\theta}-\theta|^2]}\leq \epsilon.
    \end{equation}
\end{thm}
\begin{proof}
    We build a random variable $S(t)$ satisfying \eqref{eq:rv_correctness_condition}, with which we will iteratively narrow down the interval $[a,b]$ containing $\theta$ until $|b-a|\leq 2\epsilon/3$, at which point we choose $\hat{\theta}=(a+b)/2$. If $\theta\in[a,b]$ we will then ensure $|\hat{\theta}-\theta|\leq \epsilon/3$. However, each iteration will involve some failure probability, which we will analyze later.

    To build the random variable $S(t)$, we first use $m$ independent samples of $X(t)$ and then take median $X_{\mathrm{median}}(t)$, which satisfies 
    $
    |X_{\mathrm{median}}(t)-\cos(\theta t)|< 1/\sqrt{8}
    $
    with probability at least $1-\delta/2$, where by \eqref{eq:rv_correctness_condition_prob} and the Chernoff bound 
    $
    \delta = c_1 e^{-c_2 m},
    $
    for some universal constant $c_1,c_2$. The way we apply the Chernoff bound is as follows: in order for the median $X_{\mathrm{median}}(t)$ to deviate from $\cos(\theta t)$ by at least $1/\sqrt{8}$, then at least $1/2$ of the samples of $X(t)$ will have such large deviations. However, since for a single sample this large deviation occurs with probability at most $1/3$, the Chernoff bound tells us that an error exceeding the above bound in $X_{\mathrm{median}}(t)$ occurs with a probability that decays exponentially in the number of samples.
    Similar to $X_{\mathrm{median}}(t)$ we can obtain $Y_{\mathrm{median}}(t)$ such that
    $
    |Y_{\mathrm{median}}(t)-{\sin}(\theta t)|< 1/\sqrt{8}
    $
    with probability at least $1-\delta/2$. With these medians we then define
    $
    S(t) = X_{\mathrm{median}}(t) + iY_{\mathrm{median}}(t).
    $
    This random variable satisfies
    $
    |S(t)-e^{i\theta t}|< 1/2
    $
    with probability at least $1-\delta$ using the union bound. It therefore allows us to solve the discrimination task in Lemma~\ref{lem:frequency_est_refine} with probability at least $1-\delta$.

    Whether each iteration proceeds correctly or not, the algorithm terminates after 
    $
    L=\lceil \log_{3/2}(3\Phi/\epsilon) \rceil
    $
    iterations. By Lemma~\ref{lem:frequency_est_refine}, in the $l$-th iteration the length of the interval $[a,b]$ is $(2/3)^{l-1}2\Phi$ and hence we sample $X(t_l)$ and $Y(t_l)$ where $t_l=(\pi/2\Phi)(3/2)^{l-1}$. We let the $l$-th iteration use $m_l$ samples of $X(t_l)$ and $Y(t_l)$ for computing the median, and therefore the failure probability is at most $\delta_l=c_1 e^{-c_2 m_l}$.

    The total number of samples (for either $X(t)$ or $Y(t)$) is therefore
    \begin{equation}
        \label{eq:total_num_samples}
        {\Gamma} = \sum_{l=1}^{L} m_l,
    \end{equation}
    whereas all the $t$ in each sample added together is
    \begin{equation}
        \label{eq:total_time}
        T = \sum_{l=1}^{L} m_l t_l.
    \end{equation}
    The above is the objective that we want to minimize. The constraint is imposed by the precision that we want to achieve, which we compute below.

    With probability
    $
    (1-\delta_1)(1-\delta_2)\cdots (1-\delta_{l-1})\delta_l,
    $
    the $l$-th iteration is the first one to fail. The resulting error is no larger than $2\Phi(2/3)^{l-2}$. Therefore the mean squared error (MSE) is at most
    \begin{equation}
        \label{eq:mean_squared_error_choose_delta}
        \begin{aligned}
            &\mathbb{E}[|\hat{\theta}-\theta|^2]\\
            &=4\Phi^2\sum_{l=1}^L (1-\delta_1)(1-\delta_2)\cdots (1-\delta_{l-1})\delta_l(4/9)^{l-2} + (1-\delta_1)(1-\delta_2)\cdots (1-\delta_{L})(\epsilon/3)^2 \\
            &\leq 4\Phi^2\sum_{l=1}^L\delta_l(4/9)^{l-2} + \epsilon^2/9.
        \end{aligned}
    \end{equation}
    Now we choose $\delta_l$ so that
    $
    \delta_l(4/9)^{l-2}\leq \frac{\epsilon^2/9}{2^{L-l}\Phi^2},
    $
    which ensures $ \mathbb{E}[|\hat{\theta}-\theta|^2]\leq \epsilon^2$.

    We therefore only need to choose
    $
    \delta_l\leq \left(\frac{4}{9}\right)^2\left(\frac{2}{9}\right)^{L-l} \leq (9/4)^{l-2}\frac{\epsilon^2/9}{2^{L-l}\Phi^2},
    $
    where we have used $\epsilon/(3\Phi)\geq (2/3)^L$.
    For the above choice of $\delta_l$, it suffices to choose $m_l = c_3 + c_4(L-l)$ for some universal constants $c_3,c_4$.
    Therefore we can compute
    \begin{equation}
        {\Gamma} = \sum_{l=1}^{L} m_l = \sum_{l=1}^{L} ( c_3 + c_4(L-l)) \leq \Or(L^2)=\Or(\log^2(3\Phi/\epsilon)), 
    \end{equation}
    and
    \begin{equation}
        T = \sum_{l=1}^{L} m_l t_l = \frac{\pi}{2\Phi}\sum_{l=1}^{L} (c_3 + c_4(L-l)) \left(\frac{3}{2}\right)^{l-1} = \Or((3/2)^L/\Phi)=\Or(1/\epsilon),
    \end{equation}
    where we have used the elementary equality that
    $
    \sum_{l=1}^L x^{l-1}(L-l) = \frac{x^L-1}{(x-1)^2}-\frac{L}{x-1}.
    $
\end{proof}

\subsection{Learning a few-mode Hamiltonian}
\label{sec:learning_few_mode_ham}

The protocol for learning a Hamiltonian on $\mathfrak{k}$ bosonic modes for $\mathfrak{k}=\Or(1)$ follows largely the same idea as the single-mode case. The Hamiltonian takes the form 
\begin{equation}
    \label{eq:ham_k_modes}
    H = \sum_{l_1=0}^{d_1}\cdots\sum_{l_\mathfrak{k}=0}^{d_\mathfrak{k}} \sum_{j_1+j_1'=l_1}\cdots\sum_{j_\mathfrak{k}+j_\mathfrak{k}'=l_\mathfrak{k}} h_{j_1j_1'\cdots j_\mathfrak{k}j_\mathfrak{k}'} (b_1^\dag)^{j_1}b_1^{j_1'}\cdots (b_\mathfrak{k}^\dag)^{j_\mathfrak{k}}b_\mathfrak{k}^{j_\mathfrak{k}'}.
\end{equation}
for coefficients $h_{j_1j_1'\cdots j_\mathfrak{k}j_\mathfrak{k}'} $ with {$h_{j_1j_1'\cdots j_\mathfrak{k}j_\mathfrak{k}'} = \overline{h_{j_1'j_1\cdots j_\mathfrak{k}'j_\mathfrak{k}} }$} and $|h_{j_1j_1'\cdots j_\mathfrak{k}j_\mathfrak{k}'} |\le 1$. {We also denote $d=\sum_q d_q$ as in Definition~\ref{def-main:bosonic-H}.} Our learning protocol will generate estimates $\hat{h}_{j_1j_1'\cdots j_\mathfrak{k}j_\mathfrak{k}'}$ such that
\begin{equation}
    \max_{j_1j_1'\cdots j_\mathfrak{k}j_\mathfrak{k}'}|\hat{h}_{j_1j_1'\cdots j_\mathfrak{k}j_\mathfrak{k}'}-{h}_{j_1j_1'\cdots j_\mathfrak{k}j_\mathfrak{k}'}|\leq \epsilon\text{ with probability at least }1-\delta.
\end{equation}

\subsubsection{Experimental setup}
\label{sec:few_mode_experimental_setup}

The experimental setup is similar to the one described in Section~\ref{sec:experimental_setup}.
There are two differences: On each mode $q$ we apply dissipation with jump operators $b_q (b_q-\alpha_q)$ and $b_q^{r}(b_q-\alpha_q)$, where {$r=\lceil d/2\rceil-1$}, so that the dynamics is stabilized within a $2^\mathfrak{k}$ dimensional subspace spanned by $\bigotimes_{q=1}^\mathfrak{k}\{\ket{0},\ket{\alpha_q}\}$. Also, instead of applying the controlled displacement on one bosonic mode we apply it on $\mathfrak{k}$ modes, and the displacement is described by a vector $\alpha=(\alpha_1,\cdots,\alpha_\mathfrak{k})$. The experiment allows us to measure $X$ and $Y$ observables whose expectation values are approximately $\cos(\braket{\alpha|H|\alpha}t)$ and $\sin(\braket{\alpha|H|\alpha}t)$ respectively. We will see that by analyzing the effective Hamiltonian, with the same robust frequency estimation procedure described in Section~\ref{sec:robust_frequency_estimation}, we can again estimate the value of $\braket{\alpha|H|\alpha}$ with Heisenberg-limited scaling.

\subsubsection{The effective Hamiltonian}

We let $|\alpha_q|=A_q$. and throughout this section we assume $A_q\geq 1$.
Similar to the single-mode case, we introduce
\begin{equation}
    \label{eq:single_mode_basis_q}
    \ket{\Psi_0^q} = \ket{0},\quad \ket{\Psi_1^q} = \frac{\ket{\alpha_q}-\ket{0}\braket{0|\alpha_q}}{\sqrt{1-e^{-A_q^2}}}.
\end{equation}
Then we define, for a bitstring $b=(b_1,\cdots,b_\mathfrak{k})\in\{0,1\}^\mathfrak{k}$,
\begin{equation}
    \label{eq:multi_mode_basis}
    \ket{\Psi_b} = \bigotimes_{q=1}^\mathfrak{k} \ket{\Psi_{b_q}^q}.
\end{equation}
We therefore have a basis for the subspace stabilized by the dissipation. The effective Hamiltonian is
\begin{equation}
    H^{\mathrm{proj}} = \Pi H\Pi,
\end{equation}
where
\begin{equation}
    \Pi=\sum_{b\in\{0,1\}^\mathfrak{k}}  \ket{\Psi_b}\bra{\Psi_b}.
\end{equation}
We want to show that this effective Hamiltonian can be well-approximated by 
\begin{equation}
    \widetilde{H}^{\mathrm{proj}} = \sum_{b\in\{0,1\}^\mathfrak{k}}  \ket{\Psi_b}\bra{\Psi_b}H\ket{\Psi_b}\bra{\Psi_b}.
\end{equation}
In other words, we want to show that the off-diagonal elements can be ignored. This can be done by computing an upper bound for the off-diagonal elements. We first describe the effect of the annihilation operator on $\ket{\Psi_1^q}$:
\begin{equation}
    \label{eq:effect_annihilation_op}
    b_q^{j}\ket{\Psi_1^q} = 
    \begin{cases}
        \ket{\Psi_1^q},&\text{if } j=0\\
        \frac{\alpha_q^j\ket{\alpha_q}}{\sqrt{1-e^{-A_q^2}}},&\text{if } j\geq 1\,.
    \end{cases}
\end{equation}
Using this we can compute $\braket{\Psi_0^q|(b_q^\dag)^{j}b_q^{j'}|\Psi_1^q}$. Note that if $j\neq 0$ then $\braket{\Psi_0^q|(b_q^\dag)^{j}b_q^{j'}|\Psi_1^q}=0$. Therefore we only need to focus on the case with $j=0$, and using \eqref{eq:effect_annihilation_op} we have
\begin{equation}
    \label{eq:different_element_overlap_single_mode_component}
    |\braket{\Psi_0^q|(b_q^\dag)^{j}b_q^{j'}|\Psi_1^q}| = \Or(A_q^{j'}e^{-A_q^2}).
\end{equation}
Similarly we can bound $\braket{\Psi_1^q|(b_q^\dag)^{j}b_q^{j'}|\Psi_1^q}$ by examining all cases of $j=0$ or $j\neq 0$ and $j'=0$ or $j'\neq 0$. We can obtain a bound 
\begin{equation}
    \label{eq:identical_element_overlap_single_mode_component}
    |\braket{\Psi_1^q|(b_q^\dag)^{j}b_q^{j'}|\Psi_1^q}| = \Or(A_q^{j+j'}).
\end{equation}
We also know that $|\braket{\Psi_0^q|(b_q^\dag)^{j}b_q^{j'}|\Psi_0^q}|\leq 1$.
Based on these observations, we have, for any $b\neq b'$ 

\begin{equation}
    \braket{\Psi_b|(b_1^\dag)^{j_1}b_1^{j_1'}\cdots (b_\mathfrak{k}^\dag)^{j_\mathfrak{k}}b_\mathfrak{k}^{j_\mathfrak{k}'}|\Psi_{b'}} = \prod_{q=1}^\mathfrak{k} \braket{\Psi^q_{b_q}|(b_q^\dag)^{j_q}b_q^{j_q'}|\Psi^q_{b'_q}} = \Or\left(e^{-\min\{A_q^2\}}\prod_{q}A_q^{j_q+j_q'}\right).
\end{equation}
Therefore
\begin{equation}
\label{eq:off_diag_entries_bound}
     |\braket{\Psi_b|H|\Psi_{b'}}|=\Or\left(e^{-\min\{A_q^2\}}\prod_{q}(A_q+1)^{d_q}\right),
\end{equation}
which gives us 
\begin{equation}
\label{eq:err_removing_orthogonal_entries_multi_mode}
    \|H^{\mathrm{proj}}-\widetilde{H}^{\mathrm{proj}}\|_\infty=\Or\left(4^\mathfrak{k}e^{-\min\{A_q^2\}}\prod_{q}(A_q+1)^{d_q}\right).
\end{equation}

We then analyze the effect of the non-orthogonality of coherent states on the phase estimation experiment. Note that
\begin{equation}
    \widetilde{H}^{\mathrm{proj}}\ket{\alpha} = \sum_{b\in\{0,1\}^b} \ket{\Psi_b}\braket{\Psi_b|H|\Psi_b}\braket{\Psi_b|\alpha}.
\end{equation}
Similar to \eqref{eq:off_diag_entries_bound} we can derive a bound
\begin{equation}
    |\braket{\Psi_b|H|\Psi_b}| = \Or\left(\prod_{q}(A_q+1)^{d_q}\right).
\end{equation}
Also observe that
\begin{equation}
    |\braket{\Psi_b|\alpha}| = \Or(e^{-\min\{A_q^2\}})
\end{equation}
if $b\neq (1,1,\cdots,1)$, we have
\begin{equation}
    \widetilde{H}^{\mathrm{proj}}\ket{\alpha} = \braket{\alpha|H|\alpha}\ket{\alpha} + \Or\left(2^{\mathfrak{k}}e^{-\min\{A_q^2\}}\prod_{q}(A_q+1)^{d_q}\right),
\end{equation}
where we have also used
\begin{equation}
    |\braket{\alpha|H|\alpha}-\braket{\Psi_{(1,1,\cdots,1)}|H|\Psi_{(1,1,\cdots,1)}}|= \Or\left(\mathfrak{k}e^{-\min\{A_q^2\}}\prod_{q}(A_q+1)^{d_q}\right).
\end{equation}
{The latter follows in turn from $\|H\ket{\Psi_{(1,1,\cdots,1)}}\| = \Or\left(\prod_{q}(A_q+1)^{d_q}\right)$, $\|H\ket{\alpha}\| = \Or\left(\prod_{q}(A_q+1)^{d_q}\right)$ (again similar to \eqref{eq:off_diag_entries_bound}) and $\|\ket{\alpha} - \ket{\Psi_{(1,1,\cdots,1)}}\| = \Or(\mathfrak{k} e^{-\min\{A_q^2\}})$ by \eqref{eq:distance_bw_Psi1_and_alpha}.}

Combining the above with \eqref{eq:err_removing_orthogonal_entries_multi_mode}, we have
\begin{equation}
    H^{\mathrm{proj}}\ket{\alpha} = \braket{\alpha|H|\alpha}\ket{\alpha} + \Or\left(4^\mathfrak{k} e^{-\min\{A_q^2\}}\prod_{q}(A_q+1)^{d_q}\right).
\end{equation}
From this we have
\begin{equation}
\label{eq:deviation_from_alpha_expectation_few_mode}
    \|(e^{-iH^{\mathrm{proj}}t}-e^{-i\braket{\alpha|H|\alpha}t})\ket{\alpha}\| = \Or\left(4^\mathfrak{k} e^{-\min\{A_q^2\}}\prod_{q}(A_q+1)^{d_q}t\right).
\end{equation}

\subsubsection{Recovering the coefficients}
 Next, we briefly discuss how to recover the coefficients from $\braket{\alpha|H|\alpha}$ for different values of $\alpha$, and postpone the proofs of the results used in this appendix to Appendix \ref{sec:polyregression}. Let $\alpha_q = A_q e^{i\theta_q}$ for $q=1,2,\cdots,\mathfrak{k}$. Then 
\begin{equation}
\label{eq:ham_expect_val_decompose}
    \braket{\alpha|H|\alpha} = \sum_{l_1=0}^{d_1}\cdots\sum_{l_k=0}^{d_\mathfrak{k}} h_{l_1\cdots l_\mathfrak{k}}(\theta)A_1^{l_1}\cdots A_\mathfrak{k}^{l_\mathfrak{k}},
\end{equation}
where $\theta=(\theta_1,\cdots,\theta_\mathfrak{k})$, and
\begin{equation}
\label{eq:ham_expect_val_poly_coef}
    h_{l_1\cdots l_\mathfrak{k}}(\theta) = \sum_{j_1+j_1'=l_1}\cdots\sum_{j_\mathfrak{k}+j_\mathfrak{k}'=l_\mathfrak{k}} h_{j_1j_1'\cdots j_\mathfrak{k}j_\mathfrak{k}'}e^{-i\theta(j_1-j_1')}\cdots e^{-i\theta(j_\mathfrak{k}-j_\mathfrak{k}')}.
\end{equation}
From the above we can see that $\braket{\alpha|H|\alpha}$ is a polynomial of $A_1,A_2,\cdots,A_q$, whose coefficients $h_{l_1\cdots l_\mathfrak{k}}(\theta)$ are Fourier sums of the angles $\theta_1,\cdots,\theta_\mathfrak{k}$. Therefore we first use polynomial interpolation to estimate the polynomial coefficients $h_{l_1\cdots l_\mathfrak{k}}(\theta)$ for any given $\theta$, and from that estimate the Hamiltonian coefficients $h_{j_1j_1'\cdots j_\mathfrak{k}j_\mathfrak{k}'}$ through a Fourier transform. 

More concretely, for a given $\theta$ we denote $p_\theta(A)=\braket{\alpha|H|\alpha}$, where $A=(A_1,\cdots,A_\mathfrak{k})$, and $\alpha=(A_1 e^{i\theta_1},\cdots,A_\mathfrak{k} e^{i\theta_\mathfrak{k}})$. $p_\theta(A)$ is therefore a polynomial of $A$ where the degrees for $A_1,\cdots,A_\mathfrak{k}$ are $d_1,\cdots,d_\mathfrak{k}$ respectively. We define a set of $(d_1+1)(d_2+1)\cdots (d_\mathfrak{k}+1)$ nodes $A_{\mu_1,\mu_2,\cdots,\mu_\mathfrak{k}}=(A^1_{\mu_1},A^2_{\mu_2},\cdots,A^\mathfrak{k}_{\mu_\mathfrak{k}})$, where each $A^{q}_{\mu_q}$ is 
\begin{equation}
    \label{eq:choose_A_few_mode}
    A^q_{\mu_q} = \frac{a_q+b_1}{2} + \frac{b_q-a_q}{2}\cos\left(\frac{(2\mu_q-1)\pi}{2(d_q+1)}\right),
\end{equation}
as defined in Appendix \ref{sec:polyregression}'s equation \eqref{eq:quadrature_pts_ab_multivariate} with some change of notation, on a hypercube $[a_1,b_1]\times\cdots\times[a_\mathfrak{k},b_\mathfrak{k}]$, so that the coefficients of $p_\theta(A)$ can be uniquely determined from $\{p_{\theta}(A_{\mu_1,\mu_2,\cdots,\mu_\mathfrak{k}})\}$. From experiments, we can generate estimates $\{\hat{p}_{\theta,\mu_1,\mu_2,\cdots,\mu_\mathfrak{k}}\}$ satisfying 
\begin{equation}
    |p_{\theta}(A_{\mu_1,\mu_2,\cdots,\mu_\mathfrak{k}})-\hat{p}_{\theta,\mu_1,\mu_2,\cdots,\mu_\mathfrak{k}}|\leq \epsilon_1
\end{equation}
for all $\theta$ and $(\mu_1,\dots,\mu_\mathfrak{k})$ and with high probability. Lemma~\ref{lem:bounding_coef_with_chebyshev_nodes_multivariate} in Appendix \ref{sec:polyregression} then tells us that as long as
\begin{equation}
\label{eq:interval_extrapolation_condition}
    b_q-a_q\geq 2,b_q\geq 2 a_q>0,
\end{equation}
for all $q=1,2,\cdots,\mathfrak{k}$, we can guarantee that the recovered coefficients $\hat{h}_{\theta,l_1\cdots l_\mathfrak{k}}$ are also accurate
\begin{equation}
    |h_{l_1\cdots l_\mathfrak{k}}(\theta)-\hat{h}_{\theta,l_1\cdots l_\mathfrak{k}}|\leq C_{d_1,d_2,\cdots,d_\mathfrak{k}} \epsilon_1,
\end{equation}
where $C_{d_1,d_2,\cdots,d_\mathfrak{k}}$ is a constant that depends only on $d_1,d_2,\cdots,d_\mathfrak{k}$. 
In estimating $h_{l_1\cdots l_\mathfrak{k}}(\theta)$ we are essentially estimating the high-order derivatives of $p_\theta(A)$ around $A=0$. Because we cannot obtain the value of $p_\theta(A)$ for small $\|A\|$ (the states $\ket{\alpha}$ and $\ket{0}$ would have large overlap in this case), the above procedure is essentially an extrapolation from the hypercube $[a_1,b_1]\times\cdots\times[a_\mathfrak{k},b_\mathfrak{k}]$ to $0$. This is the reason why we require \eqref{eq:interval_extrapolation_condition} so that the distance we extrapolate is not too large compared to the size of the interval in which we have the needed information $p_\theta(A)$.

From the above we can see that we can estimate each $h_{j_1j_1'\cdots j_kj_k'}$ to precision $\epsilon$ if we can estimate each $p_{\theta_u}(A_{i_1,i_2,\cdots,i_\mathfrak{k}})$ to precision $\Or(\epsilon)$, where
\begin{equation}
\label{eq:angles_fourier}
    \theta_u = (\theta_{u_1}^1,\theta_{u_2}^2,\cdots,\theta_{u_\mathfrak{k}}^\mathfrak{k})= \left(\frac{u_1 \pi}{l_1+1},\frac{u_2 \pi}{l_2+1},\cdots,\frac{u_\mathfrak{k} \pi}{l_\mathfrak{k}+1}\right)
\end{equation}
is chosen so that we can recover the coefficients $h_{j_1j_1'\cdots j_\mathfrak{k}j_\mathfrak{k}'}$ using the Fourier transform. 
Each $p_\theta(A)$ is obtained through the phase estimation experiment introduced in Section~\ref{sec:few_mode_experimental_setup}, and can be estimated to precision $\epsilon_1$ with $\Or(1/\epsilon_1)$ total evolution time when $A_q=\Omega(\sqrt{\log(1/\epsilon_1)})$ using \eqref{eq:deviation_from_alpha_expectation_few_mode} and the robust frequency estimation procedure in Section~\ref{sec:robust_frequency_estimation}.
We therefore have the following theorem:
\begin{thm}
    \label{thm:estimating_few_mode}
    Running the experiments in Section~\ref{sec:few_mode_experimental_setup} with dissipation strength of order $\gamma=\Or(\epsilon^{-1}\log^{2d+1/2}(1/\epsilon))$ and measuring the $X$ and $Y$ observables on the ancilla qubit as also discussed in Section~\ref{sec:few_mode_experimental_setup}, with 
    \[
    \alpha = (A_{i_1}^1 e^{i\theta^1_{u_1}},A_{i_2}^2 e^{i\theta^2_{u_2}},\cdots,A_{i_\mathfrak{k}}^\mathfrak{k} e^{i\theta^\mathfrak{k}_{u_\mathfrak{k}}}),
    \]
    where $A_{i_\nu}^\nu$ is defined in \eqref{eq:choose_A_few_mode}, $\theta_{u_\nu}^\nu$ is defined in \eqref{eq:angles_fourier}, with $a_\nu,b_\nu$ satisfying \eqref{eq:interval_extrapolation_condition} and $a_\nu=\Omega(\sqrt{\log(1/\epsilon_1)})$, we can estimate all coefficients $h_{j_1j_1'\cdots j_\mathfrak{k}j_\mathfrak{k}'}$ to precision $\epsilon$ with probability at least $1-\delta$ with 
    \begin{align*}
    \Or((1/\epsilon)\log(1/\delta)) ~~\text{total evolution time and }~~~~\Or(\log^2(\log(1/\epsilon)/\epsilon)\log(1/\delta))~~\text{experiments\,.}
    \end{align*}
\end{thm}
In the above the dissipation strength is obtained through Proposition~\ref{prop:few-modes-adiabatic-limit}.

\subsection{Scalable learning of a multi-mode Hamiltonian}
\label{sec:learn_multimode}

We now consider a system consisting of $m$ modes. Recall that the Hamiltonian is $H=\sum_a E_a$ in \eqref{eq:main-bosonic-H}, where each term $E_a$ has a different support.
We assume that each Hamiltonian term acts on at most $\mathfrak{k}=\Or(1)$ modes, and each term $E_a$ overlaps with at most $\mathfrak{d}$ other terms.
As discussed in Section~\ref{subsec:decoupling-mode-adiabatic-thm}, 
We introduce a set $\mathcal{V}$ of all interacting clusters as defined in Definition~\ref{defn:interacting_cluster}. They form a cluster interaction graph as defined in Definition~\ref{def:interaction-graph}.
By coloring the interaction graph as in Section~\ref{subsec:decoupling-mode-adiabatic-thm}, $\mathcal{V}$ can be divided into a disjoint union 
$
\mathcal{V}=\bigsqcup_{c\in\chi}\mathcal{V}_c\,,
$
where no two adjacent vertices are in the same $\mathcal{V}_c$. In other words, for any $C$ and $C'$ in $\mathcal{V}_c$, $C\cap C'=\emptyset $, and for any $C''\in \mathcal{V}$, either $C\cap C''=\emptyset$ or $C'\cap C''=\emptyset$. 
Next, we fix a color $c$ and apply the dissipation $\mathcal{L}^c_{\mathrm{dec}}$ described in \eqref{eq:coloring}. This has the effect of resulting in a decoupled effective Hamiltonian
\begin{equation}
    H_{\mathrm{eff}}  = \sum_{C\in \mathcal{V}_c} H_C.
\end{equation}
Each $H_C$ is supported only on $C$, which contains at most $\mathfrak{k}$ bosonic modes. Moreover, these $C\in \mathcal{V}_c$ do not overlap with each other due to the coloring. Consequently we can learn each $H_C$ using the few-mode learning protocol described in Section~\ref{sec:learning_few_mode_ham} in parallel, for which Theorem~\ref{thm:estimating_few_mode} directly tells us the total evolution time and number of experiments needed. The dissipation strength $\gamma$ needed can be determined from Proposition~\ref{prop:decoupling-adiabatic-limit}.

\begin{thm}
    \label{thm:estimating_multi_mode}
    Applying dissipation $\mathcal{L}^c_{\mathrm{dec}}$ for each $c\in [\chi]$ and 
    running the experiments in Section~\ref{sec:few_mode_experimental_setup} for each cluster $C\in \mathcal{V}_c$ in parallel, and measuring the $X$ and $Y$ observables on the ancilla qubit as also discussed in Section~\ref{sec:few_mode_experimental_setup}, with dissipation strength $\gamma=\Or(m^2\epsilon^{-1}\log^{2d+1/2}(1/\epsilon))$ and other parameters chosen according to Theorem~\ref{thm:estimating_few_mode}, we can estimate all coefficients of $H$ to precision $\epsilon$ with probability at least $1-\delta$
\begin{align*}
\Or((1/\epsilon)\log(m/\delta)) ~~\text{total evolution time and}~~~~ \Or(\log^2(\log(1/\epsilon)/\epsilon)\log(m/\delta))~~\text{ experiments\,.}
\end{align*}
\end{thm}

We remark that the $m$ in the logarithmic factors comes from taking a union bound to make sure that the estimation is accurate for all coefficients. Moreover, we use the fact that $M=\Or(m)$ in a low-intersection Hamiltonian to replace every $M$ with $m$ in the total evolution time and the number of experiments.

\section{Relative boundedness by the (modified) photon dissipation}\label{sec:realtive-boundedness}
    In this appendix, we prove explicit relative boundedness assumptions of a given Hamiltonian by variations of the photon dissipation. The strategy is to first upper bound the Hamiltonian by a product of number operators, then further by the sum of local number operators and finally each term by a modified on-site photon dissipations. First, we start with the upper bound on the Hamiltonian defined in Definition \ref{def-main:bosonic-H}, i.e. 
    \begin{align*}
        H=\sum_{a=1}^M \ E_a\,,\qquad\text{where}\qquad E_a= \sum_{\substack{\mathbf{j},\mathbf{j}' \in \mathbb{N}^{ \mathfrak{k}}\\ \|\mathbf{j}+\mathbf{j}'\|_1\le d}}h^{(a)}_{\mathbf{j},\mathbf{j'}} (\mathbf{b}^\dagger)^{\mathbf{j}}\,\mathbf{b}^{\mathbf{j}'}\,.
    \end{align*}
    We assume that at least either $\mathbf{j}\neq0$ or $\mathbf{j}'\neq0$, $|h^{(a)}_{\mathbf{j},\mathbf{j}'}|\le 1$ and $\overline{(h^{(a)})_{\mathbf{j},\mathbf{j}'}}=h^{(a)}_{\mathbf{j}',\,\mathbf{j}}$. Moreover, it is assumed that each $E_a$ acts on at most $ \mathfrak{k}=\mathcal{O}(1)$ modes, and for each $a$, $E_a$ overlaps with at most $ \mathfrak{d}=\mathcal{O}(1)$ other $E_{\tilde{a}}$'s.
    \begin{lem}\label{lem:rel-bounded-learning-H-N}
        For any $H$ given in Definition \ref{def-main:bosonic-H}, any state $|\psi\rangle\in \operatorname{dom}(\sum_{\ell\in\{1,...,m\}}N_\ell^{d/2})$ and $\delta\in(0,1)$
        \begin{equation*}
            \|H\ket{\psi}\|\leq M( \mathfrak{k}+d-1)^{2d}d!^{\mathfrak{k}/2-1}\max_{\ell\in\supp(H)}\|(N_\ell+I)^{d/2}\ket{\psi}\|\,.
        \end{equation*}
        with $\supp(H)=\bigcup_{a=1}^M\supp(E_A)$
    \end{lem}
    We begin by proving the result on the following subspace
    \begin{equation*}
        \mathscr{H}_f =\operatorname{span}(\{\ket{n}\,|\,n\in\N^m\})
    \end{equation*}
    and then extend the inequality to all $|\psi\rangle \in \operatorname{dom}(\sum_{\ell \in \{1, \ldots, m\}} N_\ell^{d/2})$ by continuity of the norm.
    \begin{proof}[Proof of Lemma \ref{lem:rel-bounded-learning-H-N}]
        First, let $\ket{\psi}\in\mathscr{H}_f$ and apply the triangle inequality to reduce the problem to a fixed $a\in\{1,...,M\}$ and $\mathbf{j},\mathbf{j}' \in \mathbb{N}^{ \mathfrak{k}}$, i.e.
        \begin{equation*}
            \|H\ket{\psi}\|\leq \sum_{a=1}^{M}\sum_{\substack{\mathbf{j},\mathbf{j}' \in \mathbb{N}^{ \mathfrak{k}}\\ \|\mathbf{j}+\mathbf{j}'\|_1\le d}}1_{h^{(a)}_{\mathbf{j},\mathbf{j'}}\neq0} \|(\mathbf{b}^\dagger)^{\mathbf{j}}\,\mathbf{b}^{\mathbf{j}'}\ket{\psi}\|\,.
        \end{equation*}
        Then, the commutation relation shows in each mode for $\ket{\varphi}$ in an appropriate domain 
        \begin{equation*}
            \begin{aligned}
                \|(b_\ell^{j_\ell})^\dagger b_\ell^{j'_{\ell}}\ket{\varphi}\|_\infty^2&=\bra{\varphi}(b_\ell^{j'_{\ell}})^\dagger (N_\ell+I)\cdots (N_\ell +(1+j_\ell)I) b_\ell^{j'_{\ell}}\ket{\varphi}\\
                &=\bra{\varphi}N_\ell\cdots (N_\ell-j'_\ell I) (N_\ell+(1-j'_\ell)I)\cdots (N_\ell +(1+j_\ell-j'_\ell)I)\ket{\varphi}\\
                &\leq (j_\ell-j'_\ell)!\bra{\varphi}(N_\ell+I)^{j_\ell+j_\ell'}\ket{\varphi}\,,
            \end{aligned}
        \end{equation*}
        which shows
        \begin{equation*}
            \begin{aligned}
                \|(\mathbf{b}^\dagger)^{\mathbf{j}}\,\mathbf{b}^{\mathbf{j}'}\ket{\psi}\|&\leq \|\prod_{\ell\in\supp(E_a)} \sqrt{j_\ell!}(N_\ell+I)^{(j_\ell+j'_\ell)/2}\ket{\psi}\|\\
                &\leq d!^{\mathfrak{k}/2}\|\prod_{\ell\in\supp(E_a)}(N_\ell+I)^{(j_\ell+j'_\ell)/2}\ket{\psi}\|\,.
            \end{aligned}
        \end{equation*}
        Moreover, we can apply Young's inequality to show
        \begin{align*}
            \|\prod_{\ell\in\supp(E_a)}&(N_\ell+I)^{(j_\ell+j'_\ell)/2}\ket{\psi}\|\\
            &\leq\|\sum_{\ell\in\supp(E_a)}\frac{j_{\ell}+j'_{\ell}}{\sum_{\ell'\in\supp(E_a)}(j_{\ell'}+j'_{\ell'})}(N_\ell+I)^{\sum_{\ell\in\supp(E_a)}(j_\ell+j'_\ell)/2}\ket{\psi}\|\\
            &\leq\max_{\ell\in\supp(E_a)}\|(N_\ell+I)^{d/2}\ket{\psi}\|\,,
        \end{align*}
        where we applied triangle inequality in the last step. Combining the above inequalities show
        \begin{align*}
            \|H\ket{\psi}\|&\leq \sum_{a=1}^{M}\sum_{\substack{\mathbf{j},\mathbf{j}' \in \mathbb{N}^{ \mathfrak{k}}\\ \|\mathbf{j}+\mathbf{j}'\|_1\le d}}d!^{\mathfrak{k}/2}\max_{\ell\in\supp(E_a)}\|(N_\ell+I)^{d/2}\ket{\psi}\|\,\\
            &\leq M\binom{ \mathfrak{k}+d}{d}^2d!^{\mathfrak{k}/2}\max_{\ell\in\supp(H)}\|(N_\ell+I)^{d/2}\ket{\psi}\|\\
            &\leq M( \mathfrak{k}+d-1)^{2d}d!^{\mathfrak{k}/2-1}\max_{\ell\in\supp(H)}\|(N_\ell+I)^{d/2}\ket{\psi}\|\,,
        \end{align*}
        where we followed the same combinatorial argument as in Proposition \ref{prop:few-modes-adiabatic-limit} and $\supp(H)=\bigcup_{a=1}^M\supp(E_A)$. Since $\mathscr{H}_f$ is a core of the operator $(N_\ell+I)^{d/2}$ (see \cite[lem.~2.12]{Gondolfetal2023}), we complete the input space such that the inequality holds for all $\ket{\psi}\in\operatorname{dom}(\sum_{\ell\in\{1,...,m\}}N_\ell^{d/2})$, which finishes the proof.
    \end{proof}
    Next, we prove two relative bounds: First, how the photon dissipation relatively bounds the number operator to the power $r$ and second by the modified photon dissipation (see  \eqref{eq:def-modified-dissipation}):
    \begin{lem}\label{lem:rel-bounded-learning-N-photon-loss}
        For $r\in\N$, the photon dissipation $L=b^r-\alpha^r$ given by an $\alpha\in\C$, any $\delta\in(0,1)$ and any state $|\psi\rangle\in \operatorname{dom}(L^\dagger L)$, we show:
        \begin{equation*}
            \begin{aligned}
                \|(N+I)^r\ket{\psi}\|\leq\begin{cases}
                    \frac{1}{\sqrt{1-\delta}}\|L^\dagger L\ket{\psi}\|+\sqrt{\frac{c_r}{2r(1-\delta)}\Bigl(\frac{2r-1}{2r}\frac{c_r}{\delta}\Bigr)^{2r-1}}\|\ket{\psi}\| \quad& \text{for}\qquad r\geq2\\
                    \frac{1}{\sqrt{1-\delta}}\|L^\dagger L\ket{\psi}\|+\sqrt{\frac{c_1}{4(1-\delta)}\left(\frac{3}{4}\frac{c_1}{\delta}\right)^3}\|\ket{\psi}\|\quad& \text{for}\qquad r=1
                \end{cases}
            \end{aligned}
        \end{equation*}
        with $c_r=(r+1)r+4|\alpha|^rr!^{3/2}+2\sqrt{(2r)!}|\alpha|^{2r}$.
    \end{lem}
    \begin{proof}
        Assume that $\ket{\psi}\in\mathscr{H}_f$. In a first step, we rewrite $L^\dagger L$ using
        the identities
        \begin{equation}\label{eq:a*a-number-op}
            \begin{aligned}
                &(b^\dagger)^rb^r=(N-(r-1)I)(N-(r-2)I)\cdots(N-I)N\equiv N[-r+1:0]\,,\\
                & b^r(b^\dagger)^r=(N+I)(N+2I)\cdots(N+(r-1)I)(N+rI)\equiv N[1:r]
            \end{aligned}
        \end{equation}
        with
        \begin{equation}
            N[j:k]:=(N+j)\dots (N+k)\qquad\text{and}\qquad n[j:k]:=(n+j)\dots (n+k)\,.
        \end{equation}
        Then, we rewrite $L^\dagger L$ as
        \begin{align*}
            ((b^\dagger)^r-\overline{\alpha}^r)(b^r-\alpha^r)&=N[-r+1:0]+|\alpha|^{2r}-(b^\dagger)^{r}\alpha^{r}-\overline{\alpha}^{r}b^{r}\\
            &=A-(b^\dagger)^{r}\alpha^{r}-\overline{\alpha}^{r}b^{r}\,.
        \end{align*}
        with $A\coloneqq N[-r+1:0]+|\alpha|^{2r}$. In a next step, we rewrite the product $(L^\dagger L)^2$:
        \begin{equation*}
            \begin{aligned}
                (L^\dagger L)^2&=\underbrace{A^2+|\alpha|^{2r}((b^\dagger)^rb^r+b^r(b^\dagger)^r)}_{I}\\
                &\qquad-\underbrace{\bigl(\alpha^r\{A,(b^\dagger)^r\}+\overline{\alpha}^r\{A,b^r\}\bigr)}_{II}\\
                &\qquad-\underbrace{(-\alpha^{2r}(b^\dagger)^{2r}-\overline{\alpha}^{2r}b^{2r})}_{III}\,.
            \end{aligned}
        \end{equation*}
        Then, we apply the following commutation relation for a real-valued measurable function $f:\mathbb{Z}\to\mathbb{R}$, 
        \begin{equation}\label{eq:ccr-function}
            \begin{aligned}
                b&f(N+jI)&&=&&f(N + (j+1)I)b,&\quad\quad b^\dagger&f(N-jI)\,&&=&&f(N - (j+1)I)b^\dagger\,,\\
                &f(N-jI)b\,&&=&b&f(N - (j+1)I)\,,&\quad\quad &f(N+jI)b^\dagger&&=&b^\dagger&f(N + (j+1)I)\,,
        	\end{aligned}
        \end{equation}
        to $II$ and $III$, and again  \eqref{eq:a*a-number-op} to $I$:
        \begin{equation*}
            \begin{aligned}
                I&=N[-r+1:0]^2+|\alpha|^{2r}(3N[-r+1:0]+N[1:r])+|\alpha|^{4r}\\
                II&=\alpha^r(N[-r+1:0]+N[-2r+1:-r])(b^\dagger)^r+\overline{\alpha}^rb^r(N[-r+1:0]+N[-2r+1:-r])\,,\\
                III&=-\alpha^{2r}(b^\dagger)^{2r}-\overline{\alpha}^{2r}b^{2r}\,.
            \end{aligned}
        \end{equation*}
        Then, we prove lower bounds on each sum $I$, $II$, and $III$. For $I$, we first restate Lemma \ref{lem:poly-lower-bound}:
        \begin{equation}
            \begin{aligned}
                N[-r+1:0]^2\geq(N+I)^{2r}-(r+1)r(N+I)^{2r-1}\,.
            \end{aligned}
        \end{equation}
        This directly implies
        \begin{equation*}
            \begin{aligned}
                I&\geq N[-r+1:0]^2\\
                &\geq(N+I)^{2r}-(r+1)r(N+I)^{2r-1}\,.
            \end{aligned}
        \end{equation*}
        For the second and third sum, we apply \cite[Lemma~B.3]{Gondolfetal2023}:
        \begin{equation*}
            \begin{aligned}
                II&=\alpha^r(N[-r+1:0]+N[-2r+1:-r])(b^\dagger)^r+\overline{\alpha}^rb^r(N[-r+1:0]+N[-2r+1:-r])\\
                &\leq 2|\alpha|^r\sqrt{N[1:r]}\Bigl(N[1:r]+N[-r+1:0]\Bigr)\\
                &\leq 4|\alpha|^r(r!)^{3/2}(N+I)^{3r/2}
            \end{aligned}
        \end{equation*}
        and
        \begin{equation*}
            \begin{aligned}
                III&=-\alpha^{2r}(b^\dagger)^{2r}-\overline{\alpha}^{2r}b^{2r}\leq 2|\alpha|^{2r}\sqrt{N[1:2r]}\leq2|\alpha|^{2r}\sqrt{(2r)!}(N+I)^{r}\,.
            \end{aligned}
        \end{equation*}
        Overall, we achieve
        \begin{equation*}
            \begin{aligned}
                (L^\dagger& L)^2\\
                &\geq(N+I)^{2r}-(r+1)r(N+I)^{2r-1}-4|\alpha|^r(r!)^{3/2}(N+I)^{3r/2}-2|\alpha|^{2r}\sqrt{(2r)!}(N+I)^{r}\\
                &\geq(N+I)^{2r}-\Bigl(\underbrace{(r+1)r+4|\alpha|^rr!^{3/2}+2\sqrt{(2r)!}|\alpha|^{2r}}_{\eqqcolon c_r}\Bigr)\begin{cases}
                    (N+I)^{2r-1}&\text{for}\qquad r\geq2\\
                    (N+I)^{3/2}&\text{for}\qquad r=1
                \end{cases}\\
                &\eqqcolon (N+I)^{2r}-\begin{cases}
                    c_r(N+I)^{2r-1}&\text{for}\qquad r\geq2\\
                    c_1(N+I)^{3/2}&\text{for}\qquad r=1
                \end{cases}\,.
            \end{aligned}
        \end{equation*}
        Finally, we bound the second term by a constant depending on $c$, $r$, and $\delta$ via
        \begin{align*}
            (L^\dagger L)^2\geq&(1-\delta)(N+I)^{2r}+\begin{cases}
                    \delta(N+I)^{2r}-c_r(N+I)^{2r-1}&\text{for}\qquad r\geq2\\
                    \delta(N+I)^{2r}-c_1(N+I)^{3/2}&\text{for}\qquad r=1
                \end{cases}
        \end{align*}
        and minimizing the following functions
        \begin{equation*}
            n^{2r}-\frac{c_r}{\delta}n^{2r-1}\qquad\text{and}\qquad n^{2}-\frac{c_1}{\delta}n^{3/2}\,.
        \end{equation*}
        For that, we redefined $p=2r-1$ and $x^2=n$ so that the minimization reduces to
        \begin{equation*}
            n^{p+1}-\frac{c_r}{\delta}n^{p}\qquad\text{and}\qquad x^{4}-\frac{c_1}{\delta}x^{3}\,,
        \end{equation*}
        which attains its minimums at $n=pc_r/((p+1)\delta)$ and $x=3c_1/(4\delta)$ (see \cite[Lemma 11]{mobus2023dissipation})
        so that 
        \begin{equation*}
            n^{p+1}-\frac{c_r}{\delta}n^{p}\geq-\frac{c_r}{\delta(p+1)}\left(\frac{p}{p+1}\frac{c_r}{\delta}\right)^p\qquad\text{and}\qquad x^{4}-\frac{c_1}{\delta}x^{3}\geq-\frac{c_1}{4\delta}\left(\frac{3}{4}\frac{c_1}{\delta}\right)^3\,.
        \end{equation*}
        Substituting $p$ by $2r-1$ and $x^2$ by $n$, we have shown
        \begin{equation*}
            \delta\Bigl(n^{2r}-\frac{c_r}{\delta}n^{2r-1}\Bigr)\geq-\frac{c_r}{2r}\left(\frac{2r-1}{2r}\frac{c_r}{\delta}\right)^{2r-1}\qquad\text{and}\qquad \delta\Bigl(n^{2}-\frac{c_1}{\delta}n^{3/2}\Bigr)\geq-\frac{c_1}{4}\left(\frac{3}{4}\frac{c_1}{\delta}\right)^3\,.
        \end{equation*}
        Therefore, the definition of the norm and subadditivity of the square root, i.e.
        \begin{equation*}
            \begin{aligned}
                \|(N+I)^r\ket{\psi}\|&=\sqrt{\bra{\psi}(N+I)^{2r}\ket{\psi}}\\
                &\leq\frac{1}{\sqrt{1-\delta}}\sqrt{\bra{\psi}\biggl((L^\dagger L)^{2}+\frac{c_r}{2r}\Bigl(\frac{2r-1}{2r}\frac{c_r}{\delta}\Bigr)^{2r-1}\biggr)\ket{\psi}}\\
                &\leq\frac{1}{\sqrt{1-\delta}}\|L^\dagger L\ket{\psi}\|+\sqrt{\frac{c_r}{2r(1-\delta)}\Bigl(\frac{2r-1}{2r}\frac{c_r}{\delta}\Bigr)^{2r-1}}\|\ket{\psi}\|
            \end{aligned}
        \end{equation*}
        finishes the case $r\geq2$ with $(r+1)r+4|\alpha|^rr!^{3/2}+2\sqrt{(2r)!}|\alpha|^{2r}$. Similarly, we achieve
        \begin{equation*}
            \begin{aligned}
                \|(N+I)^r\ket{\psi}\|&\leq\frac{1}{\sqrt{1-\delta}}\|L^\dagger L\ket{\psi}\|+\sqrt{\frac{c_1}{4(1-\delta)}\left(\frac{3}{4}\frac{c_1}{\delta}\right)^3}\|\ket{\psi}\|
            \end{aligned}
        \end{equation*}
        in the case $r=1$ and with $2+4|\alpha|+2\sqrt{2}|\alpha|^{2}$. Finally, we note that $L^\dagger L$ can also relatively upper bounded by the number operator $(N+I)^r$ (see \cite[Lem.~E.2]{Gondolfetal2023}) meaning that the graph norms are equivalent. Since $\mathscr{H}_f$ is a core of $(N+I)^r$ it is also a core of $L^\dagger L$, which finishes the proof by completing $\mathscr{H}_f$ with respect to the graph norm of $L^\dagger L$.
    \end{proof}
    Next, we follow a similar strategy to bound the number operator to the power of order $r+1$ by the modified photon dissipation.
    \begin{lem}\label{lem:rel-bounded-learning-N-L}
        For $r\in\N$, the modified photon dissipation is defined by $L_r=b^r(b-\alpha)$ for an $\alpha\in\C$. Then, $L'=L_r^\dagger L_r+L_1^\dagger L_1$ satisfies for any $\delta\in(0,1)$ and any state $|\psi\rangle\in \operatorname{dom}(L')$
        \begin{equation*}
            \begin{aligned}
                \|(N+I)^{r+1}\ket{\psi}\|&\leq\frac{1}{\sqrt{1-\delta}}\|L'\ket{\psi}\|+\sqrt{\frac{c}{(4r+4)(1-\delta)}\left(\frac{4r+3}{4r+4}\frac{c}{\delta}\right)^{4r+3}}\|\ket{\psi}\|
            \end{aligned}
        \end{equation*}
        with $c=(r+2)(r+1)+16|\alpha|(1+|\alpha|^2)+8\sqrt{2}|\alpha|^2$\,. 
    \end{lem}
    \begin{proof}
        First, note that it is assumed that $\ket{\psi}\in \mathscr{H}_f$ and $r\geq1$. Otherwise, we can directly apply Lemma \ref{lem:rel-bounded-learning-N-photon-loss}. Then, we rewrite $L_r^\dagger L_r$ again using
        the identities in \eqref{eq:a*a-number-op} by
        \begin{align*}
            (b^\dagger-\overline{\alpha})&(b^\dagger)^rb^r(b-\alpha)\\
            &=(b^\dagger)^{r+1}b^{r+1}-(b^\dagger)^{r+1}b^r\alpha-\overline{\alpha}(b^\dagger)^rb^{r+1}+|\alpha|^2(b^\dagger)^rb^{r}\\
            &=N[-r:0]-\alpha b^\dagger N[-r+1:0]-\overline{\alpha}N[-r+1:0]b+|\alpha|^2N[-r+1:0]\\
            &=N[-r+1:0]\big(N-r+|\alpha|^2\big)-\overline{\alpha}bN[-r:-1]-\alpha N[-r:-1]b^\dagger\\
            &\coloneqq A_r(N)-bB_r(N)-B_r(N)^\dagger b^\dagger\,.
        \end{align*}
        with $A_r(N)\coloneqq N[-r+1:0]\big(N-r+|\alpha|^2\big)$ and $B_r(N)\coloneqq\overline{\alpha}N[-r:-1]$. This shows, 
        \begin{equation*}
            \begin{aligned}
                L^\dagger_rL_r+L^\dagger_1L_1=A(N)-bB(N)-B(N)^\dagger b^\dagger
            \end{aligned}
        \end{equation*}
        with $A(N)=A_r(N)+A_1(N)$ and $B(N)=B_r(N)+B_1(N)$. For the sake of notation, we drop the input variable $N$ of the functions $A$ and $B$ in the next line
        \begin{equation*}
            \begin{aligned}
                (L^\dagger L)^2&=A^2-AbB-AB^\dagger b^\dagger-bBA+(bB)^2+bBB^\dagger b^\dagger-B^\dagger b^\dagger A+B^\dagger b^\dagger bB+(B^\dagger b^\dagger)^2\\
                &=\underbrace{A^2+bBB^\dagger b^\dagger+B^\dagger b^\dagger bB}_{I}\\
                &\qquad-\underbrace{(AbB+AB^\dagger b^\dagger+bBA+B^\dagger b^\dagger A)}_{II}\\
                &\qquad-\underbrace{(-(bB)^2-(B^\dagger b^\dagger)^2)}_{III}\,.
            \end{aligned}
        \end{equation*}
        Then, we apply the commutation relations in \eqref{eq:ccr-function} to $I$, $II$, and $III$. First, 
        \begin{equation*}
            I=A(N)^2+B(N+I)B(N+I)^\dagger (N+I)+B^\dagger(N) NB(N)\,.
        \end{equation*}
        Second,
        \begin{equation*}
            \begin{aligned}
                II=b(A(N-I)B(N)+A(N)B(N))+(A(N-I)B(N)+A(N)B(N))^\dagger b^\dagger\,,
            \end{aligned}
        \end{equation*}
        where we used that $A^\dagger=A$ holds. Third, 
        \begin{equation*}
            III=-b^2B(N-I)B(N)-(B(N-I)B(N))^\dagger (b^\dagger)^2\,.
        \end{equation*}
        Then, we find lower bounds on $I$, $II$, and $III$ separately. For $I$, we use the bound in Lemma \ref{lem:poly-lower-bound}, i.e.
        \begin{equation*}
            \begin{aligned}
                N[-r+1:0]^2\geq(N+I)^{2r}-(r+1)r(N+I)^{2r-1}\,.
            \end{aligned}
        \end{equation*}
        to show
        \begin{equation*}
            \begin{aligned}
                A(N)^2&+B(N+I)B(N+I)^\dagger (N+I)+B^\dagger(N) NB(N)\\
                &=\Bigl(N[-r+1:0](N-r+|\alpha|^2)+N(N-1+|\alpha|^2)\Bigr)^2\\
                &\qquad+|\alpha|^2(N[-r+1:0]+N)^2(N+1)\\
                &\qquad+|\alpha|^2(N[-r:-1]+N-I)^2N\\
                &\geq N[-r:0]^2\\
                &\geq(N+I)^{2r+2}-(r+2)(r+1)(N+I)^{2r+1}\,,
            \end{aligned}
        \end{equation*}
        where the bounds above are applied with respect to the spectrum because all operators are diagonalizable in the Fock basis. For the second and third sum, we apply \cite[Lemma B.3]{Gondolfetal2023}:
        \begin{equation*}
            \begin{aligned}
                b(&A(N-I)+A(N))B(N)+((A(N-I)+A(N))B(N))^\dagger b^\dagger\\
                &\leq 2|\alpha|\sqrt{N+1}\Bigl(N[-r+1:0](N-r+|\alpha|^2)+N(N-1+|\alpha|^2)\\
                &\qquad\qquad+N[-r+2:1](N+1-r+|\alpha|^2)+(N+I)(N+|\alpha|^2)\Bigr)(N[-r+1:0]+N)\\
                &\leq 16|\alpha|\sqrt{N+1}(N+I)^{2r+1}(1+|\alpha|^2)
            \end{aligned}
        \end{equation*}
        and
        \begin{equation*}
            \begin{aligned}
                -b^2&B(N-I)B(N)-(B(N-I)B(N))^\dagger (b^\dagger)^2\\
                &\leq 2|\alpha|^2\sqrt{(N+I)(N+2I)}(N[-r+1:0]+N)(N[-r+2:1]+N+I)\\
                &\leq8\sqrt{2}|\alpha|^2(N+I)^{2r+1}\,.
            \end{aligned}
        \end{equation*}
        Overall, we achieve
        \begin{equation}\label{eq:rel-boundedness-operator-inequality}
            \begin{aligned}
                (L^\dagger L)^2&\geq(N+I)^{2r+2}-(r+2)(r+1)(N+I)^{2r+1}\\
                &\qquad\qquad-16|\alpha|(N+I)^{2r+3/2}(1+|\alpha|^2)-8\sqrt{2}|\alpha|^2(N+I)^{2r+1}\\
                &\geq(N+I)^{2r+2}-\Bigl(\underbrace{(r+2)(r+1)+16|\alpha|(1+|\alpha|^2)+8\sqrt{2}|\alpha|^2}_{\eqqcolon c}\Bigr)(N+I)^{2r+3/2}\\
                &\eqqcolon (N+I)^{2r+2}- c(N+I)^{2r+3/2}\,.
            \end{aligned}
        \end{equation}
       Finally, we find a negative constant depending on $c$, $r$, and $\delta$ by employing the leading order, i.e.
        \begin{align*}
            (N+I)^{2r+2}-c(N+I)^{2r+3/2}=(1-\delta)(N+I)^{2r+2}+\delta\left((N+I)^{2r+2}-\frac{c}{\delta}(N+I)^{2r+3/2}\right)
        \end{align*}
        and lower bounding the last two terms, which reduces to 
        \begin{equation*}
            n^{2r+2}-\frac{c}{\delta}n^{2r+3/2}\,.
        \end{equation*}
        By substituting $n=x^2$ and $p+1=4r+4$, the minimization is
        \begin{equation*}
            x^{p+1}-\frac{c}{\delta}x^{p}\,.
        \end{equation*}
        which is solved in \cite[Lemma 11]{mobus2023dissipation}. 
        Therefore, the minimum is achieved at $x=p/(p+1)\frac{c}{\delta}$ with
        \begin{equation*}
            x^{p+1}-\frac{c}{\delta}x^{p}\geq-\frac{c}{\delta(p+1)}\left(\frac{p}{p+1}\frac{c}{\delta}\right)^p\,.
        \end{equation*}
        so that
        \begin{equation}\label{eq:rel-bound-negative-constant}
            \delta\left((N+I)^{2r+2}-\frac{c}{\delta}(N+I)^{2r+\frac{3}{2}}\right)\geq-\frac{c}{4r+4}\left(\frac{4r+3}{4r+4}\frac{c}{\delta}\right)^{4r+3} I\,.
        \end{equation}
        As before, this leads to the following relative bound
        \begin{equation*}
            \begin{aligned}
                \|(N+I)^{r+1}\ket{\psi}\|&\leq\frac{1}{\sqrt{1-\delta}}\|L^\dagger L\ket{\psi}\|+\sqrt{\frac{c}{(4r+4)(1-\delta)}\left(\frac{4r+3}{4r+4}\frac{c}{\delta}\right)^{4r+3}}\|\ket{\psi}\|
            \end{aligned}
        \end{equation*}
        with $c=(r+2)(r+1)+16|\alpha|(1+|\alpha|^2)+8\sqrt{2}|\alpha|^2$. Finally, we argue similar to Lemma \ref{lem:rel-bounded-learning-N-photon-loss} and note that $L'$ can also be relatively upper bounded by the number operator $(N+I)^{r+1}$ (see \cite[Lem.~E.2]{Gondolfetal2023}). Since, $\mathscr{H}_f$ is a core of $(N+I)^{r+1}$ it is also a core of $L'$, which finishes the proof by completing $\mathscr{H}_f$ with respect to the graph norm of $L^\dagger L$.
    \end{proof}
    In the next part of this appendix, we apply the above results to specific generators used in our applications. We start with the single mode example of the $d$-photon dissipation controlling any Hamiltonian of degree $2d$, for which the following bound holds true by Lemma \ref{lem:rel-bounded-learning-H-N} and \ref{lem:rel-bounded-learning-N-photon-loss}:
    \begin{cor}\label{cor:rel-bounded-gate}
        Any single-mode Hamiltonian $H$ of degree $d$ given in Definition \ref{def-main:bosonic-H} is relatively bounded by the $r$-photon dissipation defined by $L=b^r-\alpha^r$ with $\alpha\in\C$ if $r\geq d/2$ :
        \begin{equation*}
            \|H\ket{\psi}\|\leq Md^{2d}d!^{-1/2}\begin{cases}
                    \frac{1}{\sqrt{1-\delta}}\|L^\dagger L\ket{\psi}\|+\sqrt{\frac{c_r}{2r(1-\delta)}\Bigl(\frac{2r-1}{2r}\frac{c_r}{\delta}\Bigr)^{2r-1}}\|\ket{\psi}\| \quad& \text{for}\quad r\geq2\\
                    \frac{1}{\sqrt{1-\delta}}\|L^\dagger L\ket{\psi}\|+\sqrt{\frac{c_1}{4(1-\delta)}\left(\frac{3}{4}\frac{c_1}{\delta}\right)^3}\|\ket{\psi}\|\quad& \text{for}\quad r=1
                \end{cases}
        \end{equation*}
        for any $\delta\in(0,1)$, state $|\psi\rangle\in \operatorname{dom}(L^\dagger L)$ and $c_r=(r+1)r+4|\alpha|^rr!^{3/2}+2\sqrt{(2r)!}|\alpha|^{2r}$.
    \end{cor}
    \begin{proof}
        The result follows by first applying Lemma \ref{lem:rel-bounded-learning-H-N} to the Hamiltonian and then Lemma \ref{lem:rel-bounded-learning-N-photon-loss} to the upper bound given by the number operator.
    \end{proof}
    In a similar manner, we can prove that any Hamiltonian given in Definition \ref{def-main:bosonic-H} can be relatively bounded by the modified photon dissipation using Lemma \ref{lem:rel-bounded-learning-H-N} and \ref{lem:rel-bounded-learning-N-L}: 
    \begin{cor}\label{cor:rel-bounded-learning}
        Let $H$ be any Hamiltonian of degree $d$ given in Definition \ref{def-main:bosonic-H} and $L'=\sum_{\ell\in\{1,...,m\}}L_{r,\alpha_\ell}^\dagger L_{r,\alpha_\ell}+L_{1,\alpha_\ell}^\dagger L_{1,\alpha_\ell}$, $L_{r,\alpha_\ell}=b_\ell^{r}(b_\ell-\alpha_\ell)$  with $\alpha\in\mathbb{C}^{m}$ and $r\geq d/2-1$. Then, for any state $|\psi\rangle\in \operatorname{dom}(L')$ and $\delta\in(0,1)$
        \begin{equation*}
            \begin{aligned}
                \|H\ket{\psi}\|\leq M( \mathfrak{k}+d-1)^{2d}d!^{\mathfrak{k}/2-1}\biggl(&\frac{1}{\sqrt{1-\delta}}\|L'\ket{\psi}\|\\
                &+\sqrt{\frac{c}{(4r+4)(1-\delta)}\left(\frac{4r+3}{4r+4}\frac{c}{\delta}\right)^{4r+3}}\|\ket{\psi}\|\biggr)
            \end{aligned}
        \end{equation*}
        with $c=(r+2)(r+1)+16\|\alpha\|_\infty(1+|\alpha|^2)+8\sqrt{2}\|\alpha\|_\infty^2$.
    \end{cor}
    \begin{proof}
        The result follows by first applying Lemma \ref{lem:rel-bounded-learning-H-N} to the Hamiltonian, which shows 
        \begin{equation*}
            \vspace*{-1ex}\|H\ket{\psi}\|\leq M( \mathfrak{k}+d-1)^{d}d!^{\mathfrak{k}/2-1}\max_{\ell\in\supp(H)}\|(N_\ell+I)^{d/2}\ket{\psi}\|
        \end{equation*}
        with \begin{equation*}
            \supp(H)=\bigcup_{a=1}^M\supp(E_A).
        \end{equation*}
        Then, we relatively bound any power of the local number operator with the help of Lemma \ref{lem:rel-bounded-learning-N-L} and take the maximum over the coefficients, which finishes the proof. Note that we can bound $\bra{\psi} (L_{r,\alpha_\ell}^\dagger L_{r,\alpha_\ell}+L_{1,\alpha_\ell}^\dagger L_{1,\alpha_\ell})^2 \ket{\psi} \leq \bra{\psi} (L')^2 \ket{\psi}$ since terms on different modes commute.
    \end{proof}

    Finally, we state one auxiliary Lemma used in the above proofs: 
    \begin{lem}\label{lem:poly-lower-bound}
        For any $r\in \N$, $r\geq1$ and $x\geq0$,
        \begin{equation*}
            (x-r+1)\cdots(x-1)x\geq(x+1)^{2r}-(r+1)r(x+1)^{2r-1}
        \end{equation*}
        which directly implies the operator inequality
        \begin{equation}
            \begin{aligned}
                N[-r+1:0]^2\geq(N+I)^{2r}-(r+1)r(N+I)^{2r-1}\,.
            \end{aligned}
        \end{equation}
    \end{lem}
    \begin{proof}
        To prove the above result, we follow the proof of \cite[Lemma~C.3]{Gondolfetal2023}. First, we define 
        \begin{equation*}
            \begin{aligned}
                p_{\ell}(y)&=(y-\ell)^2\cdots(y-1)^2\\
                &=(y^2-2y\ell+\ell^2)\cdots(y^2-2y+1)\\
                &\eqqcolon y^{2\ell}-2\Bigl(\sum_{k=1}^\ell k\Bigr)y^{2\ell-1}+r_{2(\ell-1)}(y)\\
                &\eqqcolon y^{2\ell}-(\ell+1)\ell y^{2\ell-1}+r_{2(\ell-1)}(y)\,.
            \end{aligned}
        \end{equation*}
        Next, we prove by induction over $\ell$ that $r_{2(\ell-1)}(y)$ is positive for $y\geq\ell-1$. The induction start is trivial because $p_\ell$ for $\ell=1$ reduces to the binomial formula. Next, assume that $r_{2(\ell-1)}(y)\geq0$ for all $y\geq\ell-1$. Then,
        \begin{equation*}
            \begin{aligned}
                &p_{\ell+1}(y)\\
                &=(y-\ell-1)^2p_{\ell}(y)\\
                &=y^2p_{\ell}(y)-2y(\ell+1)p_{\ell}(y)+(\ell+1)^2p_{\ell}(y)\\
                &=y^{2(\ell+1)}-(\ell+2)(\ell+1)y^{2(\ell+1)-1}+(y^2-2y(\ell+1))r_{2(\ell-1)}(y)+(\ell+1)^2(p_{\ell}(y)+2\ell y^{2\ell})\\
                &=y^{2(\ell+1)}-(\ell+2)(\ell+1)y^{2(\ell+1)-1}+r_{2\ell}(y)
            \end{aligned}
        \end{equation*}
        for $r_{2\ell}(y) = y(y-2(\ell+1))r_{2(\ell-1)}(y)+(\ell+1)^2(p_{\ell}(y)+2\ell y^{2\ell})$. This can be further simplified by
        \begin{equation*}
            \begin{aligned}
                y&(y-2(\ell+1))r_{2(\ell-1)}(y)+(\ell+1)^2(p_{\ell}(y)+2\ell y^{2\ell})\\
                &=(y^2-2y(\ell+1)+(\ell+1)^2)r_{2(\ell-1)}(y)+(\ell+1)^2(y^{2\ell}-(\ell+1)\ell y^{2\ell-1}+2\ell y^{2\ell})\\
                &=(y-(\ell+1))^2r_{2(\ell-1)}(y)+y^{2\ell-1}(\ell+1)^2((2\ell+1)y-(\ell+1)\ell)\\
                &\geq0
            \end{aligned}
        \end{equation*}
        by assumption $y\geq\ell$ and $r_{2(\ell-1)}\geq0$. This proves
        \begin{equation*}
            \begin{aligned}
                y^{2\ell}-(\ell+1)\ell y^{2\ell-1}\leq p_\ell(y)
            \end{aligned}
        \end{equation*}
        for all $y\geq\ell$. Due to the inequality
        \begin{equation*}
            y^{2\ell}\leq(\ell+1)\ell y^{2\ell-1}
        \end{equation*}
        and $p_\ell(y)=0$ for $1\leq y\leq \ell-1$ and $y\in\N$, the inequality extends to all $\ell\geq 1$ and shoes the following inequality via its spectral decomposition in the Fock basis:        \begin{equation}\label{eq:bounds-on-increasing-product}
                \begin{aligned}
                    N[-r+1:0]^2\geq(N+I)^{2r}-(r+1)r(N+I)^{2r-1}\,.
                \end{aligned}
            \end{equation}
        \end{proof}
    
\section{Polynomial regression bounds}
\label{sec:polyregression}
In the above, we need to regard the expectation values as polynomials in $|\alpha|$, and estimate the coefficients of this polynomial to learn the coefficients in the Hamiltonian. To this end, we will study the following problem: given a polynomial $p(x)$ of degree $d$ that satisfies $|p(x)| \leq \epsilon$ for $x \in [a, b]$, how large can its coefficients be?

\begin{lem}
    \label{lem:polynomial_regression_error}
    Let $p(x) = \sum_{n=0}^d p_n x^n$. If $|p(x)|\leq \epsilon$ for $x\in[a,b]$ and $|b-a|>2|a|$, then
    \begin{equation}
    |p_n| \leq \frac{d(2d-2)!!}{n!}\left|\frac{2}{b-a}\right|^n\frac{\epsilon}{1-2|a|/|b-a|}.
    \end{equation}
\end{lem}

\begin{proof}
Note that $p_n = p^{(n)}(0)/n!$, and therefore we will first focus on upper bounding $|p^{(n)}(0)|$.
The following derivation is almost identical derivation to \cite[Lemma E.2]{StilckFrança2024}, but we still include it here for completeness. By Markov brothers' inequality, we have
\begin{equation}
    |p^{(k)}(a)|\leq \left|\frac{2}{b-a}\right|^k C_M(d,k)\epsilon,
\end{equation}
where
\begin{equation}
    C_M(d,k) = \frac{d^2(d^2-1^2)(d^2-2^2)\cdots (d^2-(k-1)^2)}{(2k-1)!!}\leq d((2d-2)!!),
\end{equation}
where the last inequality comes from $C_M(d,k)\leq C_M(d,d)=d(2d-2)!!$ (one can prove $C_M(d,k)\leq C_M(d,k+1)$ by direct computation).
Performing Taylor expansion at $x=a$, we have
$
    p(x) = \sum_{k=0}^d p^{(k)}(a)\frac{(x-a)^k}{k!}.
$
Therefore
$
p^{(n)}(0) = \sum_{k=n}^d \frac{1}{(k-n)!}p^{(k)}(a)(-a)^{k-n}.
$
Consequently,
\begin{equation}
    |p^{(n)}(0)| \leq \sum_{k=n}^d \frac{1}{(k-n)!}\left|\frac{2}{b-a}\right|^k C_M(d,k)|a|^{k-n}\epsilon.
\end{equation}
Using $(k-n)!\geq 1$ and $C_M(d,k)\leq d(2d-2)!!$, we have
\begin{equation}
    |p^{(n)}(0)| \leq d(2d-2)!!\sum_{k=n}^d \left|\frac{2}{b-a}\right|^k |a|^{k-n}\epsilon\leq d(2d-2)!!\left|\frac{2}{b-a}\right|^n\frac{\epsilon}{1-2|a|/|b-a|},
\end{equation}
if $|b-a|>2|a|$ as assumed in the lemma. 
Therefore we have the inequality we want to prove by observing that $p_n=p^{(n)}(0)/n!$.
\end{proof}

\begin{lem}
    \label{lem:chebyshev_interpolation_error}
    Let $f(x)$ be a polynomial with degree at most $d$, and let $x_i=\cos\left(\frac{(2i-1)\pi}{2(d+1)}\right)$, $i=1,2,\cdots,d+1$, be the roots of the $(d+1)$th Chebyshev polynomial of the first kind. Then for any $x\in[-1,1]$, 
    $
    |f(x)|\leq \frac{2d+1}{d+1}\sum_{i=1}^{d+1}|f(x_i)|.
    $
\end{lem}

\begin{proof}
    Because the Chebyshev polynomials $T_0(x),T_1(x),\cdots,T_{d}(x)$ form an orthogonal basis of the vector space of polynomials with degree at most $d$, we have
    \begin{equation}
        f(x) = \sum_{n=1}^d \frac{2T_n(x)}{\pi} \int_{-1}^{1}\frac{f(y)T_n(y)}{\sqrt{1-y^2}}\dd y + \frac{1}{\pi}\int_{-1}^{1}\frac{f(y)}{\sqrt{1-y^2}}\dd y.
    \end{equation}
    Because Gauss-Chebyshev quadrature with $q$ quadrature points is exact for polynomials up to degree $2q-1$, and $f(y)T_n(y)$ has degree at most $2d$, by choosing $q=d+1$ we have
    \begin{equation}
        \int_{-1}^{1}\frac{f(y)T_n(y)}{\sqrt{1-y^2}}\dd y = \frac{\pi}{d+1}\sum_{i=1}^{d+1}f(x_i)T_{n}(x_i),
    \end{equation}
    for $n=0,1,\cdots,d$. Therefore
    \begin{equation}
        f(x) = \sum_{n=1}^d \frac{2T_n(x)}{d+1} \sum_{i=1}^{d+1}f(x_i)T_{n}(x_i) + \frac{1}{d+1}\sum_{i=1}^{d+1}f(x_i).
    \end{equation}
    Consequently
    \begin{equation}
        |f(x)|\leq \sum_{n=1}^d \frac{2}{d+1} \sum_{i=1}^{d+1}|f(x_i)| + \frac{1}{d+1}\sum_{i=1}^{d+1}|f(x_i)| = \frac{2d+1}{d+1}\sum_{i=1}^{d+1}|f(x_i)|,
    \end{equation}
    where we have used the fact that $|T_n(x)|\leq 1$ for all $x\in[-1,1]$.
\end{proof}

We remark that the above lemma holds for arbitrary interval $[a,b]$ instead of just for $[-1,1]$. We only need to redefine
\begin{equation}
\label{eq:quadrature_pts_ab}
    x_i = \frac{a+b}{2} + \frac{b-a}{2}\cos\left(\frac{(2i-1)\pi}{2(d+1)}\right),
\end{equation}
for $i=1,2,\cdots,d+1$.

Combining the above lemmas, we have the following corollary:
\begin{cor}
    \label{cor:coef_est_error}
    Let $f(x)=\sum_{n=0}^d p_n x^n$ be a polynomial with degree at most $d=\Or(1)$. Assume $|f(x_i)|\leq \epsilon$ for $x_i$, $i=1,2,\cdots,d+1$, defined in \eqref{eq:quadrature_pts_ab} in which $(b-a)>2a>0$, then
    $
    |p_n|=\Or\left(\frac{2^n}{(b-a)^n}\frac{\epsilon}{1-2a/(b-a)}\right).
    $
\end{cor}

We next look at the multivariate case. Instead of going through a detailed calculation as we did for the single variable case, we will use the norm equivalence in finite dimensional vectors spaces to establish a similar bound without specifying the constant. We only need to focus on the dependence on $a_1,\cdots,a_k$, which needs to be tuned according to the required precision.

First we define the Chebyshev quadrature nodes $y_{i_1,\cdots,i_k}=(y_{i_1}^1,\cdots,y_{i_k}^k)$:
\begin{equation}
\label{eq:quadrature_pts_multivariate}
    y_{i_\nu}^{\nu} = \cos\left(\frac{(2i_\nu-1)\pi}{2(d_\nu+1)}\right)
\end{equation}
for $i_\nu=1,2,\cdots,d_\nu+1$, $\nu=1,2,\cdots,k$. The Chebyshev quadrature nodes in the domain $[a_1,b_1]\times\cdots\times[a_k,b_k]$ is then
\begin{equation}
\label{eq:quadrature_pts_ab_multivariate}
    r_{i_\nu}^{\nu} = \frac{a_\nu+b_\nu}{2} + \frac{b_\nu-a_\nu}{2}y_{i_\nu}^{\nu},
\end{equation}

\begin{lem}
    \label{lem:bounding_coef_with_chebyshev_nodes_multivariate}
    Let $r_{i_1,\cdots,i_k}=(r_{i_1}^1,r_{i_2}^2,\cdots,r_{i_k}^k)$ be as defined in \eqref{eq:quadrature_pts_ab_multivariate}. Let $p$ be a polynomial of the form
    $
    p(x) = \sum_{l_1=0}^{d_1}\cdots\sum_{l_k=0}^{d_k} p_{l_1,\cdots,l_k}x^{l_1}\cdots x^{l_k}.
    $
    If $b_\nu-a_\nu\geq 2$ and $b_\nu\geq 2a_\nu>0$ for all $\nu=1,2,\cdots,k$, then there exists a constant $C_{d_1,\cdots,d_k}>0$ that depends only on $d_1,\cdots,d_k$ such that
    \begin{equation}
    \label{eq:bounding_coef_with_chebyshev_nodes}
        \max_{l_1,\cdots,l_k}|p_{l_1,\cdots,l_k}|\leq C_{d_1,\cdots,d_k}\max_{i_1,\cdots,i_k} |p(r_{i_1,\cdots,i_k})|.
    \end{equation}
\end{lem}
 
\begin{proof}
We shift and rescale the domain by defining
$
\varphi(x) = \left(\frac{2x_1-(a_1+b_1)}{b_1-a_1},\cdots, \frac{2x_k-(a_k+b_k)}{b_k-a_k}\right),
$
and 
\begin{equation}
    q(y) = p(\varphi^{-1}(y)).
\end{equation}
Then $p(r_{i_1,\cdots,i_k}) = q(y_{i_1,\cdots,i_k})$.

Because all polynomials with degrees $d_1,d_2,\cdots,d_k$ for variables $x_1,x_2,\cdots,x_k$ form a complex vector space of dimension $(d_1+1)(d_2+1)\cdots(d_k+1)$, it is easy to check that
\begin{equation}
    \|q\|_{\mathrm{nodes}} := \max_{i_1,\cdots,i_k} |q(y_{i_1,\cdots,i_k})|
\end{equation}
is a norm for this vector space. Similarly,
\begin{equation}
      \|q\|_{\max}:= \max_{y\in[-3,3]^k} \max_{l_1,\cdots,l_k}\left|\frac{\partial^{l_1+\cdots+l_k}}{\partial_{y_1}^{l_1}\cdots \partial_{y_k}^{l_k}}q(y)\right|
\end{equation}
also defines a norm. 
Here, we take the maximum over $[-3,3]^k$ because $b_\nu\geq 2a_\nu>0$ ensures $\varphi(0)\in[-3,3]^k$.
By the equivalence between norms for finite-dimensional vector spaces, we then have a constant $C_{d_1,\cdots,d_k}$ that depends only on $d_1,\cdots,d_k$ such that
\begin{equation}
     \|q\|_{\max}\leq C_{d_1,\cdots,d_k}\|q\|_{\mathrm{nodes}}.
\end{equation}

We then relate $\|q\|_{\max}$ to the coefficients of $p(x)$. Note that
$
p_{l_1,\cdots,l_k} = \frac{1}{l_1!\cdots l_k!}\frac{\partial^{l_1+\cdots+l_k}}{\partial_{x_1}^{l_1}\cdots \partial_{x_k}^{l_k}}p(x)\Big|_{x=0} = \prod_{\nu=1}^k\frac{1}{l_\nu!}\left(\frac{2}{b_\nu-a_\nu}\right)^{l_\nu}\frac{\partial^{l_1+\cdots+l_k}}{\partial_{x_1}^{l_1}\cdots \partial_{x_k}^{l_k}}q(y)\Big|_{y=\varphi(0)}.
$
Since $b_\nu-a_\nu\geq 2$ for all $\nu$, and $\varphi(0)\in[-3,3]^k$, which is ensured by $b_\nu\geq 2a_\nu>0$, we have
$
\max_{l_1,\cdots,l_k}|p_{l_1,\cdots,l_k}|\leq \|q\|_{\max}.
$
Further we have
$
\|q\|_{\max}\leq C_{d_1,\cdots,d_k}\|q\|_{\mathrm{nodes}} = C_{d_1,\cdots,d_k}\max_{i_1,\cdots,i_k} |p(r_{i_1,\cdots,i_k})|.
$
Therefore we have \eqref{eq:bounding_coef_with_chebyshev_nodes}.
\end{proof}

\section{Pseudoinverses of unbounded operators}

We consider a separable Hilbert space $\mathcal{H}$. Let $A$ and $B$ be possibly unbounded operators on $\mathcal{H}$, and we want to provide a lower bound for $\eta$ satisfying
\begin{equation}
\label{eq:spectral_gap_lower_bound}
    \braket{\psi|(I-P)B^\dag A^\dag A B(I-P)|\psi}\geq \eta \braket{\psi|I-P|\psi},
\end{equation}
where $P$ is the projection operator onto $\ker(AB)$.

To prove this, we will need the notion of the pseudoinverse. While pseudoinverses are well-known for matrices, they have been extended to unbounded operators on Hilbert spaces \cite{ben1963contributions}.

\begin{defn} \label{def:pseudoinverse}
    Let $\mathcal H_1$, $\mathcal H_2$ be two complex Hilbert spaces. Let $T: \mathcal H_1 \to \mathcal{H}_2$ be a linear operator. Let $P_1$ be the orthogonal projection onto $\overline{\operatorname{ran}(T)}$ and $P_2$ the orthogonal projection onto $\overline{\operatorname{ran}(T^+)}$. Then, the operator $T^+$ is a pseudoinverse of $T$ if $\overline{\operatorname{dom}(T^+)} = \mathcal H_2$ and 
    \begin{align*}
        \operatorname{ran}(T) \subseteq \operatorname{dom}(T^+)\, , &\qquad \operatorname{ran}(T^+) \subseteq \operatorname{dom}(T)\, , \\
        TT^+ = P_1\, , &\qquad T^+T=P_2 \, . 
    \end{align*}
\end{defn} 
The following theorem follows from \cite[Theorem A]{tseng1949generalized} (see also \cite[Theorem 6]{ben1963contributions}):
\begin{lem}
    Let $\mathcal H_1$, $\mathcal H_2$ be two complex Hilbert spaces. Let $T: \mathcal H_1 \to \cH_2$ be a densely defined closed linear operator. Then, $T$ has a unique closed pseudoinverse $T^+$.
\end{lem}
A priori, $T^+$ might be unbounded, but if $\operatorname{ran}(T)$ is closed, we can say more:
\begin{lem} \label{lem:bounded-pseudoinverse}
    Let $\mathcal H_1$, $\mathcal H_2$ be two complex Hilbert spaces. Let $T: \mathcal H_1 \to H_2$ be a densely defined closed linear operator. Then, the pseudoinverse $T^+$ is a bounded operator if and only if $\operatorname{ran}(T)$ is closed.
\end{lem}

In order to calculate $T^+$ explicitly, Theorem 3.2 \cite{lardy1975series} provides a series expansion:
\begin{lem} \label{lem:pseudoinverse-series}
    Let $\mathcal H_1$, $\mathcal H_2$ be two complex Hilbert spaces. Let $T: \mathcal H_1 \to \mathcal H_2$ be a densely defined closed linear operator with closed range. Then,
    \begin{equation*}
        T^+ = \sum_{k=1}^\infty T^\dagger (I_{\mathcal H_2} + T T^\dagger)^{-k}
    \end{equation*}
    in the uniform operator topology.
\end{lem}
In the following, we use the above results to prove rigorous lower bounds on $\eta$ in \eqref{eq:spectral_gap_lower_bound}. The idea is to find a lower bound in terms of the norms of the pseudoinverse of $A$ and $B$ on $\eta$, which is proven in the following result: 

\begin{lem}\label{lemmaetaAB}
    Let $A,B$ be densely defined closed operators in Hilbert space $\mathcal{H}$ with closed range. We assume that $\ker(A)\subset \operatorname{ran}(B)$. Then 
    \begin{equation*}
        \braket{\psi|(I-P)B^\dag A^\dag A B(I-P)|\psi}\geq \eta \braket{\psi|I-P|\psi}
    \end{equation*}
    holds with $\eta=\|B^+\|^{-2}\|A^+\|^{-2}$, where $P$ is the orthogonal projection operator into $\ker(AB) = B^+\ker(A) + \ker(B)$.
\end{lem}

\begin{proof}
    Note that in \eqref{eq:spectral_gap_lower_bound}, we only need to consider $\ket{\psi}$ such that $P\ket{\psi}=0$ and therefore we make this assumption without loss of generality. We denote $\ket{\phi} = AB\ket{\psi}$, and the goal is to relate $\|\ket{\phi}\|$ to $\|\ket{\psi}\|$. We assume that $A$ and $B$ are densely defined closed operators with closed range such that they have bounded pseudoinverses by Lemma \ref{lem:bounded-pseudoinverse}. From $\ket{\phi} = AB\ket{\psi}$, we have
    \begin{equation*}
        B\ket{\psi} =A^+AB\ket{\psi} +(I-A^+A)B\ket{\psi}\eqqcolon A^+\ket{\phi} + \ket{r_1},
    \end{equation*}
    where $\ket{r_1}\in\ker(A)$. Similarly, we achieve
    \begin{equation}\label{eq:solution_of_linear_system_AB_new}
        \ket{\psi} = B^+B\ket{\psi}+(I-B^+B)\ket{\psi}\eqqcolon B^+A^+\ket{\phi} + B^+\ket{r_1} + \ket{r_2},
    \end{equation}
    where $\ket{r_2}\in\ker(B)$. Note that $AB\ket{r_2}=0$, and therefore $\ket{r_2}\in\ker(AB)$. Thus, as a result $(I-P)\ket{r_2}=0$. Definition \ref{def:pseudoinverse} states that $BB^+=P_B$, where $P_B$ is the projection operator onto $\operatorname{ran}(B)$. Therefore, if we further assume that $\ker(A)\subset \operatorname{ran}(B)$, then 
    \begin{equation*}
        ABB^+\ket{r_1} = AP_B\ket{r_1}= A\ket{r_1}=0\,.   
    \end{equation*}
    Therefore, $B^+\ket{r_1} \in \operatorname{ker}(AB)$ and $(I-P)B^+\ket{r_1} = 0$. Now multiplying $I-P$ to both sides of \eqref{eq:solution_of_linear_system_AB_new}, we have
    \begin{equation}
        \ket{\psi} = (I-P)\ket{\psi} = (I-P)B^+A^+\ket{\phi} + (I-P)(B^+\ket{r_1} + \ket{r_2})=(I-P)B^+A^+\ket{\phi}.
    \end{equation}
    Therefore 
    $\|\ket{\psi}\|=\|(I-P)B^+A^+\ket{\phi}\|\leq \|B^+\|\|A^+\|\|\ket{\phi}\|$, which finishes the proof.
\end{proof}

In the context of the bosonic systems introduced in Appendix \ref{subsec:mod-photon-dissipation} the jump operator $L = b^k(b - \alpha)$ split into $A = b - \alpha$ and $B = b^k$. For $k=1$, we can easily verify that $\operatorname{ran}(B) = \mathcal{H}$, $\ker(A) = \operatorname{span}\{\ket{\alpha}\}$, and $\ker(B) = \operatorname{span}\{\ket{0}\}$. Therefore, the assumptions of Lemma \ref{lemmaetaAB} are satisfied. For arbitrary $k\in\N_{\geq1}$, we need that the closed operators $b^k$ have closed range:

\begin{lem}\label{lem:bk-closed-range}
    Let $k \in \mathbb N$, $k \geq 1$. Then, $\operatorname{ran}(b^k)=\mathcal{H}$ is closed.
\end{lem}
\begin{proof}
    Since $\mathcal{H}$ is a separable Hilbert space, for every $\ket{\psi} \in \mathcal{H}$, there exists a sequence $(\psi_n)_{n\in\mathbb{N}}$ such that $\ket{\psi} = \lim_{m\rightarrow\infty}\sum_{n=0}^m \psi_n \ket{n}$. We then define the sequence $(\phi_n)_{n\in\mathbb{N}}$ by 
    \begin{equation*}
        \phi_n =\begin{cases}
            0 & \text{for } n \in \{0, \ldots, k-1\}\\
            \frac{1}{\sqrt{n[k-1:0]}} \psi_{n+k} &\text{for } n \in \{k, k+1, \ldots \}
        \end{cases}
    \end{equation*}
    The sequence $\ket{\phi_m} = \sum_{n=0}^m \phi_n \ket{n}$ is a Cauchy sequence by definition and converges to a state $\ket{\phi} \in \mathcal{H}$. Moreover, $b^k \ket{\phi_n} = \ket{\psi_n}$ converges by definition to $\ket{\psi}$. Since $b^k$ is a closed operator,  because it is the adjoint of $ (b^\dagger)^k $, which is densely defined \cite[Theorem VIII.1]{reed1981functional}, we have $b^k \ket{\phi} = \ket{\psi}$, which completes the proof.
\end{proof}

Next, we show that, for example, $AB = b(b - \alpha)$ satisfies \eqref{eq:spectral_gap_lower_bound} with $\eta = 1$.

\begin{lem}\label{lempseudoinv}
    Let $j \in \mathbb N$ and $A=(b-\alpha)^j$, $B=b^j$. Then, $\|A^+\|_{\infty}$, $\|B^+\|_{\infty} \leq \frac{1}{\sqrt{j!}}$.
\end{lem}
\begin{proof}
    By Lemma \ref{lem:bk-closed-range}, we know that $ b^j $ has closed range. By Lemma \ref{lem:pseudoinverse-series} and the identity $b^j (b^j)^\dagger = (N + I) \cdot \ldots \cdot (N + jI) =: N[1:j]$, we can represent the pseudoinverse of $B$ by 
    \begin{equation*}
        (b^j)^+ = \sum_{k=1}^\infty (b^\dagger)^j (I + N[1:j])^{-k} \, .
    \end{equation*}
    For $\ket{\psi} = \sum_{n=0}^\infty \psi_n \ket{n}$, we can estimate the norm using Tonelli's theorem:
    \begin{align*}
        \norm{B^+ \ket{\psi}}^2 
        &\leq \sum_{k=1}^\infty \left\| \sum_{n=0}^\infty \sqrt{n[1:j]} (n[1:j] + 1)^{-k} \psi_n \ket{n + j} \right\|^2 \\
        &\leq \sum_{k=1}^\infty \sum_{n=0}^\infty \frac{n[1:j]}{(1 + n[1:j])^{2k}} |\psi_n|^2 \\
        &\leq \sum_{n=0}^\infty \frac{1}{2 + n[1:j]} |\psi_n|^2 \leq \frac{1}{j!}\,,
    \end{align*}
    so that $\|B^+\|_\infty \leq \frac{1}{\sqrt{j!}}$. The operators $A$ and $B$ are related through $B = D(\alpha) A D(\alpha)^\dagger$, where $D(\alpha) = e^{\alpha b^\dagger - \overline \alpha b}$ is the bosonic displacement operator. Since the displacement operator in unitary, $\|A^+\|_\infty = \|B^+\|_\infty \leq \frac{1}{\sqrt{j!}}$.
\end{proof}

\begin{lem}
     Let $A=(b-\alpha)$, $B=b$. Then, $B^+\ker(A)$ is spanned by $B^+\ket{\alpha} = \frac{1}{\alpha}(\ket{\alpha}-e^{-|\alpha|^2/2}\ket{0})$. Therefore $\ker(AB) = \operatorname{span}\{\ket{0},\ket{\alpha}\}$.
\end{lem}
\begin{proof}
    It is easy to see that $ \ker(A) $ is spanned by $ \ket{\alpha} $, since $ A = b - \alpha $. Moreover, the continuity of the map $ \bra{\phi} : \mathcal{H} \to \mathbb{C} $, along with Lemma \ref{lem:pseudoinverse-series}, shows that $ \bra{0} B^+ \ket{\alpha} = 0 $. In a similar vein, we can compute for $ n \in \mathbb{N} $:
    \begin{align*}
        \bra{n} B^+ \ket{\alpha} 
        &= \sum_{k=1}^\infty \bra{n} b^\dagger (2I + N)^{-k} \ket{\alpha} \\
        &= \sqrt{n} \sum_{k=1}^\infty (n + 1)^{-k} \braket{n | \alpha} \\
        &= \frac{1}{\sqrt{n}} \braket{n | \alpha} \\
        &= \frac{1}{\alpha} \frac{\alpha^n}{\sqrt{n!}} e^{-|\alpha|^2/2} \, .
    \end{align*}
    Using the definition of the coherent state, the first assertion follows. 
    For the second assertion, we observe from Lemma \ref{lemmaetaAB} that $\ker(AB) = B^+ \ker(A) + \ker(B)$. Since $ \ker(B) $ is spanned by $ \ket{0} $, the second assertion follows from the first.
\end{proof}

\begin{lem} \label{lem:lower-bound-commutator-L}
    Let $L = b^k - \alpha^kI $. Then, $[L, L^\dagger] \geq k! I$.
\end{lem}
\begin{proof}
    This result follows essentially from \cite{Azouit.2016}. There, they prove that 
    \begin{equation*}
        [L, L^\dagger] = (N + I)(N + 2I) \ldots (N + kI) - N (N - I)_+ \ldots (N - (k - 1)I)_+ \, .
    \end{equation*}
    Here, for $ \mu \in \mathbb{N} $, the operator $ (N - \mu I)_+ $ denotes the positive part of $ N - \mu I $, i.e.,  
    \begin{equation}
        (N - \mu I)_+ \ket{n} =
        \begin{cases}
            0 & \text{for } n \in \{0, \ldots, \mu\}, \\
            (n - \mu) \ket{n} & \text{for } n > \mu, \ n \in \mathbb{N}\,.
        \end{cases}
    \end{equation}
    We can rewrite this as
    \begin{align*}
        [L, L^\dagger] =& N (N + 2I)(N + 3I) \ldots (N + kI) + (N + 2I)(N + 3I) \ldots (N + kI)\\& - N (N - I)_+ \ldots (N - (k - 1)I)_+ \, .
    \end{align*}
    We observe that $ (N + 2I)(N + 3I) \ldots (N + kI) \geq k! I $. Moreover, by applying the spectral decomposition, we find that
    \begin{align*}
        &N (N + 2I)(N + 3I) \ldots (N + kI) - N (N - I)_+ \ldots (N - (k - 1)I)_+ \\
        &= N \left[ (N + 2I)(N + 3I) \ldots (N + kI) - (N - I)_+ (N - 2I)_+ \ldots (N - (k - 1)I)_+ \right] \\
        &\geq 0 \, .
    \end{align*}
    This completes the proof.
\end{proof}

\begin{lem}\label{lem:L-closed-range}
    Let $k \in \mathbb N$, $k \geq 1$, $\alpha \in \mathbb C$ and $L = b^k - \alpha^k I$. Then, $L$ is closed and has closed range, where $\operatorname{dom}(L) = \operatorname{dom}(b^k)$.
\end{lem}
\begin{proof}
    We know that $ b^k $ is a closed operator because it is the adjoint of $ (b^\dagger)^k $, which is densely defined \cite[Theorem VIII.1]{reed1981functional}. Moreover, $ -\alpha^k I $ is a bounded operator. Thus, $ L $ is closed by Problem III.5.6 of \cite{kato2013perturbation}. To prove that $ L $ also has closed range, we consider the reduced minimum modulus of $L^\dagger$ (see \cite[Section III.5.1]{kato2013perturbation}), defined for some closed operator $ T $ as  
    \begin{equation*}
        \gamma(T) := \sup \{ \gamma \geq 0 \mid \|Tu\| \geq \gamma \|\tilde u\| \quad \forall u \in \operatorname{dom}(T) \} \, ,
    \end{equation*}
    where $ \|\tilde u\| = \inf_{z \in \ker(T)} \|u - z\| $. Using Lemma \ref{lem:lower-bound-commutator-L} and the fact that $ \ker(L^\dagger) = \{0\} $, we obtain
    \begin{equation*}
        \|L^\dagger u \|^2 = \bra{u} L L^\dagger \ket{u} \geq \bra{u} [L, L^\dagger] \ket{u} \geq k! \|u\|^2\,.
    \end{equation*}
    It follows that $ \gamma(L^\dagger) \geq \sqrt{k!} $. By Theorem IV.5.2 of \cite{kato2013perturbation}, $ \gamma(L^\dagger) > 0 $ implies that $ \operatorname{ran}(L^\dagger) $ is closed. By Theorem IV.5.13 of \cite{kato2013perturbation}, this further implies that $ \operatorname{ran}(L) $ is closed as well.
\end{proof}

\begin{lem}\label{lem:pseudoinverse-bound-photon-dissipation}
    Let $k \in \mathbb N$, $\alpha \in C$ and $L = b^k - \alpha^kI $. Then, $\|L^+\|_\infty\leq \frac{1}{\sqrt{k!}}$.
\end{lem}
\begin{proof}
   We find by Lemma \ref{lem:L-closed-range} that $L$ has closed range. We aim to apply Lemma \ref{lem:pseudoinverse-series} for the pseudoinverse of $ L^\dagger $. Let $ \ket{\psi} \in \operatorname{dom}(L) $. Then, by \cite[Theorem 5.1.9(iii)]{pedersen2012analysis}, it holds that
    \begin{equation*}
        L (I + L^\dagger L)^{-1} \ket{\psi} = (I + L L^\dagger)^{-1} L \ket{\psi}
    \end{equation*}
    (see also the proof of Theorem 3.2 in \cite{lardy1975series} for a direct argument). Now define 
    \begin{equation*}
        A_n := \sum_{r=1}^n (I + L L^\dagger)^{-r} L \,,
    \end{equation*}
    which is a bounded operator. We claim that $ \|A_n \ket{\psi} \| \leq \frac{1}{\sqrt{k!}} \|\ket{\psi}\| $ for all $ \ket{\psi} \in \operatorname{dom}(L) $. Assuming this claim for now, since $ \operatorname{dom}(L) $ is dense, the BLT theorem implies that $ \|A_n\|_\infty \leq \frac{1}{\sqrt{k!}} $ for all $ n \in \mathbb{N} $. Lemma \ref{lem:pseudoinverse-series} states that the operators $ A_n $ converge to $ (L^\dagger)^+ $ in norm, thus also $ \|(L^\dagger)^+\|_\infty \leq \frac{1}{\sqrt{k!}} $. Since $ (L^\dagger)^+ = (L^+)^\dagger $ (see \cite[Theorem 8(g)]{ben1963contributions}, which is attributed to \cite{tseng1949generalized}), the assertion follows. It remains to prove the claim. We start by bounding
    \begin{align}
        \sum_{r=1}^n (I + L L^\dagger)^{-r} 
        &\leq \sum_{r=1}^n (1+k!)^{-r} I \nonumber \\ 
        &\leq \sum_{r=1}^\infty (1+k!)^{-r} I = \frac{1}{k!} I \label{lem:bound-inner-sum-pseudoinverse}
    \end{align}
    for any $ n \in \mathbb{N} $, where we have used Lemma \ref{lem:lower-bound-commutator-L} in the first inequality since $ I + L L^\dagger \geq I + [L, L^\dagger] $. Next, we consider $ \ket{\psi} \in \operatorname{dom}(L) $ again and find that
    \begin{align*}
        \bra{\psi} A_n^\dagger A_n \ket{\psi} &= \bra{\psi} L^\dagger \left(\sum_{r=1}^n (I + L L^\dagger)^{-r}\right)^2 L \ket{\psi} \\
        &\leq \frac{1}{k!} \bra{\psi} L^\dagger \sum_{r=1}^n (I + L L^\dagger)^{-r} L \ket{\psi} \\
        &\leq \frac{1}{k!} \bra{\psi} \sum_{r=1}^\infty L^\dagger (I + L L^\dagger)^{-r} L \ket{\psi} \\
        &\leq \frac{1}{k!} \bra{\psi} P \ket{\psi} \leq \frac{1}{k!} \,.
    \end{align*} 
    In the first inequality, we applied \eqref{lem:bound-inner-sum-pseudoinverse}. In the third line, we used that $ (I + L L^\dagger)^{-r} $ is positive for all $ r \in \mathbb{N} $. In the final line, we used the definition of the pseudoinverse of $ L $ from Lemma \ref{lem:pseudoinverse-series} and the fact that $ L^+ L = P $, where $ P $ is the projection onto the closure of $ \operatorname{ran}(L^+) $ (see \cite[Eq. (26)]{ben1963contributions}).
\end{proof}

\section{Discussion on possible experimental realizations}
\label{app:experimental_realizations}

The required dissipation $L_{r,\alpha}=b^{r}\left(b-\alpha\right)$ can be obtained by engineering a coupling Hamiltonian $H=g a^{\dagger}b^{r}\left(b-\alpha\right)+h.c.$
between the probe and a fast decaying ancilla.
Here, we discuss possible platforms and protocols to obtain this Hamiltonian.
We study the case of $r=1,$ i.e. 
$L_{1,\alpha}=b\left(b-\alpha\right),$
where a more general and detailed discussion is left for future work.

The standard two-photon dissipation used for cat code is typically obtained by coupling the two resonators, the probe (referred to as the storage mode) and the ancilla (buffer mode),
to asymmetrically threaded SQUID (ATS).
The ATS consists of two Josephson junctions with energies $E_{j_{1}},E_{j_{2}}$ that are connected in parallel and split in the middle by a linear inductor with inductance $L_b$ which forms two loops (see Refs \cite{lescanne2020exponential, chamberland2022building, Guillaud.2023} for more details and illustrations).
The loops have magnetic fluxes of $\varphi_{\text{ext},1},$ $\varphi_{\text{ext},2}.$
The potential energy of this system is given by:
\begin{align}
\begin{split}
& U_{\varphi}=\frac{1}{2}E_{L,b}\varphi^{2}-2E_{j_{1}}\cos\left(\varphi_{\text{ext},1}+\varphi_{1}\right)-2E_{j_{2}}\cos\left(\varphi_{\text{ext},2}+\varphi_{2}\right)\\
&=\frac{1}{2}E_{L,b}\varphi^{2}-2E_{J}\cos\left(\varphi_{\Sigma}\right)\cos\left(\varphi+\varphi_{\Delta}\right)-2\Delta E_{J}\sin\left(\varphi_{\Sigma}\right)\sin\left(\varphi+\varphi_{\Delta}\right),
\label{eq:potential_energy_ats}
\end{split}
\end{align}
where $\varphi$ is the operator of the phase difference across the ATS,
$E_{J}=\frac{E_{j_{1}}+E_{j_{2}}}{2},\Delta E_{J}=\frac{E_{j_{1}}-E_{j_{2}}}{2},$
and $\varphi_\Sigma=\frac{\varphi_1+\varphi_2}{2},$ $\varphi_\Delta=\frac{\varphi_1-\varphi_2}{2}.$
Taking $\varphi_{\Delta}=\pi/2,$ $\varphi_{\Sigma}=\pi/2+\epsilon\left(t\right)$ where $ \epsilon\left(t\right)=\epsilon_{0}\cos\left(\omega_{p}t\right)$:
\begin{align*}
U_{\varphi}=\frac{1}{2}E_{L,b}\varphi^{2}+2E_{J}\epsilon_{1}\left(t\right)\cos\left(\varphi+\varphi_{\Delta}\right)-2\Delta E_{J}\sin\left(\varphi+\varphi_{\Delta}\right).    
\end{align*}
Taking $\Delta E_{J}=0$ we get $U_{\varphi}=\frac{1}{2}E_{L,b}\varphi^{2}-2E_{J}\epsilon\left(t\right)\sin\left(\varphi\right),$ the full Hamiltonian is then
\begin{align*}
 H=\omega_{a}a^{\dagger}a+\omega_{b}b^{\dagger}b-2E_{J}\epsilon\left(t\right)\sin\left(\varphi_{a}\left(a+a^{\dagger}\right)+\varphi_{b}\left(b+b^{\dagger}\right)\right)   
\end{align*}
Expanding the sine up to third order yields
\begin{align*}
&H=\omega_{a}a^{\dagger}a+\omega_{b}b^{\dagger}b-2E_{J}\epsilon\left(t\right)\left[\varphi_{a}\left(a+a^{\dagger}\right)+\varphi_{b}\left(b+b^{\dagger}\right)\right]\\
&+\frac{E_{J}}{3}\epsilon\left(t\right)\left[\varphi_{a}\left(a+a^{\dagger}\right)+\varphi_{b}\left(b+b^{\dagger}\right)\right]^{3}.
\end{align*}
In the standard two-photon dissipation a drive of $\left(\epsilon_{d}e^{-i\omega_{d}t}+\epsilon_{d}^{*}e^{i\omega_{d}t}\right)\left(a+a^{\dagger}\right)$ is introduced, in our case we need to introduce instead a time-dependent “beam-splitter” interaction between the modes $\left(\epsilon_{d}e^{-i\omega_{d}t}+\epsilon_{d}^{*}e^{i\omega_{d}t}\right)\left(a^{\dagger}b+h.c.\right).$
Moving to the interaction frame with respect to $\omega_{a}a^{\dagger}a+\omega_{b}b^{\dagger}b$ we get
\begin{align*}
&H_{\text{rot}}=-2E_{J}\epsilon\left(t\right)\left[\varphi_{a}\left(ae^{-i\omega_{a}t}+a^{\dagger}e^{i\omega_{a}t}\right)+\varphi_{b}\left(be^{-i\omega_{b}t}+b^{\dagger}e^{i\omega_{b}t}\right)\right]+\\
&\frac{E_{J}}{3}\epsilon\left(t\right)\left[\varphi_{a}\left(ae^{-i\omega_{a}t}+a^{\dagger}e^{i\omega_{a}t}\right)+\varphi_{b}\left(be^{-i\omega_{b}t}+b^{\dagger}e^{i\omega_{b}t}\right)\right]^{3}+\\
&\left(\epsilon_{d}e^{-i\omega_{d}t}+\epsilon_{d}^{*}e^{i\omega_{d}t}\right)\left(a^{\dagger}be^{i\left(\omega_{a}-\omega_{b}\right)t}+h.c.\right).    
\end{align*}
By taking $\omega_{p}=2\omega_{b}-\omega_{a}$ we get that the $b^{\dagger2}a+h.c.$ are non-rotating, and by taking $\omega_{d}=\omega_{b}-\omega_{a}$ we obtain a non-rotating $a^{\dagger}b+h.c.$.
Given that $\omega_{b}\neq\omega_{a}$ it can be seen that all other terms are fast rotating,
hence the rotating wave approximation leaves us with 
\begin{align*}
H_{\text{RWA}}=g_{2}\left(b^{2}a^{\dagger}+h.c.\right)+g_{1}\left(ba^{\dagger}+h.c.\right)=g_{2}\left[ b\left(b+\frac{g_{1}}{g_{2}}\right)a^{\dagger}+h.c. \right]  
\end{align*}
where $g_{2}=\frac{E_{J}\epsilon_{0}}{2}\varphi_{b}^{2}\varphi_{a}$,  $g_{1}=\epsilon_{d}\varphi_{b}\varphi_{a}$.

This protocol requires replacing the drive $\left(\epsilon_{d}e^{-i\omega_{d}t}+\epsilon_{d}^{*}e^{i\omega_{d}t}\right)\left(a+a^{\dagger}\right)$
with a time dependent coupling $\left(\epsilon_{d}e^{-i\omega_{d}t}+\epsilon_{d}^{*}e^{i\omega_{d}t}\right)\left(a^{\dagger}b+h.c.\right)$ which may not be feasible experimentally. We therefore propose another approach with this platform that may be more realistic.
In \eqref{eq:potential_energy_ats} we can choose the following parameters $\varphi_{\Delta}=0,$ $\varphi_{\Sigma}=\pi/4+\epsilon\left(t\right), \epsilon\left(t\right)\ll1$
and also $\Delta E_J >0.$ The potential energy, in leading order of $\epsilon \left(t\right),$ is then:
\begin{align*}
U_{\varphi}=\frac{1}{2}E_{L,b}\varphi^{2}-2E_{J}\frac{1}{\sqrt{2}}\left(1-\epsilon\left(t\right)\right)\cos\left(\varphi\right)-2\Delta E_{J}\frac{1}{\sqrt{2}}\left(1+\epsilon\left(t\right)\right)\sin\left(\varphi\right),    
\end{align*}
and the Hamiltonian is thus:
\begin{align*}
H&=\omega_{a}a^{\dagger}a+\omega_{b}b^{\dagger}b-\sqrt{2}\left(E_{J}\cos\left(\varphi\right)+\Delta E_{J}\sin\left(\varphi\right)\right)+\sqrt{2}\epsilon\left(t\right)\left(E_{J}\cos\left(\varphi\right)-\Delta E_{J}\sin\left(\varphi\right)\right)\\
&=H_{0}+H_{1}+H_{\epsilon},
\end{align*}
where $H_0=\omega_{a}a^{\dagger}a+\omega_{b}b^{\dagger}b,$
$H_1=-\sqrt{2}\left(E_{J}\cos\left(\varphi\right)+\Delta E_{J}\sin\left(\varphi\right)\right),$
and $H_\epsilon=\sqrt{2}\epsilon\left(t\right)\left(E_{J}\cos\left(\varphi\right)-\Delta E_{J}\sin\left(\varphi\right)\right).$
Let us focus on $H_{\epsilon}$ and expand it up to the third order:
\begin{align*}
&H_{\epsilon}=-\sqrt{2}\Delta E_{J}\epsilon\left(t\right)\left[\varphi_{a}\left(a+a^{\dagger}\right)+\varphi_{b}\left(b+b^{\dagger}\right)-\frac{1}{6}\left[\varphi_{a}\left(a+a^{\dagger}\right)+\varphi_{b}\left(b+b^{\dagger}\right)\right]^{3}\right]\\
&-\sqrt{2}\epsilon\left(t\right)E_{J}\frac{\left(\varphi_{a}\left(a+a^{\dagger}\right)+\varphi_{b}\left(b+b^{\dagger}\right)\right)^{2}}{2}.    
\end{align*}
Note that the since $\cos\left(\varphi\right)$ contains $a^{\dagger}a$ and $b^{\dagger}b$ terms the frequencies $\omega_{a},\omega_{b}$ are changed to $\bar{\omega}_{a}=\omega_{a}-\sqrt{2}E_{J}\varphi_{a}^{2}$, $\bar{\omega}_{b}=\omega_{b}-\sqrt{2}E_{J}\varphi_{b}^{2}$.
Moving to the interaction frame with respect to $\bar{\omega}_{a}a^{\dagger}a+\bar{\omega}_{b}b^{\dagger}b$ we get that $\epsilon\left(t\right)$ is resonant with $\left(a^{2}b^{\dagger}e^{-i\left(2\bar{\omega}_{a}-\bar{\omega}_{b}\right)t}+h.c.\right)$ if $\epsilon\left(t\right)$ has a frequency of $2\bar{\omega}_{a}-\bar{\omega}_{b}.$
In addition, $\epsilon\left(t\right)$ is resonant with $\left(ab^{\dagger}e^{-i\left(\bar{\omega}_{a}-\bar{\omega}_{b}\right)t}+h.c.\right)$ if $\epsilon\left(t\right)$ has a frequency of $\bar{\omega}_{a}-\bar{\omega}_{b}.$
We can therefore use $\epsilon\left(t\right)$ that has two frequencies: $\epsilon\left(t\right)=\epsilon_{1}\cos\left(\omega_{p,1}t\right)+\epsilon_{2}\cos\left(\omega_{p,2}t\right)$ with $\omega_{p,1}=2\bar{\omega}_{a}-\bar{\omega}_{b}$, $\omega_{p,2}=\bar{\omega}_{a}-\bar{\omega}_{b}.$
It should be checked now that all other terms are fast rotating and can be therefore neglected.
In $H_{\epsilon}$ all other terms rotate as $\bar{\omega}_{a/b},2\bar{\omega}_{a/b},3\bar{\omega}_{a/b}$ hence $\epsilon\left(t\right)$ is not resonant with any of these terms. $H_{1}$ does not contain any $\epsilon\left(t\right)$ hence all its terms are fast rotating. We therefore obtain $H_{\text{RWA}}=g_{2}\left(b^{2}a^{\dagger}+h.c.\right)+g_{1}\left(ba^{\dagger}+h.c.\right)$ with $g_{2}=\Delta E_{J}\epsilon_{1}\frac{\varphi_{b}^{2}\varphi_{a}}{\sqrt{2}}$, 
$g_{1}=\sqrt{2}E_{J}\epsilon_{2}\varphi_{a}\varphi_{b}$.

For $r>1$ higher-orders of $\sin\left(\varphi\right), \cos\left(\varphi\right)$ should be used with suitable frequency matching conditions.

\end{document}